\documentclass[journal]{IEEEtran}
\ifCLASSINFOpdf
\else
\fi

\usepackage{glossaries}
\usepackage{acronym}
\usepackage{amsthm}
\usepackage{amsmath}
\usepackage{amsfonts}
\usepackage{mathtools}
\usepackage{amssymb}
\usepackage{pgfplots}
\usepackage{graphics}
\usepackage{graphicx}
\usepackage{epstopdf}
\usepackage{algorithm}
\usepackage{algorithmicx}
\usepackage{algorithm,algpseudocode}
\usepackage{bm}
\usepackage{bbm}
\usepackage{ragged2e}
\usepackage{subfig}
\usepackage{xr}
\usepackage{xfrac}
\usepackage{fix-cm}

\newacronym{5G}{5G}{Fifth-Generation}
\newacronym{4G}{4G}{Fourth-Generation}
\newacronym{eMBB}{eMBB}{Enhanced Mobile BroadBand}
\newacronym{mMTC}{mMTC}{Massive Machine Type Communications}
\newacronym{URLLC}{URLLC}{Ultra-Reliable Low Latency Communications}
\newacronym{NR}{NR}{New Radio}
\newacronym{MIMO}{MIMO}{Multiple Input Multiple Output}
\newacronym{LTE}{LTE}{Long Term Evolution}
\newacronym{D2D}{D2D}{Device-to-Device}
\newacronym{M2M}{M2M}{Machine-to-Machine}
\newacronym{V2V}{V2V}{Vehicle-to-Vehicle}
\newacronym{V2X}{V2X}{Vehicle-to-Everything}
\newacronym{ProSe}{ProSe}{Proximity Services}
\newacronym{PS}{PS}{Public Safety}
\newacronym{IoT}{IoT}{Internet of Things}
\newacronym{CSI}{CSI}{Channel State Information}
\newacronym{HARQ}{HARQ}{Hybrid Automatic Repeat Request}
\newacronym{AMC}{AMC}{Adaptive Modulation Coding}
\newacronym{TDD}{TDD}{Time Division Duplex}
\newacronym{TDMA}{TDMA}{Time Division Multiple Access}
\newacronym{QoS}{QoS}{Quality of Service}
\newacronym{FDD}{FDD}{Frequency Division Duplex}
\newacronym{DTMC}{DTMC}{Discrete Time Markov Chain}
\newacronym{DRX}{DRX}{Discontinuous Reception}
\newacronym{DL}{DL}{Downlink}
\newacronym{UL}{UL}{Uplink}
\newacronym{3GPP}{3GPP}{3rd Generation Partnership Project}
\newacronym{SNR}{SNR}{Signal-to-Noise Ratio}
\newacronym{SINR}{SINR}{Signal-to-Interference-plus-Noise Ratio}
\newacronym{OP}{OP}{Optimization Problem}
\newacronym{RB}{RB}{Resource Block}
\newacronym{RE}{RE}{Resource Emplacement}
\newacronym{PUCCH}{PUCCH}{Physical Uplink Control Channel}
\newacronym{BS}{BS}{Base Station}
\newacronym{UB}{UB}{Upper Bound}
\newacronym{UE}{UE}{User Equipment}
\newacronym{RI}{RI}{Rank Index}
\newacronym{PMI}{PMI}{Precoding Matrix Indicator}
\newacronym{CQI}{CQI}{Channel Quality Index}
\newacronym{DCI}{DCI}{Downlink Control Information}
\newacronym{AWGN}{AWGN}{Additive White Gaussian Noise}
\newacronym{QPSK}{QPSK}{Quadrature Phase-Shift Keying}
\newacronym{PDCCH}{PDCCH}{Physical Downlink Control Channel}
\newacronym{PUSCH}{PUSCH}{Physical Uplink Shared Channel} 
\newacronym{RRC}{RRC}{Radio Resource Control} 
\newacronym{E-UTRA}{E-UTRA}{Evolved Universal Mobile Telecommunications System Terrestrial Radio Access,}  
\newacronym{MANET}{MANET}{Mobile Ad-Hoc Networks}  
\newacronym{OFDMA}{OFDMA}{Orthogonal Frequency Division Multiple Access}  
\newacronym{SC-FDMA}{SC-FDMA}{Single Carrier Frequency Division Multiple Access}  
\newacronym{EC}{EC}{Energy Consumption}  
\newacronym{EE}{EE}{Energy Efficiency} 
\newacronym{MU}{MU}{Master User Equipment}
\newacronym{TTI}{TTI}{Time Transmission Interval}
\newacronym{ISD}{ISD}{Inter-Site Distance}
\newacronym{RF}{RF}{Radio Frequency}
\newacronym{UAV}{UAV}{unmanned aerial vehicle}
\newacronym{CDF}{CDF}{Cumulative Distribution Function}
\newacronym{RMAB}{RMAB}{Restless Markov Multi-armed Bandit}
\newacronym{MDP}{MDP}{Markov Decision Process}
\newacronym{CMDP}{CMDP}{Constrained Markov Decision Process}
\newacronym{POMDP}{POMDP}{Partially Observable Markov Decision Process}
\newacronym{CPOMDP}{CPOMDP}{Constrained Partially Observable Markov Decision Process}
\newacronym{PWLC}{PWLC}{Piecewise-Linear Convex} 
\newacronym{MILP}{MILP}{Mixed-Integer Linear Program}
\newacronym{PBVI}{PBVI}{Point-Based Value Iteration}
\newacronym{CPBVI}{CPBVI}{Constrained Point-Based Value Iteration}
\newacronym{GCPBVI}{GCPBVI}{Greedy Constrained Point-Based Value Iteration}

\newtheorem{theorem}{Theorem}[section]
\newtheorem{lemma}[theorem]{Lemma}

\newtheorem{example}{Example}

\newtheorem{definition}[theorem]{Definition}
\DeclareGraphicsExtensions{.pdf,.eps}
\graphicspath{{.//}}

\algnewcommand\algorithmicswitch{\textbf{switch}}
\algnewcommand\algorithmiccase{\textbf{case}}
\algnewcommand\algorithmicassert{\texttt{assert}}
\algnewcommand\Assert[1]{\State \algorithmicassert(#1)}%
\algdef{SE}[SWITCH]{Switch}{EndSwitch}[1]{\algorithmicswitch\ #1\ \algorithmicdo}{\algorithmicend\ \algorithmicswitch}%
\algdef{SE}[CASE]{Case}{EndCase}[1]{\algorithmiccase\ #1}{\algorithmicend\ \algorithmiccase}%
\algtext*{EndSwitch}%
\algtext*{EndCase}%

\newcommand{\argmax}{\mathop{\mathrm{argmax}}}

\DeclarePairedDelimiter\ceil{\lceil}{\rceil}

\newcommand{\bigO}[1]{\ensuremath{\mathop{}\mathopen{}O\mathopen{}\left(#1\right)}}

\begin{document}
%

\title{A Dynamic and Incentive Policy for Selecting \gls{D2D} Mobile Relays}
%
%

\author{Rita Ibrahim,~\IEEEmembership{Member,~IEEE,}
        Mohamad Assaad,~\IEEEmembership{Senior Member,~IEEE,}
        and~Berna Sayrac,~\IEEEmembership{Member,~IEEE,}}
\maketitle

\begin{abstract}
User-to-network relaying enabled via \acrfull{D2D} communications is a promising technique for improving the performance of cellular networks. Since in practice relays are in mobility, a dynamic relay selection scheme is unavoidable. In this paper, we propose a dynamic relay selection policy that maximizes the performance of cellular networks (e.g. throughput, reliability, coverage) under cost constraints (e.g. transmission power, power budget). We represent the relays' dynamics as a \acrfull{MDP} and assume that only the locations of the selected relays are observable. Therefore, the dynamic relay selection process is modeled as a \acrfull{CPOMDP}. Since the exact solution of such framework is intractable to find, we develop a point-based value iteration solution and evaluate its performance. In addition, we prove the submodularity property of both the reward and cost value functions and deduce a greedy solution which is scalable with the number of discovered relays. For the muti-user scenario, a distributed approach is introduced in order to reduce the complexity and the overhead of the proposed solution. We illustrate the numerical results of the scenario where throughput is maximized under energy constraint and evaluate the gain that the proposed relay selection policy achieves compared to a traditional cellular network.
\end{abstract}

%
\IEEEpeerreviewmaketitle

\section{Introduction}
\IEEEPARstart{I}{n} \gls{D2D} enabled cellular networks, user to network relaying can be handled for improving the performance of cellular networks. Under realistic assumption that relays are in mobility, it is crucial to define a strategy for designating the relays that will respectively serve each \acrfull{MU} in the network. In this paper, we propose a relay selection policy that maximizes the performance of cellular networks (e.g. throughput, reliability, coverage) under cost constraints (e.g. transmission power, power budget). We assume that the relays' dynamics are represented by a \acrfull{MDP} where the \gls{MU} cannot directly observe the locations of all the potential relays that have been discovered. Therefore, the sequential relay decision process is modeled by a \acrfull{CPOMDP}. Since the exact solutions of such framework are computationally intractable to find, we have developed an approximated solution as well as discussed the existing trade-off between its complexity and its preciseness. Moreover, proving the submodularity property of the reward and cost functions leads us to propose a greedy form of the approximated solution. Numerical results are presented to endorse our relay selection policies and to show how introducing \gls{D2D} relaying can highly improve the performance of cellular networks. Furthermore, a system-level simulator is developed in order to implement our strategy of relay selection and to test its performance in a nearly realistic cellular network.  

\section{Concept and related work \label{sec:RS_Intro}}
In traditional cellular networks, single hop communications are deployed between the users and the \gls{BS}. However, introducing relays to cellular networks has become one of the major concern of cellular network planners that aim to improve the capacity and the coverage of their networks. The emergence of \gls{D2D} communications encourages the deployment of \gls{UE}-to-Network relaying functionality. \gls{D2D} communications will take place between the \gls{MU} and the relays however cellular communications are maintained between the relays and the \gls{BS}. The main advantages of such relaying feature is: (i) improving the performance of the network (e.g. capacity enhancement, coverage extension, transmission power reduction, load balancing, network offloading etc.) and enabling new services (e.g. data on demand).

\gls{D2D} relaying technique has been the subject of interest of significant research in both academia and industry. Several tools have been used for evaluating the performance of this technique. Stochastic geometry is used for analytically modeling and analyzing the performance of cellular network with fix \gls{D2D} relays in \cite{Chen2014Stochastic} and mobile \gls{D2D} relays in \cite{Liu2015Device}. Monte-Carlo simulations in \cite{Vanganuru2012System} and \cite{Huang2009Capacity} show how enabling \gls{D2D} communications to carry relayed traffic can enhance the capacity and coverage of cellular networks. A system-level simulator in \cite{Babun2015MultiHop} was developed to evaluate the extension of the cellular coverage due to \gls{D2D} relaying.

Despite the performance gain that \gls{UE}-to-Network relaying has promised, several challenging issues require further investigation. 
One can ask how the relays should be strategically positioned in order to optimize the performance of the network.Sharing the spectrum between \gls{D2D} and cellular communciations is one of the existing challenges. Using stochastic geometry modeling, authors in \cite{Atat2017Energy} studied underlay \gls{D2D} enabled cellular networks. They analytically derived the tradeoff generated by the spectrum partition between \gls{D2D} and cellular communications and they deduce the optimal spectrum partition that guarantees the fairness in the network. Several solutions of power and/or resource allocation in \gls{D2D} relayed cellular networks have been studied in the recent literature. The work \cite{Hasan2015Distributed}, that proposes a distributed resource allocation for \gls{D2D} relay-aided cellular networks using game theory tool, exposes a summary of the different existing centralized and distributed resource allocation schemes. In addition to resource and power allocation, a mode and path selection algorithm was developed and simulated in \cite{Silva2014Performance}. Furthermore, \gls{CMDP} problems were formulated in \cite{Niyato2009Optimization} and \cite{Zhang2016Energy} to obtain the optimal decision of packet scheduling that mobile relays should take.

One of the main challenges of \gls{D2D} relay-aided cellular networks is to decide how relays should be selected in order to achieve the  performance of cooperative relaying. Due to relay mobility, a dynamic relay selection policy is unavoidable. Indeed, a relay selected at one position can be no longer helpful at another position. In this work, we address this question by proposing a dynamic relay selection strategy that \gls{MU} may take for optimizing a certain performance metric (e.g. throughput, coverage, reliability etc.). Since the relaying functionality is costly for the relays (e.g. in terms of energy, data consumption etc), we propose to charge \gls{MU}s for using their selected relays. This cost aims to encourage the users to behave as potential relays. Based on the fact that the reward as well as the cost of each relay depends on its location, we propose a dynamic relay selection policy that maximizes some \gls{QoS} metric of cellular networks while satisfying some cost constraints.    

\subsection{Related Work} 
Relay selection is crucial for improving the performance of \gls{D2D} relay-aided cellular networks. The rich literature on relay selection problems can be divided into two categories: mobile relay selection and fix relay selection schemes. Most of the previous works considers fix relay selection scenarios. Some of these schemes are briefly presented in the following. Based on a stochastic formulation, authors in \cite{Wei2010Distributed} propose a fully distributed single relay association scheme that aims to increase the spectral efficiency of the network. An energy efficient relay selection was proposed in \cite{Seonghwa2013Energy} based on a \gls{DTMC} modeling of the relay node with \gls{DRX} mechanism. Based on an iterative technique, \cite{Kim2014Iterative} proposes a joint relay selection and power allocation problem scheme for relay-aided \gls{D2D} underlying cellular networks. The work in \cite{Ma2017RelaySelction} uses a queuing theory model to propose a single relay selection scheme that optimizes the network in terms of relay remaining battery life, end-to-end data rate and end-to-end delay criteria. For underlay \gls{D2D} enabled cellular networks, interference is mitigated between cellular and \gls{D2D}communications by considering a distributed relay selection algorithm in \cite{Ma2012distributed}.  

However, the aforementioned works do not take user mobility into account but consider fix relays which limits their applicability in cellular mobile networks. Considering the mobility of the relays seems to be a challenging scenario. Mobile relays have been the subject of the work in \cite{Li2011dynamic}. Based on the knowledge of the relay mobility pattern, a dynamic relay selection scheme aiming to minimize the cost of relaying under \gls{QoS} requirements was proposed. The studied optimization model is based on \gls{CMDP}. Since relay selection appears as decision making process, several \gls{MDP} based formulation has been studied in order to propose optimal relay selection scheme. For example, while predicting the channel states of the available relays, authors in \cite{LiCooperative2011} address a relay selection policy that maximizes the long term transmission rate. This decision strategy is obtained by solving a \gls{POMDP} with dynamic programming-based algorithm.

\subsection{Contribution and Organization}
The use cases of \gls{D2D} communications that will be studied in this paper is user to network relaying in cellular networks. The mobility of the relays is the main challenge for relay selection decision. We propose a dynamic relay selection policy that maximizes a certain performance metric of the network (e.g. throughput, reliability, coverage etc.) under cost constraints (e.g. energy consumption, data consumption etc.). The main contribution in this work can be summarized as follows: 
\begin{itemize}
\item \gls{CPOMDP} formulation of the sequential relay selection process that a \gls{MU} will decide and discussion concerning the complexity of this problem.
\item Proposition of a dynamic policy of relay selection that approximately optimizes the formulated problem based on a greedy point based value iteration. 
\item Extending the results to multiple users scenario. 
\item Numerical results as well as system-level simulations show the performance that cellular networks may gain by implementing the proposed relay selection policy.
\end{itemize}

The main particularity of this work compared to other relay selection schemes previously proposed in the literature is the following:
\begin{itemize}
\item The consideration of a realistic scenario where relays are in mobility, thus dynamic relay selection scheme is necessary.
\item The cellular network performance is optimized under cost constraints.
\item The sequential relay decision process is modeled by a \acrfull{CPOMDP} (and not a \gls{CMDP}) since the states of the potential relays cannot be observed until these relays are selected by the \gls{MU}. 
\item An approximated solution is proposed for avoiding the intractability of exact solutions. A trade-off between the preciseness of the approximated solution and its complexity is derived.
\item The relay selection is not limited to one relay, thus the performance gain is improved due to the increasing in the cooperative diversity order. Actually, the proposed relay selection policy aims to maximize the performance of the network under cost constraints without any constraints on the number of chosen relays. 
\item The relays' mobility pattern is assumed communicated to the agent decision maker ( \gls{MU} or the \gls{BS}).
\end{itemize}

The rest of this paper is organized as follows. Section \ref{sec:RS_SystemModel} describes the system model for a single \gls{MU} scenario. Section \ref{sec:RS_ProbForm} formulates the optimization problem as a \gls{CPOMDP}. Since the exact solutions of such problem are intractable to find, a low-complexity dynamic relay selection, called \gls{CPBVI}, is proposed in section \ref{sec:RS_Policies}. The submodularity property is verified for this problem, thus a greedy form of this approximation, called \gls{GCPBVI}, is deduced. These results are extended to a multi-\gls{MU} scenario in section \ref{sec:RS_MultiPlayer}. For this scenario, a distributed approach is exposed for reducing the complexity of centralized solutions. Numerical results in section \ref{sec:RS_NumResults} corroborate our claims. The performance enhancement of the cellular networks is shown in section \ref{sec:RS_SimuResults} by implementing the proposed relay selection scheme in a system-level simulator. Section \ref{sec:RS_Conclusion} concludes the paper whereas the proofs are provided in the appendices. 
\section{System Model \label{sec:RS_SystemModel}}
For the sake of clarity, we describe the system model for a single \gls{MU} scenario. We show that extending the result to multi-\gls{MU} scenario is a straight forward process.
\subsection{Network model}
We consider a single cell scenario with one \gls{MU} and a set of $K$ potential mobile relay. \gls{MU} is allowed to use \gls{D2D} communications to access the network via mobile relays and by that improving the performance of its cellular communications. As figure \ref{fig.RS_NR_scenario_SUE} shows, the \gls{MU} is allowed to use \gls{D2D} communications to access the network via mobile relays.  \gls{MU} discovers $K$ nearby potential relays by launching a discovery process of period $T$. The time between two discovery processes is partitioned into decision epochs $t$ (e.g. \gls{TTI} in \gls{LTE} networks) of constant duration (with $t \in \lbrace{ 1, 2,..., T \rbrace}$). The goal of this work is to determine the relay selection policy that should be applied by the \gls{MU} at each epoch $t$ in order to maximize its cumulative reward under cost constraints.

$\mathcal{K}=\lbrace{0, 1, 2,..., K \rbrace}$ denotes the set of $K$ existing potential relays (from index $1$ to $K$) as well as the direct cellular link (index $0$). The relay selection policy consists in deciding whether the \gls{MU} will have direct communication with the network or will pass by some mobile relays. In the latter case, this policy chooses the subset of relays that will be used for attaining the network. Please note that \gls{D2D} relaying can be applied to both \gls{DL} and \gls{UL} communications. Due to the practical consideration that mobile terminals do not support simultaneous signal transmission and reception, a two-phases transmission scheme is assumed. Therefore, for \gls{UL} (resp. \gls{DL}) relayed communication, the transmission protocol is divided into two phases: (i) in the first phase the relay receives the data from the \gls{MU} (resp. \gls{BS}) and (ii) in the second phase the relay transmits the received data to the \gls{BS} (resp. \gls{MU}). 

\begin{figure}
\centering
\includegraphics[scale=0.7]{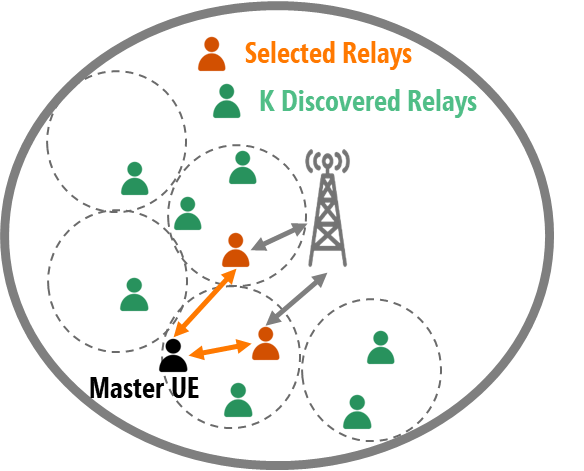}
\caption{Single-\gls{UE} system model}
\label{fig.RS_NR_scenario_SUE} 
\end{figure}

\subsection{Mobility Model}
Relays' locations in a service coverage area are quantized and represented by a set of regions $\mathcal{S} = \lbrace S_1,S_2, . . . ,S_{|\mathcal{S}|}\rbrace$. We assume that the relays remain in the same region during a decision epoch. We denote by $s_i\left( t\right)$ the location of relay $i$ at epoch $t$. In the next epoch, the location of each relay is changed (i.e. by either staying in the same region or moving to another neighboring region). Thus, the mobility of relay $i$ is modeled by a transition matrix $\bm{P}_i$ where each element $P_i\left(S_n,S_{n'}\right)$ of this matrix denotes the probability that relay $i$ moves from region $S_n \in \mathcal{S}$ to region $S_{n'} \in \mathcal{S}$ in the next decision epoch. In this work, we assume that the \gls{MU} is aware of the mobility pattern of each discovered relay (i.e. equivalent to the relay's transition matrix). The vector $\bm{s}_t=\left( s_1\left(t\right),s_2\left(t\right),...,s_K\left(t\right) \right)$ denotes the location of the $K$ potential relays at epoch $t$. 

We underline that the location of each relay is changing over time epochs based on its transition probability matrix and independently from the fact that this relay has been or not selected. We assume that the \gls{MU} will be able to know the current localization region of only the chosen relays. On the other hand, the \gls{MU} will just have a certain belief of the locations of the non selected relays. It is a practical assumption because when a relay is selected by the \gls{MU} some signaling is exchanged between these two equipments in order to allow \gls{D2D} communications. This exchange may include some information that indicates the region of the selected relay. However, when the relay is not selected then no need for any signaling between the \gls{MU} and the relay since no communication will take place between these two equipments. In this case, recognizing the region of the relays requires an overhead of signaling which can be avoided.

\subsection{Cost and Reward Model}
We propose a relay selection policy that aims to optimize a cellular reward under cost constraints. Here, cost and reward models are described. The reward function, which depends on the relay's location, represents the benefit of choosing a relay in terms of throughput and/or reliability and/or coverage etc. The cost function, which depends on the relay's location, defines the charge that the \gls{MU} owes to each selected relay (i.e. in terms of energy and/or incentive budget etc.). We respectively denote by $r_i\left(s_i\right)$ and $c_i\left(s_i\right)$ the reward and the cost of the $i^{th}$ relay in location $s_i \in \mathcal{S}$. We assume that when the \gls{MU} chooses a set of relays then the total reward (resp. cost) is the sum of the reward (resp. cost) of each relay in the chosen set. We denote by $\bm{a}=\left(a_1,a_2,...,a_K\right)$ the vector such that $a_i=1$ if relay $i$ is selected and $0$ otherwise. Thus, the total reward and cost at a given state $\bm{s}$ are given by:
\begin{equation}
\label{eq.TotalReward}
R\left(\bm{\bm{s},\bm{a}}\right)=\sum \limits_{i \in {\mathcal{K}}}r_i\left(s_i\right)\mathbbm{1}_{\lbrace a_i=1\rbrace}
\end{equation}
\begin{equation}
\label{eq.TotalCost}
C\left(\bm{\bm{s},\bm{a}}\right)=\sum \limits_{i \in {\mathcal{K}}}c_i\left(s_i\right)\mathbbm{1}_{\lbrace a_i=1\rbrace}
\end{equation}
\section{Problem formulation \label{sec:RS_ProbForm}}
Enabling user to network relaying functionality based on \gls{D2D} communications leads to the following question: which mobile relays should be chosen by the \gls{MU} for ensuring an enhancement in cellular communications. This decision depends essentially on the reward and cost parameters of the discovered relays. The main challenge of such relay selection policy remains in the mobility of these potential relays. The goal is to make the relay selection decision ($\bm{a}_0,...,\bm{a}_T$) that optimize the following:
\begin{equation}
\label{eq.Prob1}
\max \mathbb{E}\left[ \sum \limits _ {t=1}^T \gamma^{t} R\left(\bm{s}_{t} ,\bm{a}_{t}\right) \right] \text{ s.t. } \mathbb{E}\left[ \sum \limits _ {t=1}^T\gamma^{t} C\left(\bm{s}_{t} ,\bm{a}_{t}\right) \right] \leq C_{th} 
\end{equation}
where $\gamma \in \left[ 0,1 \right]$ is a discount factor that represents the difference in importance between future rewards (resp. costs) and present rewards (resp. costs).

The considered reward model includes a large scope of reward metrics. In the following we give few examples that indicate how the reward mode \ref{eq.TotalReward} can be applied:
\begin{itemize}
\item Throughput criteria: different packets are transmitted to the selected relays, thus the total throughput is the sum of the throughput of each selected link. Considering the example of Shannon capacity over a bandwidth W, the total reward is given by $R\left(\bm{\bm{s},\bm{a}}\right)=$
\[\resizebox{0.9\hsize}{!}{$\hspace{-20pt}W\sum \limits_{i \in \mathcal{K}| a_i=1} \log_2\left[1+\min \left\{ \text{SINR}_{\text{\gls{MU}}-\text{Relay}_i}\left(s_i\right), \text{SINR}_{\text{Relay}_i-\text{\gls{BS}}}\left(s_i\right) \right\} \right]$}
\]
\item Reliability criteria: same packets are transmitted to the selected relays,  thus the error probability $q$ is equal to the product of the error probability of each selected link $q_i$. For example, in the case of considering the bit error ratio in the case of \gls{QPSK} modulation and \gls{AWGN} channel, the error probability of relay $i$ at state $s_i$ can be written as follows: 
\[\resizebox{0.9\hsize}{!}{$\hspace{-20pt}
q_i\left(s_i\right)=\frac{1}{2}erfc\left(\sqrt{ \text{SINR}_{\text{\gls{MU}}-\text{Relay}_i}\left(s_i\right)}\right)\times\frac{1}{2}erfc\left(\sqrt{\text{SINR}_{\text{Relay}_i-\text{\gls{BS}}}\left(s_i\right)}\right)$}
\]
When action $\bm{a}$ is taken, then $q\left(\bm{a},\bm{s}\right)=\prod\limits_{i \in \mathcal{K} | a_i=1} q_i\left(s_i\right)$. Thus, minimizing this error probability corresponds to maximizing the $-\log$ of the probability error of each selected link 
\[\resizebox{0.9\hsize}{!}{$\hspace{-20pt}
\min\limits_{\bm{a}} q\left(\bm{a},\bm{s}\right)=\min\limits_{\bm{a}} \prod\limits_{i \in \mathcal{K} | a_i=1}q_i\left(s_i\right)=\max\limits_{\bm{a}}  -\sum\limits_{i \in \mathcal{K} | a_i=1}\log\left[q_i\left(s_i\right)\right]$}
\]
Therefore, we can express the reliability reward as $-\log$ of the error probability in order to have the overall reward as the sum of the reward of each selected relay $R\left(\bm{\bm{s},\bm{a}}\right)=$
\[\resizebox{0.9\hsize}{!}{$\hspace{-20pt}
-\sum\limits_{i \in \mathcal{K} | a_i=1}\log\left[ \frac{1}{4}erfc\left(\sqrt{ \text{SINR}_{\text{\gls{MU}}-\text{Relay}_i}\left(s_i\right)}\right)erfc\left(\sqrt{\text{SINR}_{\text{Relay}_i-\text{\gls{BS}}}\left(s_i\right)}\right) \right]$}
\] 
The process above can be applied to any other error probability function (e.g. \cite{Polyanskiy2010}). 
\item Coverage probability criteria: same packets are transmitted to the selected relays,  thus the overall outage probability is equal to the product of the outage probability of each selected link. Similarly to reliability criteria, we can express the coverage reward as $-\log$ of the outage probability in order to have the overall reward as the sum of the reward of each selected relay.
\end{itemize}

The considered cost model can be applied to a wide range of cost metrics. In the following, we give some examples to illustrate how the cost model given by \ref{eq.TotalCost} can be applied:
\begin{itemize}
\item Energy criteria: the total consumed energy is the sum of the energy consumed by each selected relay. For example, if we denote by $P_i\left( s_i\right)$ the transmission power of relay $i$ when it is in state $s_i$, thus:
\[
C\left(\bm{\bm{s},\bm{a}}\right)=\sum \limits_{i \in \mathcal{K}| a_i=1}P_i\left(s_i\right)
\]
\item Incentive criteria: the total charged cost is equal to the sum of the incentive budget required by each selected relay. For example, one can denote by $L_i\left( s_i\right)$ the tokens used by relay $i$ in state $s_i$, thus:
\[
C\left(\bm{\bm{s},\bm{a}}\right)=\sum \limits_{i \in \mathcal{K}| a_i=1}L_i\left(s_i\right)
\]
\end{itemize}
\subsection{\gls{RMAB} representation}
The \gls{MU} relay selection procedure consists of choosing a limited set of mobile relays that optimizes its subjective function under cost constraints. However, the positions of the potential relays are partially known at the epoch of decision and turn to be more observable as time passes. This problem can be modeled by a \gls{RMAB} due to the three following reasons:
\begin{itemize}
\item It is a multi-armed bandit problem because the set of potential relays represents the arms among which the \gls{MU} will make its selection. 
\item It is a Markov multi-armed bandit problem because the location of the $i^{th}$ arm is modeled by a discrete, irreducible and aperiodic Markov chain with finite space $\mathcal{S}$ and transition probability matrix $\bm{P}_i$ (i.e. relays' locations are the states of this \gls{RMAB})
\item  It is restless because the potential relays change their location/state from a decision epoch to another independently from the current decision (whether they've been selected or not). 
\end{itemize}
The objective of the relay selection strategy is to optimize problem \ref{eq.Prob1} through a sequence of selection $1 \leq t\leq T$.  For each selection decision, there exists a trade-off between exploitation phase that tends to achieve the highest expected reward and exploration phase that tends to get more information concerning the positions of other relays. In the sequel, we describe the \gls{CPOMDP} formulation equivalent to the proposed \gls{RMAB} problem.\\

\subsection{\gls{CPOMDP} formulation \label{ssec:CPOMDP}}
In order to study the \gls{RMAB} defined above, we formulate its equivalent finite-horizon \gls{CPOMDP} which is characterized by the tuple $\left\langle \mathcal{\bm{S}}, \mathcal{A}, T\left(\right), \mathcal{Z}, O\left(\right), R, C,\bm{b}_0,T, C_{th}, \gamma \right\rangle$ defined below:
\begin{itemize}
\item \textbf{State} $\bm{s}=\left( s_1, s_2, ...., s_K \right)$ denotes the state vector or the location vector of the $K$ potential relays (where $s_i \in \mathcal{S}$ for all $1 \leq i \leq K$). $\bm{\mathcal{S}}$ represents the set of all possible state vectors with $|\bm{\mathcal{S}}|=|\mathcal{S}|^K$.

\item \textbf{Action} $\bm{a}=\left( a_1, a_2, ...., a_K \right) \in \lbrace 0,1\rbrace^K$ denotes the vector of the $K$ binary actions such that $a_i \in \lbrace 0 , 1 \rbrace$ specifies whether relay $i$ is selected ($a_i=1$) or not ($a_i=0$) for all $1 \leq i \leq K$. $\hat{a}=\lbrace i:a_i=1\rbrace$ represents the set of the indexes of the selected relays. $\mathcal{A}$ denotes the set of all action vectors, it contains $|\mathcal{A}| = \sum \limits _{i=1}^{K}{{K}\choose{i}}=2^K$ elements.   

\item \textbf{Transition function }$T\left(\bm{s},\bm{s}'\right): \bm{\mathcal{S}}\times  \bm{\mathcal{S}} \rightarrow \left[0,1 \right]$ represents the probability of transiting between states. $T\left(\bm{s},\bm{s}'\right)$ characterizes the probability of passing to state $\bm{s}'$ in the next decision epoch knowing that the current state is $\bm{s}$. Assuming that the relays' mobility are independent and described by their transition probability matrix, then $T\left(\bm{s},\bm{s}'\right)=\prod\limits_{i=1}^K P_i\left(s_i,s_i'\right)$. This probability is independent from the action $\bm{a}$ since the selection decision has a purely observational role (i.e. the relays change their locations independently from the fact that they have been selected or not). $\bm{T}$ represents the transition matrix of $|\bm{\mathcal{S}}|\times |\bm{\mathcal{S}}|$ elements. 

\item \textbf{Observation} $\bm{z}=\left( z_1, z_2, ...., z_K \right)$ denotes the observation vector of the $K$ potential relays. Selecting a relays leads to the observation of its state, hence when relay $i$ is not selected then $z_i=\emptyset$ otherwise $z_i=s_i$. The set of all observations is denoted by $\mathcal{Z}$. 

\item \textbf{Conditional observation probability} $O\left(\bm{z}',\bm{s}',\bm{a}\right),: \mathcal{Z} \times \bm{\mathcal{S}} \times \mathcal{A} \rightarrow \left[0,1 \right]$ represents the probability of receiving an observation $\bm{z}^{'}\in Z$ knowing that the decision policy takes action $\bm{a} \in \mathcal{A}$ and by that transits to state $\bm{s}^{'} \in \bm{\mathcal{S}}$, we define:
\[O\left(\bm{z}^{'},\bm{s}^{'},\bm{a} \right)=O\left( \bm{z}^{'}|\bm{s}^{'},\bm{a}\right)\]
\[=Pr\left(\bm{z}_{t+1}=\bm{z}'|\bm{s}_{t+1}=\bm{s}', \bm{a}_{t}=\bm{a}\right)\]

\item \textbf{Reward} $R\left(\bm{s},\bm{a}\right)$ is the reward achieved by taking action $\bm{a}$ when the $K$ potential relays are in state $\bm{s}$. We suppose $R_{min} \leq R\left(\bm{s},\bm{a}\right)\leq R_{max}$ for all $\bm{s} \in \bm{\mathcal{S}}$ and $\bm{a} \in {\mathcal{A}}$.

\item \textbf{Cost} $C\left(\bm{s},\bm{a}\right)$  is the cost charged by taking action $\bm{a}$ when the $K$ potential relays are in state $\bm{s}$. We suppose $C_{min} \leq C\left(\bm{s},\bm{a}\right)\leq C_{max}$ for all $\bm{s} \in \bm{\mathcal{S}}$ and $\bm{a} \in {\mathcal{A}}$.

\item \textbf{Initial belief} $\bm{b}_0$ is a vector of $\bm{\mathcal{S}}$ elements that denotes the initial distribution probability of being at each state $\bm{s}\in \bm{\mathcal{S}}$. 

\item \textbf{Horizon} $T$ of the \gls{CPOMDP} represents the number of epochs of the relay selection policy between two successive discovery processes (i.e. $T$ is the discovery periodicity). 

\item $C_{th}$ the cost threshold.

\item $\gamma \in \left[ 0,1\right]$ as a discount factor.
\end{itemize}

The \gls{MU} chooses its action as function of the history of observations and actions that have been executed in the past. The result of  \cite{smallwood1973optimal} demonstrates that using the belief states for defining the optimal policy provides as much information as using the entire history of actions taken and observations received. Indeed, there is no need to explicitly save this history but having the current belief state is sufficient for deciding the upcoming actions.

Therefore, a \gls{POMDP} can be represented as a belief \gls{MDP} where every belief is a state. Since there is infinite belief states, this \gls{MDP} is defined over a continuous set of belief states. At a given epoch $t$, the \gls{MU} gathers all the information concerning the past decisions in the belief state defined as follows:
\begin{equation}
\label{eq.BeliefDef}
b_t\left( \bm{s} \right) := P\left( {\bm{s}}_t=\bm{s}| \bm{z}_t, \bm{a}_{t-1},\bm{z}_{t-1}, ... ,\bm{a}_0 \right) \forall \bm{s} \in {\bm{\mathcal{S}}}
\end{equation} 
The belief vector of $|{\bm{\mathcal{S}}}|$ elements is denoted by $\bm{b}_t=\left(b_t\left( \bm{s}_1\right),b_t\left(\bm{s}_2\right),...,b_t\left( \bm{s}_{|\bm{\mathcal{S}}|}\right) \right)$. By analogy, we define $b_t^i\left( {s} \right) :=P\left( {{s}}_i\left(t\right)={s}| \bm{z}_t, \bm{a}_{t-1},\bm{z}_{t-1}, ... ,\bm{a}_0 \right)$ as the belief state of the mobile relay $i$ and its corresponding belief vector: 
\[
\bm{b}_t^i=\left(b_t^i\left( S_1\right),b_t^i\left(S_2\right),...,b_t^i\left(S_{|{\mathcal{S}}|}\right) \right)
\]
Given the fact that the mobile relays move independently, the belief state $b^t\left( \bm{s} \right)$, at a given state $\bm{s}=\left(s_1,s_2,...,s_K\right) \in {\bm{\mathcal{S}}}$ and epoch $t$, can be deduced from the belief state of each relay as follows: 
\[
b_t\left( \bm{s} \right)=b_t{\left(s_1,s_2,...,s_K\right)}=\prod \limits_{i\in \mathcal{K}}b_t^i\left(s_i\right)
\] 

The belief state $\bm{b}^i_{t}$ of relay $i$ at epoch $t$ is recursively computed based on the previous belief state $\bm{b}^i_{t-1}$, previous action $\bm{a}_{t-1}$ and current observation $\bm{z}_{t}=\left(z_1\left(t\right),z_2\left(t\right),...,z_K\left(t\right)\right):$
\begin{equation}
\label{eq.UeBeliefUpdate}
 \bm{b}^i_{t} = \begin{cases}
\bm{b}^i_{t-1}\bm{P}_i  & \text{if } a_i\left(t-1\right)=0 \\
\bm{P}_i\left(z_i\left(t\right),:\right) & \text{if } a_i\left(t-1\right)=1
\end{cases} 
\end{equation}
where $\bm{P}_i\left(S_j,:\right) =\left(P_i\left(S_j,S_1\right),P_i\left(S_j,S_2\right), ..., P_i\left(S_j,S_{|\mathcal{S}|}\right)\right)$.
Therefore, when the action $\bm{a}$ (with equivalent set $\hat{{a}}$)  is taken and the observation $\bm{z}$ is detected then the belief of being in state $\bm{s}=\left( s_1, s_2, ..., s_K\right) \in {\bm{\mathcal{S}}}$ is updated as follows:
\begin{equation}
\label{eq.BeliefUpdate}
{b}_{\bm{z}}^{\bm{a}}\left({\bm{s}}\right)=\prod \limits _{i\in\hat{a}_t}P_i\left(z_i,s_i\right)\prod \limits _{j\in \mathcal{K}|\hat{a}_t}\sum \limits _ {s_j' \in \mathcal{S} } b_j\left(s'_j\right)P_j\left(s_j',s_j\right)
\end{equation}
The belief states pursue a Markov Process where the belief state at a given epoch depends on the belief, action and observation of the previous epoch with the following transition function:

\begin{equation}
\label{eq.Tb_bprime}
\tau\left(\bm{b},\bm{a},\bm{b}'\right)=\sum \limits_{\bm{z}\in Z}Pr\left(\bm{b}'|\bm{b},\bm{a},\bm{z}\right)Pr\left(\bm{z}|\bm{b},\bm{a}\right)
\end{equation}
 
\begin{equation}
\label{eq.Pz_ba}
\text{with }Pr\left(\bm{z}|\bm{b},\bm{a}\right)=\sum \limits_{\bm{s}'\in \mathcal{\bm{S}}} O\left(\bm{s}',\bm{a},\bm{z}\right)\sum \limits _{\bm{s}\in \mathcal{\bm{S}}}T\left(\bm{s},\bm{s}'\right)b\left(\bm{s}\right)
\end{equation}
\[
\text{and }P\left(\bm{b}'|\bm{b},\bm{a},\bm{z}\right)=\begin{cases}
1 & \text{ iff }  \bm{b}'=\bm{{b}}^{\bm{a}}_{\bm{z}}\\
0 & \text{ otherwise }
\end{cases}
\]
A policy $\pi\left(\bm{b}\right): \bm{b}\rightarrow\bm{a}$ is a function that determines the action $\bm{a}$ to take at each belief state $\bm{b}$. For a given initial policy $\bm{b}_0$, a policy $\pi$ is characterized by a value function $V^{r,\pi}\left(\bm{b}\right)$ for the reward evaluation and $V^{c,\pi}\left(\bm{b}\right)$ for the cost evaluation:
\begin{equation}
\label{eq.ValueFun_r}
V^{r,\pi}\left(\bm{b}_0\right)= \sum \limits_{t=1}^T\gamma^{t} \rho_r\left(\bm{b}_t,\pi\left( \bm{b}_t \right)\right) =\mathbb{E}\left[ \sum \limits_{t=1}^T\gamma^{t} R\left(\bm{s}_t,\bm{a}_t\right)|\bm{b}_0,\pi  \right]
\end{equation}
\begin{equation}
\label{eq.ValueFun_c}
V^{c,\pi}\left(\bm{b}_0\right)= \sum \limits_{t=1}^T\gamma^{t} \rho_c\left(\bm{b}_t,\pi\left( \bm{b}_t \right)\right) =\mathbb{E}\left[ \sum \limits_{t=1}^T\gamma^{t} C\left(\bm{s}_t,\bm{a}_t\right)|\bm{b}_0,\pi  \right]
\end{equation}
where the belief-based reward $\rho_r\left(\bm{b}, \bm{a}\right)$ and the belief-based cost $\rho_c\left(\bm{b}, \bm{a}\right)$ are: 
\begin{equation}
\label{eq.BeliefReward}
\rho_r(\bm{b},\bm{ a}) =\sum \limits _{\bm{s}\in \mathcal{\bm{S}}}b(\bm{s})R(\bm{s}, \bm{a})
\end{equation}
\begin{equation}
\label{eq.BeliefCost}
\rho_c(\bm{b},\bm{ a}) =\sum \limits _{\bm{s}\in \mathcal{\bm{S}}}b(\bm{s})C(\bm{s}, \bm{a})
\end{equation}
The value functions $V_t^{r,\pi}$ and $V_t^{c,\pi}$ at epoch $t$ and under policy $\pi$ are given by the \textit{Bellman} equation:
\begin{equation}
\label{eq.Bellman_r}
V_t^{r,\pi}\left(\bm{b}\right)=\rho_r\left(\bm{b}, \bm{a}_{\pi}\right)+\gamma\sum \limits _ {\bm{z}\in Z}Pr\left(\bm{z}|\bm{a}_{\pi},\bm{b}\right)V_{t+1}^{r,\pi}\left(\bm{b}^{\bm{a}_{\pi}}_{\bm{z}}\right)
\end{equation}
\begin{equation}
\label{eq.Bellman_c}
V_t^{c,\pi}\left(\bm{b}\right)=\rho_c\left(\bm{b}, \bm{a}_{\pi}\right)+\gamma\sum \limits _ {\bm{z}\in Z}Pr\left(\bm{z}|\bm{a}_{\pi},\bm{b}\right)V_{t+1}^{c,\pi}\left(\bm{b}^{\bm{a}_{\pi}}_{\bm{z}}\right)
\end{equation}
where $Pr\left(\bm{z}|\bm{a}_{\pi},\bm{b}\right)$ is given by equation (\ref{eq.Pz_ba}).

The action-value function $Q_t^{r,\pi}\left(\bm{b},\bm{a}\right)$ is the reward of taking action $\bm{a}$ at epoch $t$ and following policy $\pi$ thereafter:
\begin{equation}
\label{eq.Qfunction_r}
Q_t^{r,\pi}\left(\bm{b},\bm{a}\right)=\rho_r\left(\bm{b}, \bm{a}\right)+\gamma\sum \limits _ {\bm{z}\in Z}Pr\left(\bm{z}|\bm{a},\bm{b}\right)V_{t+1}^{r,\pi}\left(\bm{b}^{\bm{a}}_{\bm{z}}\right)
\end{equation}
The action-value function $Q_t^{c,\pi}\left(\bm{b},\bm{a}\right)$ is the cost of taking action $\bm{a}$ at epoch $t$ and following policy $\pi$ thereafter:
\begin{equation}
\label{eq.Qfunction_c}
Q_t^{c,\pi}\left(\bm{b},\bm{a}\right)=\rho_c\left(\bm{b}, \bm{a}\right)+\gamma\sum \limits _ {\bm{z}\in Z}Pr\left(\bm{z}|\bm{a},\bm{b}\right)V_{t+1}^{c,\pi}\left(\bm{b}^{\bm{a}}_{\bm{z}}\right)
\end{equation}
The objective of the optimal relay selection policy is to achieve the highest reward under cost constraints. The optimal policy $\pi^*$ maximizes the value function $V^{r,\pi}$ under constraints on $V^{c,\pi}$. There exists an optimal Markov policy which defines the action to be executed for each belief state supposing that the upcoming actions will be chosen in an optimal manner. The value of the optimal policy $\pi^*$ function $V_t^{*}\left(\bm{b}\right)$ satisfies the Bellman optimal equation:
\[V_t^{r,*}\left(\bm{b}\right)=\max \limits _{\bm{a}}Q_t^{r,\pi}\left(\bm{b},\bm{a}\right)\]
\begin{equation}
\label{eq.BellmanOpt}=
\max \limits _{\bm{a}} \left[ \rho_r\left(\bm{b}, \bm{a}\right)+\gamma\sum \limits _ {\bm{z}\in Z}Pr\left(\bm{z}|\bm{a},\bm {b}\right)V_{t+1}^{r,*}\left(\bm{b}^{\bm{a}}_{\bm{z}}\right) \right]
\end{equation}
s.t.
\[
V_t^{c,*}\left(\bm{b}\right)= \left[ \rho_c\left(\bm{b}, \bm{a}\right)+\gamma\sum \limits _ {\bm{z}\in Z}Pr\left(\bm{z}|\bm{a},\bm{b}\right)V_{t+1}^{c,*}\left(\bm{b}^{\bm{a}}_{\bm{z}}\right) \right]\leq C_{th}
\]
We define $\mathcal{L}$ the Bellman dynamic operator: \[ V_t^{*}\left(\bm{b}\right)=\mathcal{L}V_{t+1}^{*}\left(\bm{b}\right)\]

\section{Relay Selection Policies \label{sec:RS_Policies}}
\subsection{Exact Solution of \gls{CPOMDP} \label{subsec:RS_PWLC}}
\cite{smallwood1973optimal} shows that the value function of \gls{POMDP} is a \gls{PWLC} function over the infinite belief simplex (denoted by $\Delta$). Thus, the value function can be characterized by a finite set of value function vectors (hyperplans) \cite{Cassandra98exactand}. Each vector represents the optimal value of the value function in a given region of the belief space. The exact solution of the \gls{CPOMDP}, proposed in \cite{isom2008piecewise}, is inspired from the value iteration of \gls{POMDP}. This technique presents the value functions $V_t^{r,\pi}$ and $V_t^{c,\pi}$ as \gls{PWLC} functions over the infinite belief simplex $\Delta$. Therefore, the reward and cost value functions $V_t^{r,\pi}$ and $V_t^{c,\pi}$ are represented as a finite set of $\alpha$-vectors pairs $\left(\alpha_r^{i},\alpha_c^{i}  \right)$.

Considering $\mathcal{V}_t$ the set of $\alpha$-vectors pairs at epoch $t$, then the set $\mathcal{V}_{t+1}$ at the following epoch is constructed  by applying the following dynamic programming operator $\mathcal{L}$ over all the action-observation pairs: 
\footnote{The cross sum operator $\oplus$ of two sets $\mathcal{A}$ and $\mathcal{B}$ is given by the set : $\mathcal{A} \oplus \mathcal{B}= \lbrace a+b | a \in A , b\in B\rbrace$.}
 \[
 \Gamma_r^{\bm{a},*} \leftarrow \alpha_r^{\bm{a},*}\left(\bm{s} \right) = R \left( \bm{s},\bm{a} \right) \text{ and }  \Gamma_c^{\bm{a},*} \leftarrow \alpha_c^{\bm{a},*}\left(\bm{s} \right) = C \left( \bm{s},\bm{a} \right) 
 \]
 \[
 \Gamma_r^{\bm{a},\bm{z}} \leftarrow \alpha_{r,i}^{\bm{a},z}\left(\bm{s} \right) =\gamma \sum \limits _{\bm{s}' \in \bm{\mathcal{S}}}T \left( \bm{s},\bm{s}' \right)O\left( \bm{z},\bm{s}',\bm{a} \right)\alpha_r^i\left(\bm{s'}\right)
 \]
 \[\resizebox{0.9\hsize}{!}{$\hspace{-20pt}\Gamma_c^{\bm{a},\bm{z}} \leftarrow \alpha_{c,i}^{\bm{a},z}\left(\bm{s} \right) =\gamma \sum \limits _{\bm{s}' \in \bm{\mathcal{S}}}T \left( \bm{s},\bm{s}'\right)O\left( \bm{z},\bm{s}',\bm{a} \right)\alpha_c^i\left(\bm{s}'\right), \forall \left(\alpha_r^i,\alpha_c^i\right) \in \mathcal{V}_t$}
\]
 Thus,
 \[
 \Gamma_r^{\bm{a}}= \Gamma_r^{\bm{a},*}+ \oplus_{\bm{z}\in\mathcal{Z}} \Gamma_r^{\bm{a},\bm{z}}
\text{ and }  \Gamma_c^{\bm{a}}= \Gamma_c^{\bm{a},*}+ \oplus_{\bm{z}\in\mathcal{Z}} \Gamma_c^{\bm{a},\bm{z}}
 \]
This dynamic programming update produces $\mathcal{V}_{t+1}$:
\[
\mathcal{V}_{t+1}\leftarrow \mathcal{LV}_t =\bigcup \limits_{\bm{a}\in A}\left( \Gamma_c^{\bm{a}}, \Gamma_r^{\bm{a}}\right)
\]
At the worst case, each dynamic programming update will generate an exponentially increasing  $|\mathcal{V}_{t+1}|=|\mathcal{A}||\mathcal{V}_t|^{|\mathcal{Z}|}$ pairs of $\alpha$-vectors. However, some pairs of $\alpha$-vectors are never the optimal one in any region of the belief simplex (i.e. called useless vectors). Therefore, for mitigating this exponential explosion, different pruning algorithms were developed in order to exclude these useless vectors. For \gls{POMDP} problem, \cite{Cassandra98exactand} have proposed an algorithm that identifies these useless vectors in order to ignore them. For \gls{CPOMDP} a pruning operation for the value functions were proposed in \cite{isom2008piecewise} in order to generate the minimal set of pairs of $\alpha$-vectors at each iteration. The pruning function $Prune$ that we adopt for implementing this solution consists of keeping each pair $\left( \alpha_r^i, \alpha_c^i\right) \in \mathcal{V}_{t}$ that satisfies the cumulative cost constraint $\alpha_c^i . \bm{b}\leq C_{th}$ and that have the higher cumulative reward $\alpha_r^i . \bm{b}$ in some region of the belief simplex $\Delta$. This can be decided by solving a \gls{MILP} for each pair of $\alpha$-vector (one can refer to equation (3) in \cite{kim2011point}). This pruning method eliminates each pair of vectors $\left( \alpha_r^i, \alpha_c^i\right)$ that violates the cumulative cost constraint at a given iteration $t$. Note that this operation may lead to a suboptimal policy. Therefore, randomized policies are the subject of further study to achieve the optimal solution of \gls{CPOMDP}.

As presented in algorithm \ref{Algo_PWLC},  the iteration $t$ of the exact solution of \gls{CPOMDP} follows this procedure: (i) exact dynamic programming to generate the pairs $\alpha$-vectors, (ii) pruning operation, based on a mixed integer linear program, that produces the minimal set of $\alpha$-vectors, (iii) deducing the optimal value function.

\begin{algorithm}
\caption{Iteration of \gls{CPOMDP} Exact Solution}\label{Algo_PWLC}
\begin{algorithmic}[1]
\State \textbf{Input:} $\alpha$-vector set $\mathcal{V}_t$,  Actions $\mathcal{A}$, States $\bm{\mathcal{S}}$, Observations $\mathcal{Z}$, Reward function $R\left(\bm{s},\bm{a}\right)$, cost function $C\left(\bm{s},\bm{a}\right)$, Cost threshold $C_{th}$ 
\For {$\bm{a} \in \mathcal{A}$}
\State $\left( \alpha_r^{\bm{a},*},\alpha_c^{\bm{a},*} \right)\leftarrow \left( R\left(.,\bm{a}\right),C\left(.,\bm{a}\right)\right)$
\For{$\bm{z} \in \mathcal{Z}$}
\For{$\left( \alpha_r^i,\alpha_c^i \right) \in \mathcal{V}_{t}$}
\State  $\alpha_{r,i}^{\bm{a},\bm{z}}\left( \bm{s} \right)=\gamma\sum \limits_{\bm{s}' \in \bm{\mathcal{S}}}T\left( \bm{s},\bm{s}' \right)O\left( \bm{z},\bm{s}',\bm{a} \right)\alpha_r^i\left( \bm{s}'\right)$
\State  $\alpha_{c,i}^{\bm{a},\bm{z}}\left( \bm{s} \right)=\gamma\sum \limits_{\bm{s}' \in \bm{\mathcal{S}}}T\left( \bm{s},\bm{s}' \right)O\left( \bm{z},\bm{s}',\bm{a}  \right)\alpha_c^i\left(\bm{s}'\right)$
\State $\Gamma_r^{\bm{a},\bm{z}}=\Gamma_r^{\bm{a},\bm{z}} \cup  \alpha_{r,i}^{\bm{a},\bm{z}} $ and $\Gamma_c^{\bm{a},\bm{z}}=\Gamma_c^{\bm{a},\bm{z}} \cup  \alpha_{c,i}^{\bm{a},\bm{z}} $ 
\EndFor
\EndFor
\State $\Gamma^{\bm{a}}_r=  \alpha_r^{\bm{a},*}+\bigoplus \limits_{\bm{z}\in\mathcal{Z}} \Gamma_r^{\bm{a},\bm{z}} $ and $\Gamma^{\bm{a}}_c=  \alpha_c^{\bm{a},*}+\bigoplus \limits_{\bm{z}\in\mathcal{Z}} \Gamma_c^{\bm{a},\bm{z}} $
\EndFor
\State $\mathcal{V}_{t+1}= \mathcal{V}_{t+1}  \bigcup Prune\left( \bigcup \limits_{\bm{a}} \left( \Gamma_r^{\bm{a}},\Gamma_c^{\bm{a}}\right) \right) $
\State \textbf{Output:} $\mathcal{V}_{t+1}$
\end{algorithmic}
\end{algorithm}
Due to \gls{PWLC} propriety of the value function, the value iteration algorithm of \gls{CPOMDP} is limited to find the set of $\alpha$-vector pairs $\mathcal{V}_{t+1}$ that represents the value functions $V_{t+1}^r$ and  $V_{t+1}^c$ given the previous set $\mathcal{V}_{t}$. Constructing the set of hyperplans $\mathcal{V}_{t+1}$ by considering all the possible pairs of observations and actions over the previous set $\mathcal{V}_t$ has an exponential complexity of $\bigO{ |\mathcal{A}||\mathcal{V}_t|^{|\mathcal{Z}|}}$ (i.e. considering the states it gives $\bigO{|\bm{{\mathcal{S}}}|^2|A||V_t|^{|\mathcal{Z}|}}$). Since many pairs of vectors in $\mathcal{V}_t$ are dominated by others, pruning algorithms are developed in order to eliminate useless vectors and find the smallest subset sufficient for representing the value functions. Given the exponential complexity increase at each iteration, the necessity of pruning operations is obvious. However, pruning techniques consists on resolving a Mixed-Integer Linear Program \gls{MILP} for each pair of $\alpha$-vector. Therefore, value iteration algorithm for \gls{CPOMDP} remains computationally demanding to solve as the size of the problem increases. This requires the exploration of approximated solutions for finding the optimal solution of \gls{CPOMDP}.

\subsection{\gls{CPBVI} \label{subsec:RS_CPBVI}}
Since the value iteration algorithms for \gls{POMDP} do not scale to highly sized real problems, an approximate \gls{POMDP} planning solution called \gls{PBVI} was introduced in \cite{Pineau2003}. Most \gls{POMDP} problems unlikely reach most of the points in the belief simplex $\Delta$. Thus, it is preferable to focus the planning on the most probable belief points without considering all the possible belief points as exact algorithms do. Instead of considering the entire belief simplex, \gls{PBVI} limits the value update to a representative small set of  belief points $\mathcal{B}=\lbrace \bm{b}_0,\bm{b}_1,...,\bm{b}_q \rbrace$.  An $\alpha$-vector is initialized for each belief point and then the value of this vector is iteratively updated. The \gls{PBVI} algorithm can be simply adapted to solve \gls{CPOMDP} problem (e.g. \cite{kim2011point}). 

\subsubsection{\gls{CPBVI} algorithm}
We call \gls{CPBVI} the proposed suboptimal algorithm of relay selection inspired from \gls{PBVI}. Algorithm \ref{Algo_PWLC} is modified in such a way that the value function update is restrictively done over a finite belief set $\mathcal{B}$ and then the pruning algorithm chooses the dominated pair of vectors for each belief state $\bm{b} \in \mathcal{B}$. This allows the \gls{CPBVI} algorithm to achieve much better scalability.  At each iteration, \gls{CPBVI} follows the following steps for computing the set of $\alpha-$vectors $\mathcal{V}^{B}_{t+1}$ at epoch $t+1$ given the previous one  $\mathcal{V}^{B}_{t}$:
\begin{enumerate}
\item The first step is to generate the sets $\Gamma_r^{\bm{a},\bm{z}}$ and $\Gamma_c^{\bm{a},\bm{z}}$ for all $\bm{a} \in \mathcal{A}$, all $\bm{z} \in \mathcal{Z}$ and all pairs of $\alpha$-vectors $\left(\alpha_r^i,\alpha_c^i\right) \in \mathcal{V}^{B}_{t}$:
 \[
 \Gamma_r^{\bm{a},\bm{z}} \leftarrow \alpha_r^{\bm{a},\bm{z}}\left(\bm{s} \right) =\gamma \sum \limits _{\bm{s}' \in \bm{\mathcal{S}}}T \left( \bm{s},\bm{s}' \right)O\left( \bm{z},\bm{s}',\bm{a}  \right)\alpha_r^i\left(\bm{s'}\right)
 \]
 \[
 \Gamma_c^{\bm{a},\bm{z}} \leftarrow \alpha_c^{\bm{a},\bm{z}}\left(\bm{s} \right) =\gamma \sum \limits _{\bm{s}' \in \bm{\mathcal{S}}}T \left( \bm{s},\bm{s}'\right)O\left( \bm{z},\bm{s}',\bm{a}  \right)\alpha_c^i\left(\bm{s'}\right)
 \]
 \item The next step is to generate the sets $\Gamma_r^{\bm{a}}$ and $\Gamma_c^{\bm{a}}$ for all $\bm{a} \in \mathcal{A}$:
\[
\Gamma^{\bm{a}}_r= R\left(.,\bm{a}\right)+\bigoplus \limits_{\bm{z}\in\mathcal{Z}} \Gamma_r^{\bm{a},\bm{z}} \text{ and } \Gamma^{\bm{a}}_c= C\left(.,\bm{a}\right)+\bigoplus \limits_{\bm{z}\in\mathcal{Z}} \Gamma_c^{\bm{a},\bm{z}}
\]
Contrarily to \gls{PBVI} algorithm used for approximately solving \gls{POMDP}, we need the cross-summation overall the possible observations to find the sets $\Gamma^{\bm{a}}_r$ and $\Gamma^{\bm{a}}_c$. This cross-summation is mandatory in order to not impose the cost constraint on each action and observation pair while computing the best pair of $\alpha$-vector for each belief state $\bm{b}$. A local combinatorial explosion of $|\mathcal{A}||\mathcal{B}|^{|\mathcal{Z}|}$ is required.
\item The final step is the pruning operation. In our algorithm we propose to find one optimal pair of $\alpha$-vectors $\left(\alpha_r^{\bm{b}},\alpha_c^{\bm{b}}\right)$ for each belief point $\bm{b} \in \mathcal{B}$ as follows:
\[
\left(\alpha_r^{\bm{b}},\alpha_c^{\bm{b}}\right) = \argmax \limits_{\alpha_r \in \Gamma_r^{\bm{a}}, \alpha_c \in \Gamma_c^{\bm{a}}} \lbrace \alpha_r.\bm{b} : \alpha_c.\bm{b} \leq C_{th} \rbrace
\]
Note that  the proposed deterministic policy ensures the satisfaction of the cumulative cost constraint \ref{eq.Prob1} at each epoch. Thus, such deterministic policies can be sub-optimal for \gls{CPOMDP} (by analogy to \gls{CMDP}). Ideally, randomized policies that consider convex combination of $\alpha$-vectors during the pruning operation can be applied to guarantee optimality. 
\item Finally $\mathcal{V}^{B}_{t-1}=\bigcup \limits_{\bm{b} \in \mathcal{B}} \left(\alpha_r^{\bm{b}},\alpha_c^{\bm{b}}\right) $.
\end{enumerate}
\begin{algorithm}
\caption{\gls{CPBVI} Iteration}\label{algo.CPBVI}
\begin{algorithmic}[1]
\State \textbf{Input:} Pair of $\alpha$-vectors $ \mathcal{V}^{B}_{t}$, Actions $\mathcal{A}$, States $\bm{\mathcal{S}}$, Observations $\mathcal{Z}$, Rewards $R\left(\bm{s},\bm{a} \right)$, costs $C\left(\bm{s},\bm{a} \right)$, Belief subset $\mathcal{B}$, thresholds $C_{th}$ 
\For{$\bm{a} \in \mathcal{A}$}
\State $\left( \alpha_r^{\bm{a} ,*},\alpha_c^{\bm{a} ,*} \right)\leftarrow \left( R\left(.,\bm{a} \right),C\left(.,\bm{a} \right)\right)$
\For{$\bm{z} \in \mathcal{Z}$}
\For{$\left( \alpha_r^i,\alpha_c^i \right) \in \mathcal{V}^B_{t}$}
\State  $\alpha_{r,i}^{\bm{a} ,\bm{z} }\left(\bm{s} \right)=\gamma\sum \limits_{\bm{s}' \in \bm{\mathcal{S}}}T\left( \bm{s} ,\bm{s}' \right)O\left( \bm{s} ',\bm{a} ,\bm{z}  \right)\alpha_r^i\left( \bm{s} '\right)$
\State  $\alpha_{c,i}^{\bm{a} ,\bm{z}}\left( \bm{s} \right)=\gamma\sum \limits_{\bm{s} ' \in \bm{\mathcal{S}} }T\left( \bm{s} ,\bm{s} ' \right)O\left( \bm{s} ',\bm{a} ,\bm{z}  \right)\alpha_c^i\left( \bm{s} '\right)$
\State $\Gamma_r^{\bm{a} ,\bm{z} }=\Gamma_r^{\bm{a} ,\bm{z} } \cup  \alpha_{r,i}^{\bm{a} ,\bm{z} }$  and $\Gamma_c^{\bm{a} ,\bm{z} }=\Gamma_c^{\bm{a} ,\bm{z} } \cup  \alpha_{c,i}^{\bm{a} ,\bm{z} }$ 
\EndFor
\EndFor
\State $\Gamma_r^{\bm{a}}=  \alpha_r^{\bm{a},*}+\bigoplus \limits_{\bm{z} \in \mathcal{Z}} \Gamma_r^{\bm{a}, \bm{z}}$ and $\Gamma_c^{\bm{a}}=  \alpha_c^{\bm{a},*}+\bigoplus \limits_{\bm{z} \in \mathcal{Z}} \Gamma_c^{\bm{a}, \bm{z}}$
\EndFor
\For{$\bm{b} \in {\mathcal{B}}$}
\State $\left( \alpha_r^{\bm{b}},\alpha_c^{\bm{b}}\right) =  \argmax \limits_{ \alpha_r \in \Gamma_r^{\bm{a}}; \alpha_c \in \Gamma_c^{\bm{a}} } \sum \limits_{\bm{s} \in \bm{\mathcal{S}}} \alpha_r\left( \bm{s}\right)b\left(\bm{s}\right)$ s.t. $\sum \limits_{\bm{s} \in \bm{\mathcal{S}}} \alpha_c\left( \bm{s}\right)b\left(\bm{s}\right) \leq C_{th}$
\State $\mathcal{V}^B_{t+1}=\mathcal{V}^B_{t+1}\bigcup \left( \alpha_r^{\bm{b}},\alpha_c^{\bm{b}}\right)$
\EndFor
\State \textbf{Output:} $\mathcal{V}^B_{t+1}$
\end{algorithmic}
\end{algorithm} 

In order to study the complexity of each iteration of the \gls{CPBVI} algorithm, we detail the complexity of each step. Step 1 creates $|\mathcal{A}||\mathcal{Z}||\mathcal{B}|$ pair of vectors because the  previous set of $\alpha$-vector is limited to $|\mathcal{B}|$ components. The cross sum in the second step generates $\bigO{|\mathcal{A}||\mathcal{B}|^{|\mathcal{Z}|}}$ operations. Since the size of the $\alpha$-vector set remains constant (equals to $|\mathcal{B}|$), each \gls{CPBVI} update takes only polynomial time to be executed. This complexity is linear with the number of possible actions $|\mathcal{A}|$. However, in our settings, $|\mathcal{A}|=2^K$ leads to a poor scalability of our algorithm  in the number of potential relays $K$. Therefore, we propose, in subsection \ref{subsec:RS_GCPBVI}, a greedy version of the \gls{CPBVI} algorithm.

\subsubsection{Performance Evaluation}
A belief set $\mathcal{B}$ is characterized by a density $\epsilon^{\mathcal{B}}$ which is the maximum distance from any point in the belief simplex $\Delta$ to the set $\mathcal{B}$. 
\begin{definition}
The density $\epsilon^{\mathcal{B}}$ of a belief set $\mathcal{B}$ is defined as:
\begin{equation}
\label{eq.DensityB}
 \epsilon^{\mathcal{B}}:=\max \limits_{\bm{b}'\in \Delta}\min \limits_{\bm{b}\in \mathcal{B}} {||\bm{b}-\bm{b}'||}_1
\end{equation} 
\end{definition}
We denote by $V^{r,B}_t$ and $V^{c,B}_t$ the reward and cost value functions produced by the proposed \gls{CPBVI} algorithm. From theorem 1 of \cite{Pineau2003}, we know that for a belief set $\mathcal{B}$ of density $\epsilon^{\mathcal{B}}$, the errors $\eta^r_t$ and $\eta_t^c$ of the \gls{CPBVI} algorithm at horizon $t$ are bounded as follows:
\[
\eta_t^r:={||V^{r,B}_t-V_t^{r,*}||}_\infty \leq \frac{\left(R_{max}-R_{min}\right)\epsilon^{\mathcal{B}}}{{\left(1-\gamma\right)}^2}
\]
\begin{equation}
\label{eq.PBVI_Error}
\eta_t^c:={||V^{c,B}_t-V_t^{c,*}||}_\infty \leq \frac{\left(C_{max}-C_{min}\right)\epsilon^B}{{\left(1-\gamma\right)}^2}
\end{equation}

One can remark that this result does not take into account the case where the discount factor is equal to $1$. Hence, we extend the result to the case where $\gamma=1$.
\begin{theorem}
At horizon $h$ and discount factor $\gamma=1$, the errors $\eta^r_h$ and $\eta_h^c$ of applying the \gls{CPBVI} algorithm over a belief set $\mathcal{B}$ of density $\epsilon^{\mathcal{B}}$ are bounded as follows:
\[
\eta_h^r \leq \sum\limits_{t=1}^h t\left( R_{max}-R_{min}\right)\epsilon^{\mathcal{B}}
= \frac{h\left( h+1\right)}{2}\left( R_{max}-R_{min}\right)\epsilon^{\mathcal{B}}\]
\begin{equation}
\label{eq.PBVI_Error_Gamma1}
\eta_h^c \leq \sum\limits_{t=1}^h t\left( C_{max}-C_{min}\right)\epsilon^{\mathcal{B}}= \frac{h\left( h+1\right)}{2}\left( C_{max}-C_{min}\right)\epsilon^{\mathcal{B}}
\end{equation}
\end{theorem}

\subsubsection{Belief set $\mathcal{B}$ Selection}
The bound calculation of the error that occurs with the application of the \gls{CPBVI} algorithm shows the importance of the selected belief set $\mathcal{B}$. Indeed, these bounds are proportional to the density $\epsilon^{\mathcal{B}}$ of the chosen belief set $\mathcal{B}$. For this reason we show how to find the belief set, denoted by $\mathcal{B}^{\epsilon}$, that limits the upper bounds  $\eta_h^r$ and $\eta_h^c$ to a small $\epsilon$. It is obvious that a trade-off exists between the size of the belief set $\mathcal{B}$ and the precision of the value functions $\epsilon$.

We denote by $\bar{\Delta}^h$ the set of reachable belief states in our settings with an horizon $h$. Therefore, for a relay $i$, the corresponding set of reachable beliefs $\bar{\Delta}^h_i$ is given by:
\begin{equation}
\label{eq.Delta_i}
\bar{\Delta}^h_i=\lbrace{\bm{P}^t_i\left(S_j,: \right): t=\lbrace 1,...,h \rbrace \text{ and } S_j \in \mathcal{S} \rbrace}
\end{equation}
where $\bm{P}_i^t$ correspond to the transition matrix $\bm{P}_i$ to the power of $t$ (i.e. the transition matrix from a state to another after $t$ epochs). Therefore, the set $\bar{\Delta}^h_i$ contains $h\times |\mathcal{S}|$ vectors of  $|\mathcal{S}|$  elements. Thus, considering all the candidate relays $K$, the size of the overall reachable belief points for an horizon $h$ is given by $|\bar{\Delta}^h|={h |\mathcal{S}|}^K$ . 

It is clear that, for an horizon $h$, considering the set of reachable belief points $\bar{\Delta}^h$ as the belief set $\mathcal{B}$ of the \gls{CPBVI} algorithm minimizes the errors $\eta_h^r$ and $\eta_h^c$ (i.e. because all the reachable belief points were taken into account in $\mathcal{B}$). However, due to the exponentially large size of the set $\bar{\Delta}^h$, we propose to construct a smaller set of belief points $\mathcal{B}^{\epsilon}$ that is sufficient for upper bounding the performance errors $\eta_h^r$ and $\eta_h^c$ by a small $\epsilon > 0$. We denote by $\mathcal{N}_t \left( \bm{s} \right)$ the set of reachable states after passing $t$ epochs and knowing that the initial state is $\bm{s}$ (i.e. $\mathcal{N}_1\left( \bm{s} \right)$ represents the neighbor states of $ \bm{s}$).
For constructing the belief set $\mathcal{B}^{\epsilon}$, we give the following definition.
\begin{definition}
We define the $h$-belief set $\mathcal{B}\left( \bm{s},h\right)$ for an initial state $ \bm{s}_0$ and an horizon $h$ as follows:
\begin{equation}
\label{def.BeliefSet}
\mathcal{B}\left( \bm{s},h\right)=\lbrace \bigcup \limits_{t=1}^h \bm{T}^n\left(\mathcal{N}_t\left( \bm{s}\right) ,: \right): n=1,...,h-t  \rbrace
\end{equation}
with $\bm{T}^n$ the transition matrix $\bm{T}$ power to $n$. 
\end{definition}
We emphasize that the size of a set $|\mathcal{B}\left( \bm{s}_0,h\right)|$ is bounded by $|\mathcal{B}\left( \bm{s}_0,h\right)| \leq \sum \limits_{t=1}^h \left(h-t\right){\mathcal{N}_1\left( \bm{s}_0 \right)}^{t-1}$.

\begin{theorem}
\label{th.DensityBound}
The density $\epsilon^{\mathcal{B}}$ of an $h$-belief set $\mathcal{B}\left( \bm{s},h\right)$ is bounded by:
\begin{equation}
\label{eq.DensityBound}
\epsilon^{\mathcal{B}}\left( \mathcal{B}\left( \bm{s},h\right) \right) \leq 2\sum_{i=1}^K\frac{\lambda_i^{*h}}{\pi_{i,min}}
\end{equation}
where $\lambda_i^*$ is the highest eigenvalue of the transition matrix ${\bm{P}}_i$ of relay $i$ and $\pi_{i,min}=\min \limits_{{s} \in {\mathcal{S}}}{\pi}_i\left( {s} \right)$ with $\pi_i$ the stationary distribution corresponding to ${\bm{P}}_i$.
\end{theorem}
\begin{proof}
Please refer to the proofs section \ref{proof.DensityBound}.
\end{proof}
\begin{theorem}
\label{th.Belief_set_eps}
For an initial state $\bm{s}_0$ and a problem of horizon $T$ and discount factor $\gamma<1$, the belief set $\mathcal{B}^{\epsilon}$ that should be selected in order to achieve an $\epsilon$ performance (i.e. $\eta_T^r \leq \epsilon$ and $\eta_T^c \leq \epsilon$) is given by:
\begin{equation}
\label{eq.Belief_set_eps}
\mathcal{B}^{\epsilon}=\mathcal{B}\left( \bm{s}_0, \ceil*{\min\left( \frac{f_r\left(\epsilon \right)}{\log\left( \lambda ^*\right)}, \frac{f_c\left(\epsilon \right)}{\log\left( \lambda ^*\right)} \right)}\right)
\end{equation} 
with 
\[
f_r\left(\epsilon \right)=\log\left( \frac{\epsilon \pi_{min}{\left(1-\gamma\right)}^2}{2K\left(R_{max}-R_{min}\right)} \right)\]
\[f_c\left(\epsilon \right)=\log\left( \frac{\epsilon \pi_{min}{\left(1-\gamma\right)}^2}{2K\left(C_{max}-C_{min}\right)} \right)
\]
\[
\lambda^*=\max\limits_{i \in \mathcal{K}}\lambda_i^* \text{ ; } \pi_{min}^*=\min\limits_{i \in \mathcal{K}}\pi_{i,min}^*
\]
\end{theorem}
\begin{proof}
Please refer to the section \ref{proof.Belief_set_eps} for the proof of the theorem as well as its extension to $\gamma=1$.
\end{proof}

\subsection{Greedy \gls{CPBVI} \label{subsec:RS_GCPBVI}}
The computational complexity of each iteration of the \gls{CPBVI} algorithm is proportional to the size of the set of actions $|\mathcal{A}|=2^K$ which increases exponentially with the number of potential relays $K$. Therefore, we propose a \gls{GCPBVI} algorithm that exploits greedy maximization which consists of iteratively choosing the relays that should be selected in the aim of optimizing the value functions of the problem. This can be done by replacing $\argmax$ in line 14 of algorithm \ref{algo.CPBVI} by $greedy-argmax$. Greedy maximization is the main motivation for considering \gls{CPBVI} algorithms to solve \gls{CPOMDP}.  In fact, \gls{CPBVI} algorithms perform $\argmax$ operations over a finite set of belief points which is essential for the application of the greedy approach (i.e. contrarily to exact methods that compute the value function over all the continuous belief simplex). The \gls{GCPBVI} algorithm is basically deduced from the submodularity property of the $Q$-function.

\subsubsection{Submodularity of $Q$-function}
Since the upcoming result can be applied on both reward and cost $Q$-functions (given respectively by equations (\ref{eq.Qfunction_r}) and (\ref{eq.Qfunction_c})), we omit the index $r$ (for reward) and $c$ (for cost) and use $Q_t^{\pi}$ notation to generalize the results on both $Q_t^{r,\pi}$ and $Q_t^{c,\pi}$. Recall that:
\[
Q_t^{\pi}\left(\bm{b},\bm{a}\right)=\rho\left(\bm{b}, \bm{a}\right)+\sum \limits _ {\bm{z}\in Z}Pr\left(\bm{z}|\bm{a},\bm{b}\right)V_{t+1}^{\pi}\left(\bm{b}^{\bm{a}}_{\bm{z}}\right)\]
\begin{equation}
\label{eq.Q_function_new}=\rho\left(\bm{b}, \bm{a}\right)+\sum \limits _ {k=t+1}^{T}G_k^{\pi}\left(\bm{b}^t,\bm{a}^t\right) 
\end{equation}
where $G_k^{\pi}\left(\bm{b}^t,\bm{a}^t\right) $ is the expected immediate value (reward or cost) under policy $\pi$ at epoch $k$ conditioned on the belief $\bm{b}_t$ and the action $\bm{a}_t$ at epoch $t$, thus:
\begin{equation}
\label{eq.G_k}
G_k^{\pi}\left(\bm{b}^t,\bm{a}^t\right) =\gamma^{k}\sum \limits _{\bm{z}^{t:k}}P\left(\bm{z}^{t:k}|\bm{b}^t,\bm{a}^t, \pi\right)\rho\left(\bm{b}^k,\bm{a}^k|\pi\right)
\end{equation}
where $\bm{z}^{t:k}$ is a vector of observations received in the interval $t$ to $k$ epochs.
\begin{theorem}
\label{th.SubmodQ}
For all policies $\pi$, $Q_t^{r,\pi}\left(\bm{b},\bm{a}\right) $ and $Q_t^{c,\pi}\left(\bm{b},\bm{a}\right) $ are non-negatives, monotones and submodulars in $\bm{a}$.
\end{theorem}

\begin{proof}
Please refer to section \ref{proof.SubmodQ}.
\end{proof}

\subsubsection{\gls{GCPBVI} algorithm}
Since the $Q$-functions in terms of reward and cost are submodulars, then a greedy version of the \gls{CPBVI} algorithm is deduced (see algorithm \ref{Algo_GCPBVI}). The objective of this greedy algorithm is to avoid the iteration over all the possible actions $\mathcal{A}$  of size $2^K$. The iteration of the \gls{GCPBVI} algorithm shows that we limit the computation on considering each relay $i \in \mathcal{K}$ aside in such a way that the set $\mathcal{A}$ of all possible action is not introduced at any level of the algorithm (i.e. avoiding by that the exponential complexity in $K$).  Indeed, at each iteration, the \gls{GCPBVI} follows the steps below for computing the set of $\alpha-$vectors $\mathcal{V}^{G}_{t+1}$ at epoch $t+1$ given the previous set $\mathcal{V}^{G}_{t}$:
\begin{enumerate}
\item The first step is to generate the sets $\Gamma_r^{k,z_k}$ and $\Gamma_c^{k,z_k}$ for all the relays $k \in \mathcal{K}$, all possible observations of each relay $z_k \in \mathcal{S}$ and all pairs of $\alpha$-vectors $\left(\alpha_r^i,\alpha_c^i\right) \in \mathcal{V}^{G}_{t}$:
 \[
 \Gamma_r^{k,z_k} \leftarrow \alpha_r^{k,z_k}\left(\bm{s} \right) =\gamma \sum \limits _{\bm{s}' \in \bm{\mathcal{S}}}T \left( \bm{s},\bm{s}' \right)O\left(  z_k, \bm{s}',k \right)\alpha_r^i\left(\bm{s'}\right)
 \]
 \[
 \Gamma_c^{k,z_k} \leftarrow \alpha_c^{k,z_k}\left(\bm{s} \right) =\gamma \sum \limits _{\bm{s}' \in \bm{\mathcal{S}}}T \left( \bm{s},\bm{s}'\right)O\left( z_k, \bm{s}',k \right)\alpha_c^i\left(\bm{s'}\right)
 \]
 \item The next step is to generate the sets $\Gamma_r^{k}$ and $\Gamma_c^{k}$ for each relay $k \in \mathcal{K}$:
\[
\Gamma^{k}_r= R\left(.,\bm{a}\right)+\bigoplus \limits_{{z}_k \in\mathcal{S}} \Gamma_r^{k,z_k} \text{ and } \Gamma^{k}_c= C\left(.,\bm{a}\right)+\bigoplus \limits_{{z}_k\in\mathcal{S}} \Gamma_c^{k,z_k}
\]
 The computation of the sets $\Gamma_r^k$ and $\Gamma_c^k$ of each relay $k\in\mathcal{K}$  has a complexity of $\bigO{|\mathcal{S}||\mathcal{B}|^{|\mathcal{S}|}}$.
\item The final step aims to find the optimal action to take for each belief state $\bm{b} \in \mathcal{B}$. The $greedymax$ optimization is applied:
\[
\left(\alpha_r^{\bm{b}},\alpha_c^{\bm{b}}\right) = greedy\argmax \limits_{\alpha_r, \alpha_c } \lbrace \alpha_r.\bm{b} : \alpha_c.\bm{b} \leq C_{th} \rbrace
\]
The greedy optimization complexity is limited to $\bigO{K^2}$. 
\item Finally $\mathcal{V}^{G}_{t+1}=\bigcup \limits_{\bm{b} \in \mathcal{B}} \left(\alpha_r^{\bm{b}},\alpha_c^{\bm{b}}\right) $.
\end{enumerate}
Thus, the complexity of one \gls{GCPBVI} iteration $\bigO{K^2|\mathcal{S}||\mathcal{B}|^{|\mathcal{S}|}}$ enables much better scalability in the number of potential relays $K$. In the sequel, we denote by $V_t^{r,G}$ the reward value function generated by the \gls{GCPBVI} algorithm. 

\begin{algorithm}[ptb]
\caption{\gls{GCPBVI} Iteration}\label{Algo_GCPBVI}
\begin{algorithmic}
\State \textbf{Input:}  Pair of $\alpha$-vectors $ \mathcal{V}^{G}_{t}$, Relays $\mathcal{K}$, States $\mathcal{S}$, Observations $\mathcal{Z}$, Rewards $r_i\left(s\right)$, Costs $c_i\left(s\right)$ for relay $i$, Belief subset $\mathcal{B}^\epsilon$, Threshold $C_{th}$
\For{$k=1:K$}
\State $\left( \alpha_r^{k,*},\alpha_c^{k,*} \right)\leftarrow \left( r_i\left(.\right),c_i\left(.\right)\right)$
\For{$z_k \in \mathcal{S}$}
\For{$\left( \alpha_r^i,\alpha_c^i \right) \in \mathcal{V}^{G}_{t}$}
\State  $\alpha_{r,i}^{k,z_k}\left( \bm{s} \right)=\sum \limits_{\bm{s}' \in \mathcal{\bm{S}}}T\left( \bm{s},\bm{s}' \right)O\left( z_k, \bm{s}',k \right)\alpha_r^i\left( \bm{s}'\right)$
\State $\alpha_{c,i}^{k,z_k}\left( \bm{s} \right)=\sum \limits_{\bm{s}' \in \mathcal{\bm{S}}}T\left( \bm{s},\bm{s}' \right)O\left(  z_k, \bm{s}',k \right)\alpha_c^i\left( \bm{s}'\right)$
\State $\Gamma_r^{k ,z_k }=\Gamma_r^{k,z_k} \cup  \alpha_{r,i}^{k,z_k}$  and $\Gamma_c^{k,z_k}=\Gamma_c^{k,z_k} \cup  \alpha_{c,i}^{k,z_k}$ 
\EndFor
\EndFor
\State $\Gamma_r^{k}=  \alpha_r^{k,*}+\bigoplus \limits_{z_k \in \mathcal{S}} \Gamma_r^{k,z_k}$ and $\Gamma_c^{k}=  \alpha_c^{k,*}+\bigoplus \limits_{z_k \in \mathcal{S}}  \Gamma_c^{k,z_k}$
\EndFor
\For{$\bm{b} \in \mathcal{B}^\epsilon$}
\State $v_{sum}=0$, $\alpha_r^{\bm{b}}=0$ , $\alpha_c^{\bm{b}}=0$, $K_{left}=\mathcal{K}$, $\Gamma^{\bm{a},\bm{b}}=\emptyset$
\While {$|K_{left}|>0$}
\State $ \left( k^*,\alpha_r^*, \alpha_c^*\right)= \argmax \limits_{k \in K_{left}; \alpha_r \in \Gamma_r^k;  \alpha_c \in \Gamma_c^k} \frac{\sum \limits_{{\bm{s} \in \bm{\mathcal{S}}}} \alpha_r\left( \bm{s}\right){b}\left(\bm{s}\right)}{ \sum \limits_{\bm{s} \in \bm{\mathcal{S}}} \alpha_c\left( \bm{s}\right){b}\left(\bm{s}\right) } $  
\If {$v_{sum}+\alpha_c^* < C_{th}$}
\State  $\alpha_r^{\bm{b}}=\alpha_r^{\bm{b}}+\alpha_r^*$ ,  $\alpha_c^{\bm{b}}=\alpha_c^{\bm{b}}+\alpha_c^*$ and $v_{sum}=v_{sum}+\alpha_c^{*}\bm{b}$
\State $\Gamma^{\bm{a},\bm{b}}=\Gamma^{\bm{a},\bm{b}} \cup k^*$
\EndIf
\State $K_{left}=K_{left}|{ k^*}$
\EndWhile
\State $\mathcal{V}^{G}_{t+1}=\mathcal{V}^{G}_{t+1} \bigcup \left( \alpha_r^{\bm{b}},\alpha_c^{\bm{b}}\right)$ and $\Gamma_{t+1}^{\bm{a}}= \Gamma_{t+1}^{\bm{a}} \bigcup \Gamma^{\bm{a},\bm{b}}$
\EndFor
\State \textbf{Output:} $\mathcal{V}^{G}_{t+1}$ and $\Gamma_{t+1}^a$
\end{algorithmic}
\end{algorithm}

\subsubsection{Performance Evaluation}
\begin{theorem}
\label{th.Greedy_Bound}
At a given epoch $t$, the error in the reward value function due to greedy optimization is bounded by:
\begin{equation}
\label{eq.Greedy_Bound}
V^{r,G}_t\left( \bm{b}\right)\geq \left(1-\frac{1}{e}\right)^{2t}V_t^{r,B}\left( \bm{b}\right)
\end{equation}
\end{theorem}

\begin{proof} 
Please refer to proofs section \ref{proof.Greedy_Bound}.
\end{proof}
\section{Extension to Multi player scenario \label{sec:RS_MultiPlayer}}
The system model presented in section \ref{sec:RS_SystemModel} is extended to the multi player scenario where $N$ \gls{MU} aim to use \gls{D2D} aided relaying in order to improve the performance of their cellular communications. For the multi-player scenario, we generalized the single player notation as follows:
\begin{itemize}
\item $\mathcal{K}_i$ of size $K_i$ elements denotes the set of candidate relays discovered by each \gls{MU} $i$ with $i \in \lbrace 1,...,N\rbrace$. Thus, the set of potential relays is given by $\mathcal{K}=\bigcap\limits_{i=1}^N \mathcal{K}_i$ with $K=|\mathcal{K}|$.
\item $\overrightarrow{\bm{a}}=\left( \bm{a}_1,\bm{a}_2,...,\bm{a}_N\right)$ denotes the matrix of all the actions $\bm{a}_i$ taken by each \gls{MU} $i$ with $i \in \lbrace 1,...,N\rbrace$. Thus, the size of the set of possible actions $|\mathcal{A}|=2^{K}$. The size of matrix $\overrightarrow{\bm{a}}$ is $K \times N$ with $a\left(i,j\right)=1$ if relay $i$ is selected by \gls{MU} $j$ and $0$ otherwise.
\item ${\bm{s}}=\left( \bm{s}_1,\bm{s}_2,...,\bm{s}_K\right)$ denotes the matrix of all the states $\bm{s}_i$ of each relay $i$ with $i \in \lbrace 1,...,K\rbrace$ . Thus, the size of the set of possible states $|\bm{\mathcal{S}}|=\mathcal{S}^{K}$. We denote by $\bm{s}_{\mathcal{K}_i}$ the state vector of the candidate relays of \gls{MU} $i$.
\item The reward and cost model for a given action $\overrightarrow{\bm{a}}=\left( \bm{a}_1,\bm{a}_2,...,\bm{a}_N\right)$ and a given state ${\bm{s}}=\left( \bm{s}_1,\bm{s}_2,...,\bm{s}_K\right)$ is the following:
\begin{equation}
\label{eq.TotalReward_Gen}
R\left({\bm{s}} ,\overrightarrow{\bm{a}}\right)=\sum\limits_{i=1}^N R_i \left(\bm{s}_{\mathcal{K}_i},\bm{a}_i \right)
\end{equation}
where $R_i \left(\bm{s}_{\mathcal{K}_i},\bm{a}_i \right)$ of each \gls{MU} $i$ is given by equation (\ref{eq.TotalReward}). 
\begin{equation}
\label{eq.TotalCost_Gen}
C\left({\bm{s}} ,\overrightarrow{\bm{a}}\right)=\sum\limits_{i=1}^N C_i \left(\bm{s}_{\mathcal{K}_i},\bm{a}_i \right)
\end{equation}
where $C_i \left(\bm{s}_{\mathcal{K}_i},\bm{a}_i \right)$ of each \gls{MU} $i$ is given by equation (\ref{eq.TotalCost}). 
\end{itemize}
 For this multi player scenario, the relay selection scheme aims to find the decision policy $\left(\overrightarrow{\bm{a}}_1,...,\overrightarrow{\bm{a}}_T\right)$ that optimizes the following problem:
\[
\max \mathbb{E}\left[ \sum \limits _ {t=1}^T \gamma^{t} R\left({\bm{s}}_t ,\overrightarrow{\bm{a}}_t\right) \right]
\]
\begin{equation}
\label{eq.Prob1_Gen}
\text{s.t.  } \mathbb{E}\left[ \sum \limits _ {t=1}^T \gamma^{t} C_i\left({\bm{s}}_t ,\overrightarrow{\bm{a}}_t\right) \right] \leq C_{th} \,\, \forall i \in \lbrace 1,...,N \rbrace
\end{equation}
\subsection{Centralized Relay Selection Policy \label{subsec.MP_Cent}}
The first strategy of relay selection is to make the decision by the \gls{BS} in a centralized manner. Thus, the problem \ref{eq.Prob1_Gen} can be modeled, based on the single player scenario given in subsection \ref{ssec:CPOMDP}, as a \gls{CPOMDP} with the following characteristics: ${\bm{s}}$ represents the state of all the potential relays; $\overrightarrow{\bm{a}}$ represents the matrix of the taken action (with $\overrightarrow{\bm{a}}\left(i,j\right)=1$ if relay $i$ is selected by \gls{MU} $j$); $\bm{T}$ represents the transition matrix with $\bm{T}\left({\bm{s}};{\bm{s}'}\right)$ the probability of passing from a state ${\bm{s}}$ to another ${\bm{s}'}$; reward $R \left(\bm{s},\overrightarrow{\bm{a}}\right)$ and cost $C \left(\bm{s},\overrightarrow{\bm{a}} \right)$ models are respectively given by equations (\ref{eq.TotalReward_Gen}) and (\ref{eq.TotalCost_Gen}); $\bm{b}_0$ represents the initial belief of being in each state $\bm{s} \in \bm{\mathcal{S}}$; $T$ is the horizon of the problem and $C_{th}$ the cost threshold that should not be exceed by any \gls{MU}. 
Considering a centralized approach leads to the assumption that the state of each relay $i$ is reported to the \gls{BS} each time this relay $i$ is selected by one of the \gls{MU}s. Therefore, the observation $\bm{z}=\left(z_1,z_2,...,z_K\right)$ corresponds to $z_i=s_i$ if relay $i$ is selected by any \gls{MU} ($\overrightarrow{\bm{a}}\left(i,j\right)=1$ for any $j\in\mathcal{K}$) and $z_i=\emptyset$ otherwise ($\overrightarrow{\bm{a}}\left(i,j\right)=0$ for all $j\in\mathcal{K}$). $O\left(\bm{z}',\bm{s}',\overrightarrow{\bm{a}}\right)$ represents the probability of receiving an observation $\bm{z}'$ knowing that action $\overrightarrow{\bm{a}}$ is taken and that the relays transit to state $\bm{s}'$. 

The motivation of such centralized approach lies on the observation process described above. Since the \gls{BS} is informed by the state of each selected relay, \gls{BS} will have a global knowledge of all the selected relays' state contrarily to the local observation of each \gls{MU} which is limited by its own decision. Therefore, in this centralized approach, the \gls{BS} profits from having a global state information for inducing a more efficient relay selection decision.

This \gls{CPOMDP} is a straight-forward generalization of the single player \gls{CPOMDP} formulation. Thus, the results provided in the previous sections remain intact. This includes the \gls{GCPBVI} algorithm that is still applicable for the case of multi-player scenario since the submodularity property of the $Q$-functions remains valid. However, this centralized approach suffers from two challenges:
\begin{enumerate}
\item \textbf{Overhead}: This centralized approach requires that each \gls{MU} reports the states of its selected relays. This procedure generates an overhead of signaling.
\item \textbf{Complexity}: The state space for multi-\gls{MU}s scenario will blowup because it exponentially increase with the total number of \gls{MU}s $N$ as well as the total number of candidates relays $K$. Indeed, the total number of states is $|\bm{\mathcal{S}}|={|{\mathcal{S}}|}^{|\bigcap\limits_{n=1}^N \mathcal{K}_n|}$.
\end{enumerate}
\subsection{Distributed Relay Selection Policy \label{subsec.MP_Dist}}
To address the exponentially increasing in the size of the state space for multi-player scenario and to avoid the overhead of signaling, we propose a distributed variant for resolving the generalized \gls{CPOMDP}. We divide the multi-player problem into $N$ single-player problems (given by \ref{eq.Prob1}), one for each \gls{MU}, and solve them independently. For such distributed approach, each \gls{MU} will not take advantage of the observation of other \gls{MU}s because the states selected by each \gls{MU} are not shared between each others. However, distributed designs remain interesting for escaping the large amount of signaling that is required to be exchanged as well as reducing the computational complexity at the \gls{BS} level. 

The distributed approach is equivalent to considering $N$ parallel and independent single \gls{MU} problems. Each \gls{MU} $n$ (with $n=\lbrace 1,...,N\rbrace$) formulates its own relay selection problem as a \gls{CPOMDP} as shown in section  \ref{sec:RS_SimuResults}. Then, each \gls{MU} launches its relay selection strategy based on the \gls{GCPBVI} algorithm proposed in \ref{subsec:RS_GCPBVI}. Numerical results \ref{sec:RS_NumResults} show that this distributed approach reduces the complexity of the problem while achieving performance close to that of the centralized approach.

\section{Numerical Results \label{sec:RS_NumResults}}

\subsection{Numerical Settings}
We evaluate our claims by considering the case where the reward that we aim to maximize is the throughput of the \gls{MU} to \gls{BS} communications under energy consumption constraints. We consider that the relays are moving within different rectangular regions (i.e. we used to refer to these regions as states in the \gls{POMDP} formulation). An example of partitioning the horizontal plane as well as the vertical plane into $5$ parts is illustrated in figure \ref{fig.RS_NR_scenario} with $N=3$ \gls{MU}s and $K=6$ relays. This example of partitioning generates $25$ different locations for mobile relays. It is obvious that the preciseness and the performance of the results are improved by considering smaller partitioning granularity. However, this may lead to increase the complexity of the solution. 
\begin{figure} 
\hspace{-20pt}
\captionsetup{justification=centering}
\includegraphics[scale=0.6]{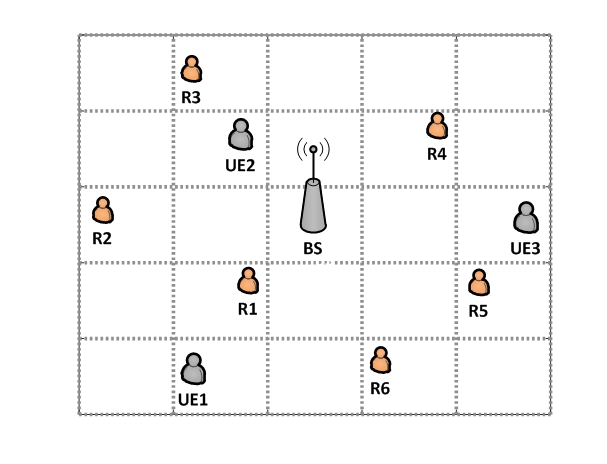}
\caption{\gls{D2D} relaying scenario for numerical results}
\label{fig.RS_NR_scenario} 
\end{figure}
We call $S_x$ (resp. $S_y$) the number of divisions that is applied at the horizontal (resp. vertical) axis. Each state $s \in \mathcal{S}$ is defined by a $x$ and $y$ coordinates. We denote by $\epsilon_{\text{fix}}$ as the probability that the relay stays in the same region. This parameter is essential for defining the mobility pattern that is assumed in this numerical section. The transition matrix in the horizontal plane $\bm{P}_x$ is constructed in such a way that each relay move left or right with an equal probability of $\frac{1}{2}\left(1-\sqrt{\epsilon_{\text{fix}}}\right)$. The transition matrix in the vertical plane $\bm{P}_y$ is constructed in such a way that each relay move up or down with an equal probability of $\frac{1}{2}\left(1-\sqrt{\epsilon_{\text{fix}}}\right)$. Considering that the horizontal and vertical travels are independent, the transition matrix of each relay $i$ $\bm{P}_i$ is deduced as follows: $\bm{P}_i\left[\left(x,y\right),\left(x',y'\right)\right]=\bm{P}_x\left(x,x'\right)\bm{P}_y\left(y,y'\right)$. One can remark that this structure of mobility matrix leads to a probability $\epsilon_{\text{fix}}$ of staying in the same location. The value of $\epsilon_{\text{fix}}$  is given in table \ref{table.RS_NumSettings}.

The cost and reward of a relay depends on its position $s=\left(x,y\right)$ in the network. Closer the source and the destination nodes are higher is the reward and  smaller is the cost of the corresponding source to destination link. In this section, we assume that the throughput is the reward criteria  and the energy consumption is the cost metric. The maximal achieved throughput is $R_{max}=500$ kbps/RB and the maximum transmitted power is $C_{max}=250$ mW. The reward (resp. cost) value of a given link between two nodes is inversely proportional (resp. proportional) to the values of the horizontal and vertical divisions $d_x$ and $d_y$. For the numerical results, we consider the following reward and cost models:
\begin{equation}
\label{eq.NR_Reward}
r\left(d_x,d_y\right)=\frac{R_{max}}{d_x \times d_y}
\end{equation}
\begin{equation}
\label{eq.NR_Cost}
c\left(d_x,d_y\right)=\frac{C_{max}}{\left(S_x-d_x+1\right)+\left(S_y-d_y+1\right)}
\end{equation}

The cost threshold $C_{th}$ that the average cumulative average cost should not exceed is given by table \ref{table.RS_NumSettings}. Since throughput is the performance criteria, the reward of the \gls{MU}-\gls{BS} link passing through relay $i$ is equal to the half of the $\min$ of the throughput of both links \gls{MU}-relay $i$ and relay $i$-\gls{BS}. However, the cost of such link is equal to the transmission power of relay $i$. Beside computing the cumulative average reward and cost, we evaluate the average cumulative \acrfull{EE} of the proposed algorithms. The average cumulative \gls{EE} metric of \gls{MU} $n$ is given by:
\begin{equation}
\label{eq.RS_EE}
EE_n=\mathbb{E}\left[ \sum\limits_{t=1}^T\gamma^{T-t} \frac{R_n\left(\bm{s}_t,\bm{a}_t\right)}{C_n\left(\bm{s}_t,\bm{a}_t\right)} \right]
\end{equation} 
The numerical settings that will commonly be used in the sequel are summarized in the table \ref{table.RS_NumSettings}.
\begin{table}
\centering
  \begin{tabular}{|c|c|}
\hline 
Settings Parameter & Value \tabularnewline
\hline
\hline 
Immobile Probability $\epsilon_{\text{fix}}$ & $0.7$ \tabularnewline
\hline 
Max. Reward $R_{max}$ & $500$  Kbps/RB \tabularnewline
\hline 
Max. Cost $C_{max}$ & $250$ mW \tabularnewline
\hline 
Cost Threshold $C_{th}$ & $1000$ mW \tabularnewline
\hline 
Reward and Cost Model & Given by \ref{eq.NR_Reward} and  \ref{eq.NR_Cost}\tabularnewline
\hline 
Discovery period $T$ & $5$ \tabularnewline
\hline 
Discount Factor $\gamma$ & $1$ \tabularnewline
\hline 
Number of realizations & $100$ \tabularnewline
\hline 
$\epsilon^B$ of \gls{CPBVI} algorithm & $0.01$ \tabularnewline
\hline 
\end{tabular}
\caption{Numerical Settings}
\label{table.RS_NumSettings}
\end{table}


\subsection{Single \gls{MU} scenario}
\subsubsection{\gls{GCPBVI} performance}
Before evaluating the performance of our relay selection policies, we show the motivation behind these approximations. In fact, we illustrate in figure \ref{fig.CompPWLCvsCPBVI} and \ref{fig.CompCPBVIvsGCPBVI} how the suggested algorithms highly reduce the computation complexity of the exact solution of the \gls{CPOMDP} problem. For a single \gls{MU} scenario and $K=5$ relays, we plot in figure \ref{fig.CompPWLCvsCPBVI} the $\log10$ of the complexity ratio between the exact solution and the proposed \gls{CPBVI} solution (i.e. given in  \ref{subsec:RS_CPBVI}) as function of the number of possible regions (i.e. called states). This shows that exact solution for of realistic sizes \gls{CPOMDP} is unfeasible . 

In addition, the exponential number of actions $2^K$ motivated us to propose a greedy design of the \gls{CPBVI} algorithm that avoids the study of all the possible actions. We consider $\mathcal{S}=25$ possible states and we plot in figure \ref{fig.CompCPBVIvsGCPBVI} the complexity ratio between the \gls{CPBVI} solution and its greedy alternative (i.e. given in  \ref{subsec:RS_GCPBVI}) as function of the number of the potential relays $K$. This figure shows that, even for a small number of potential relays $K=10$, we can reduce the complexity of a factor of $12$.

\begin{figure}
\centering
        \begin{minipage}[b]{0.475\textwidth}
            \centering
            \includegraphics[scale=0.35]{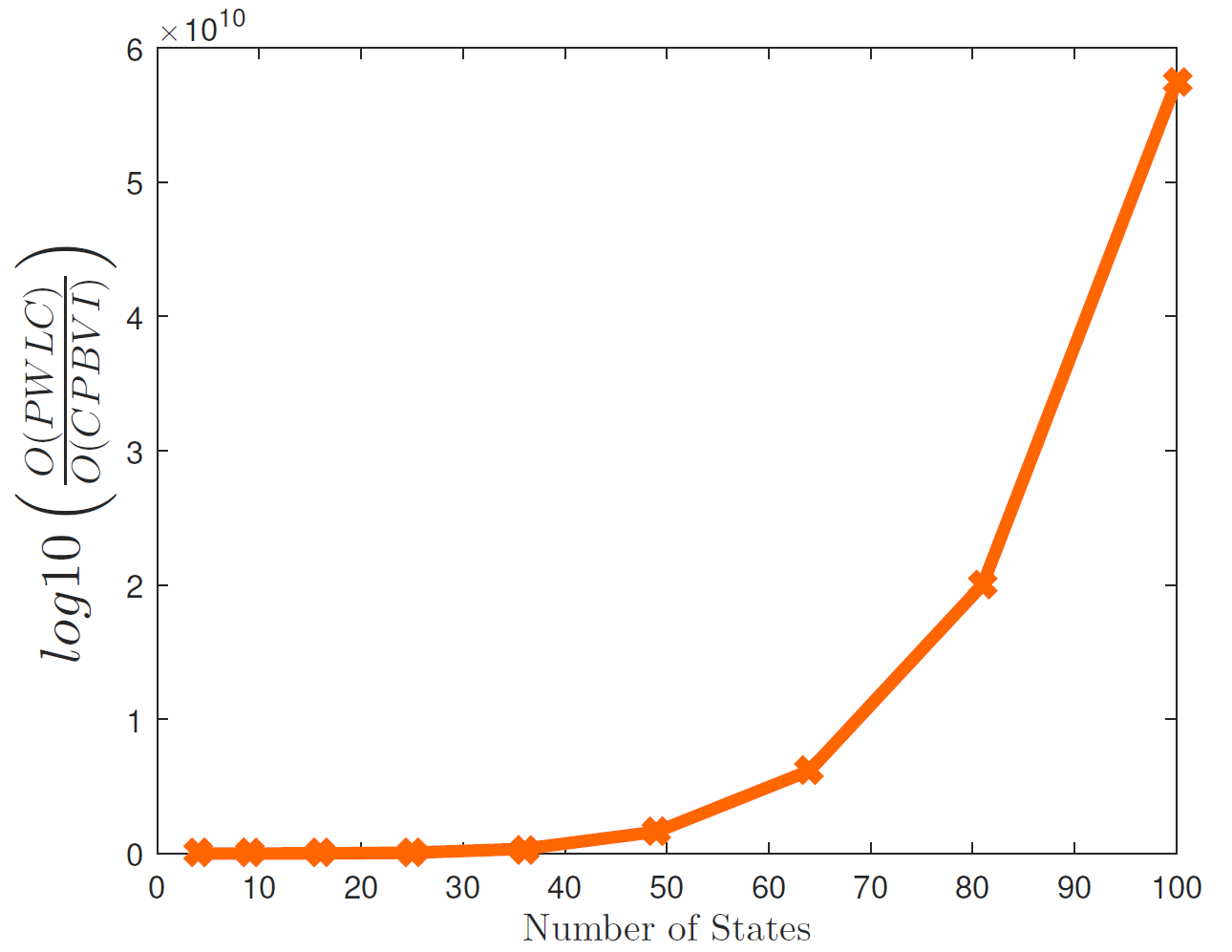}
            \caption[Network2]%
            {{\small Complexity of \gls{CPBVI} vs exact solution}}    
            \label{fig.CompPWLCvsCPBVI}
        \end{minipage}
        \hfill
        \begin{minipage}[b]{0.475\textwidth}  
            \centering 
            \includegraphics[scale=0.4]{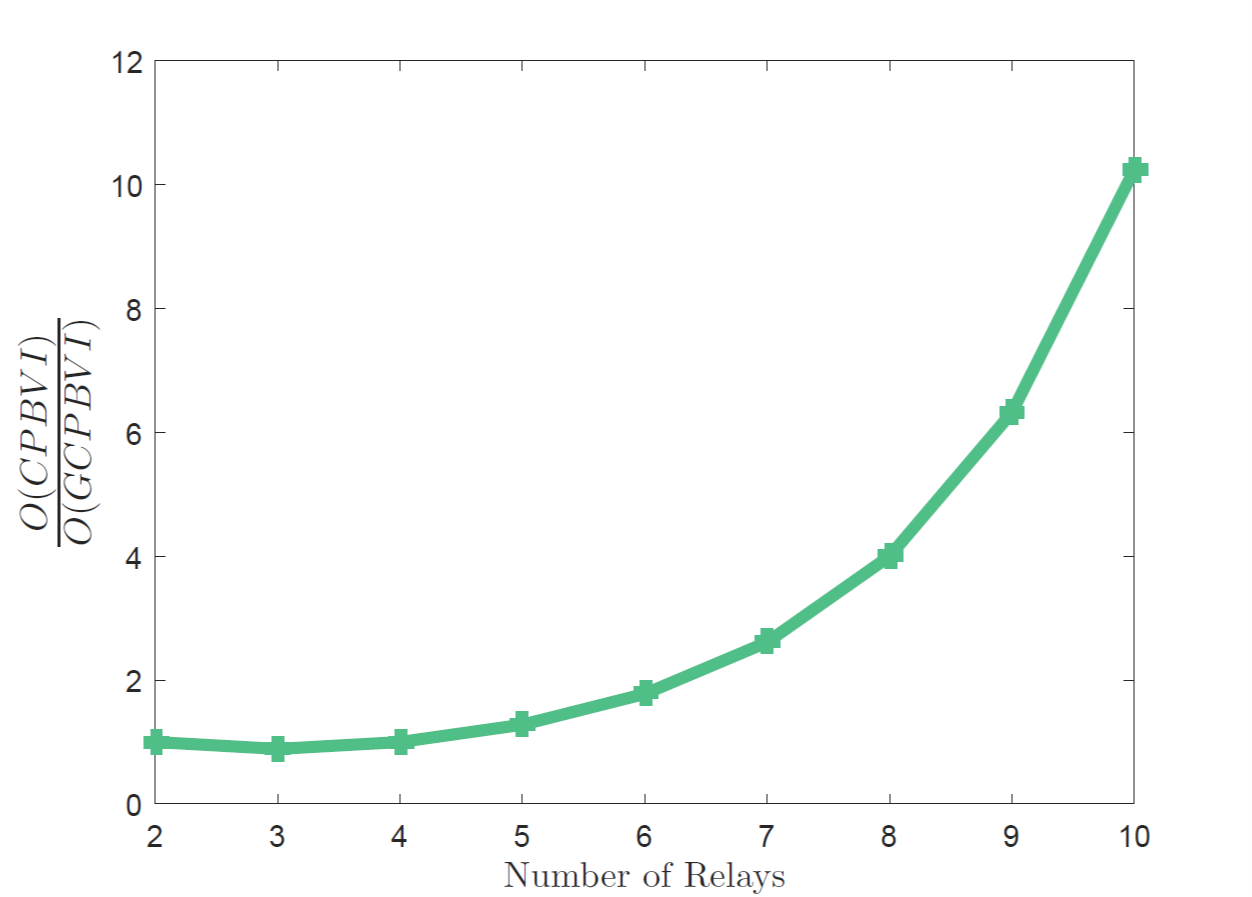}
            \caption[]%
            {{\small Complexity of \gls{GCPBVI} vs \gls{CPBVI}}}
            \label{fig.CompCPBVIvsGCPBVI}
\end{minipage}
\caption{Complexity study of the proposed relay selection policies}
       \label{fig.PWLCvsCPBVIvsGCPBVI} 
    \end{figure}
    
Moreover, we evaluate the performance of the greedy \gls{GCPBVI} algorithm (given in \ref{subsec:RS_GCPBVI}) compared to the \gls{CPBVI} scheme (given by \ref{subsec:RS_CPBVI}). Indeed, the comparison is done in terms of the average cumulative reward, cost and \gls{EE} of both \gls{CPBVI} and \gls{GCPBVI} relay selection solutions. For this comparison, we consider the following simple scenario: single \gls{MU} with $K=2$ potential relays and $S_x=3$ times $S_y=3$ rectangular regions (i.e. $9$ states). Figure \ref{fig.RS_NS_GCPBVI} verifies that the \gls{GCPBVI} almost return the same average cumulative reward as the \gls{CPBVI} algorithm. Both solutions have an average cumulative cost that is lower than the cost threshold $C_{th}$ given in \ref{table.RS_NumSettings}. Therefore, similar average cumulative energy efficiency is deduced for both algorithms.

   \begin{figure}[ptb]
        \centering
        \begin{minipage}[b]{0.475\textwidth}
            \centering
            \includegraphics[scale=0.4]{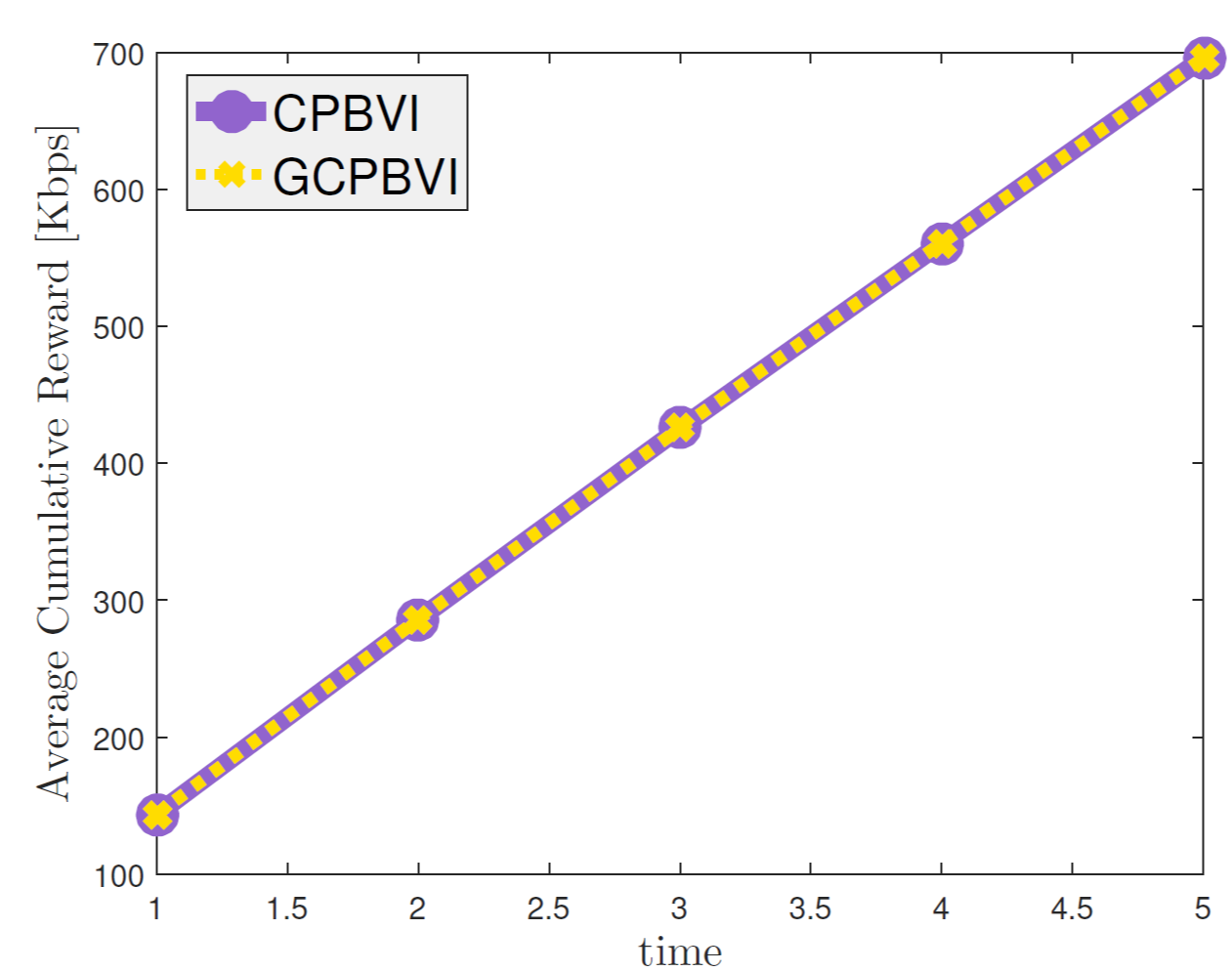}
            \caption[Network2]%
            {{\small Average Cumulative Reward}}    
            \label{fig.SP_CR}
        \end{minipage}
        \hfill
        \begin{minipage}[b]{0.475\textwidth}  
            \centering 
            \includegraphics[scale=0.4]{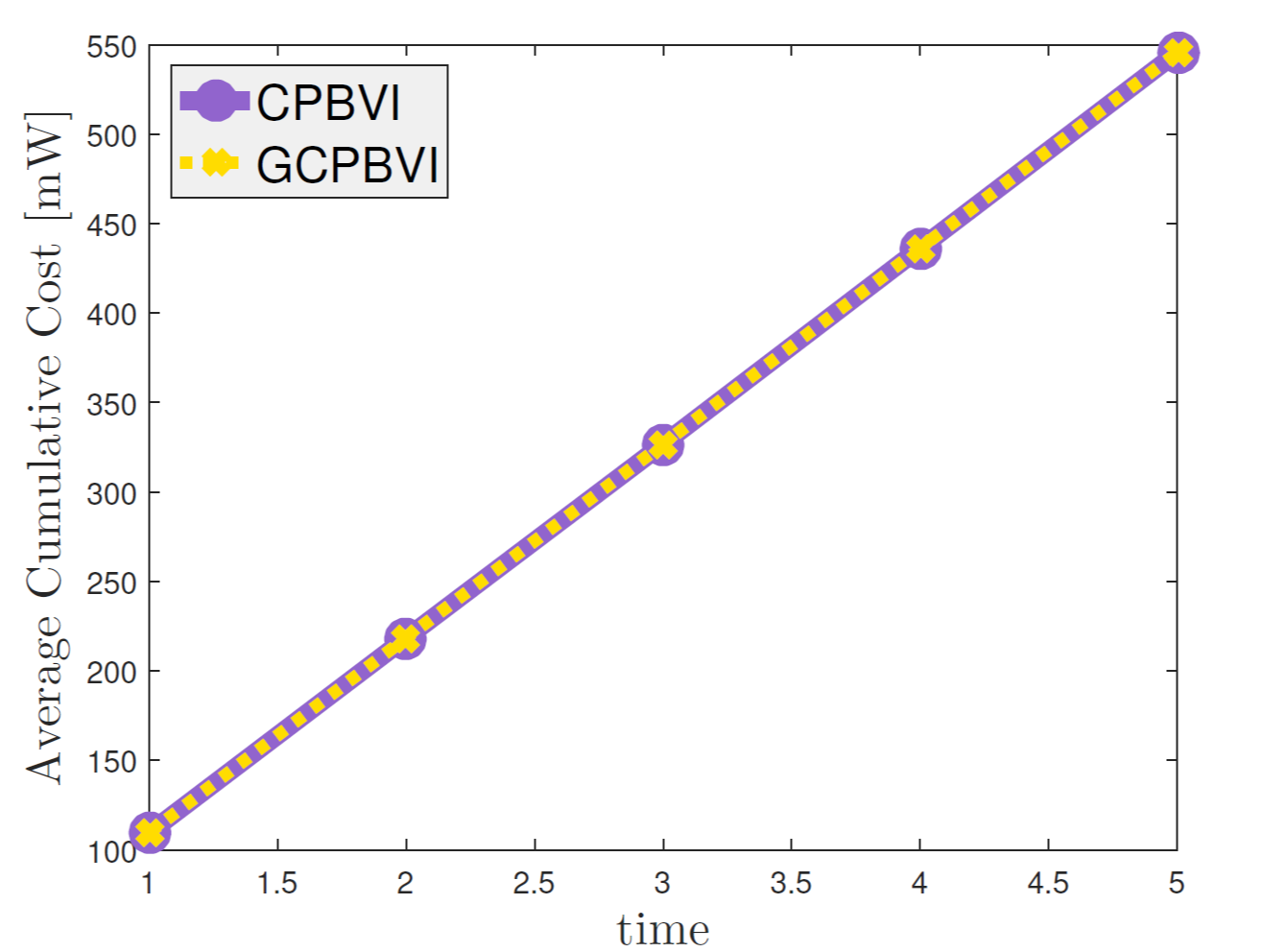}
            \caption[]%
            {{\small Average Cumulative Cost}}    
            \label{fig.SP_CC}
        \end{minipage}
        \vskip\baselineskip
         \centering
        \begin{minipage}[b]{0.475\textwidth}   
            \centering 
            \includegraphics[scale=0.4]{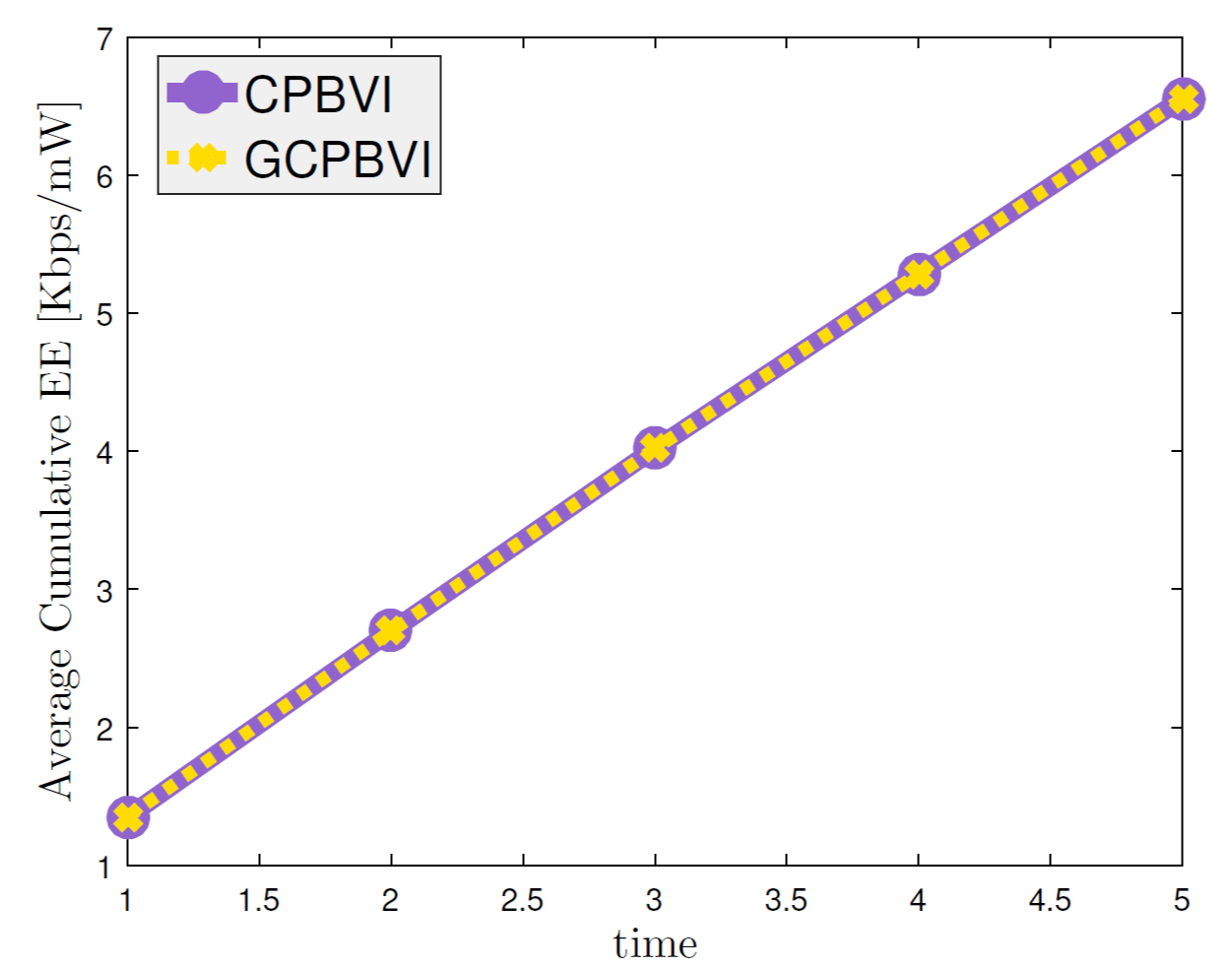}
            \caption[]
            {{\small Average Cumulative \gls{EE}}}    
            \label{fig.SP_CEE}
        \end{minipage}
\caption{Comparison between the performance of following relay selection policies: (i) the \gls{CPBVI} policy (see \ref{subsec:RS_CPBVI}) and (ii) the \gls{GCPBVI} policy (see \ref{subsec:RS_GCPBVI})}
       \label{fig.RS_NS_GCPBVI} 
    \end{figure}
\subsubsection{D2D relaying performance}
Recall that the aim of this work is to enhance the performance of cellular networks. Thus, we consider the throughput as the performance metric to study and show how implementing the proposed relay selection policy can improve the throughput of cellular networks. We consider the single \gls{MU} scenario with $K=3$ relays and $|\mathcal{S}|=16$ regions and we show the results of both scenarios: with and without \gls{D2D} relaying. For different speeds of state changing $v$, we plot in figure \ref{fig.SP_D2DvsCell} the histogram of the average cumulative reward of both scenarios with and without \gls{D2D} relaying. In the scenario where \gls{D2D} is enabled we apply the proposed \gls{GCPBVI} based relay selection algorithm. The speed $v$ in this figure illustrates the velocity of the relay in moving between the regions. Hence, a speed $v$ is modeled by considering a transition matrix of $\bm{P}^v$. Figure \ref{fig.SP_D2DvsCell} shows that we can gain up to $55\%$ percent in terms of throughput by deploying our policy of relay selection in a cellular network. We note a slow decreasing in the throughput when the speed of the relays increases.

\begin{centering}
\begin{figure} 
\centering
\captionsetup{justification=centering}
\includegraphics[scale=0.4]{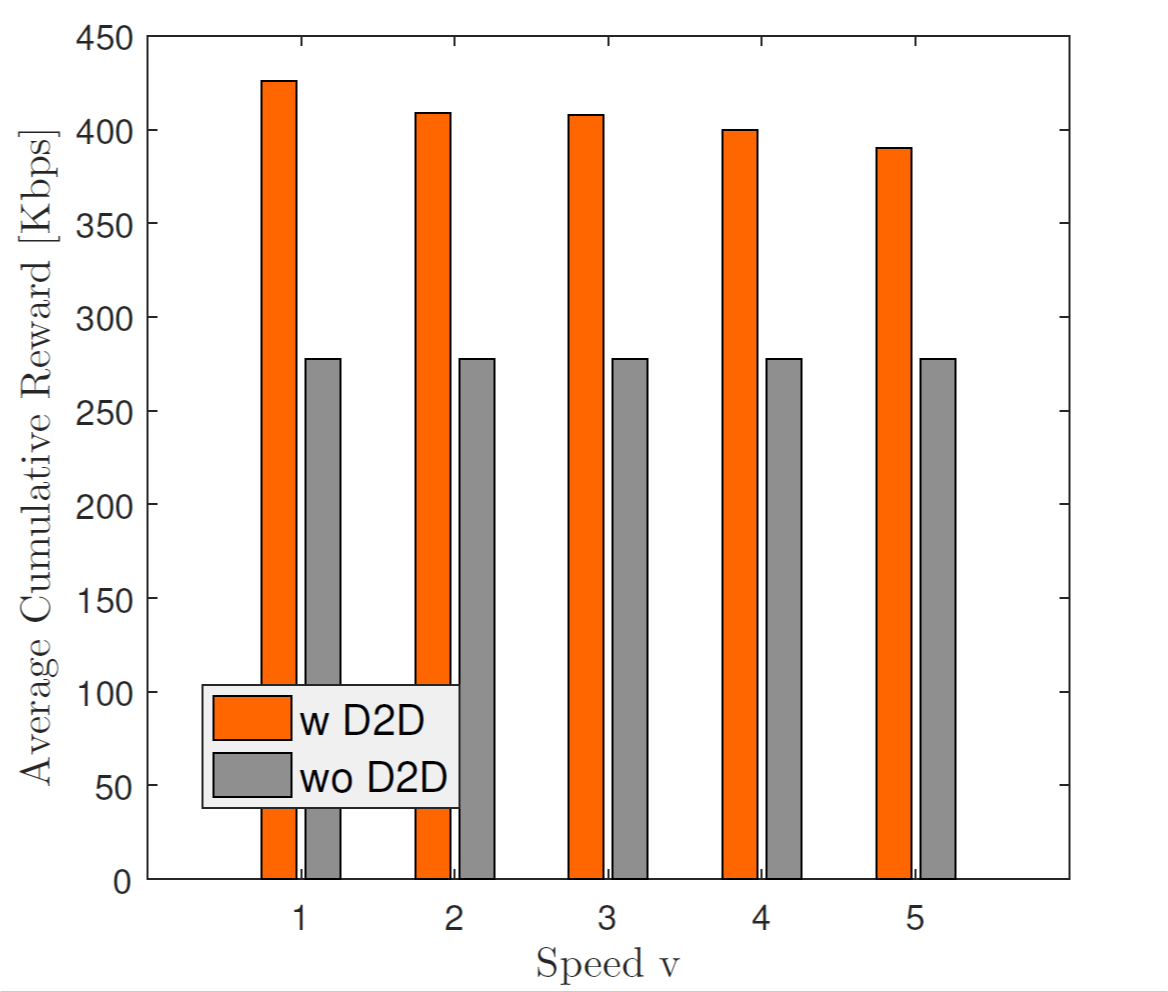}
\caption{For single \gls{MU}, performance comparison between both scenarios: with and without \gls{D2D} relaying}
\label{fig.SP_D2DvsCell} 
\end{figure}
\end{centering}

\subsection{Multi-\gls{MU} scenario}
We consider the case of multiple \gls{MU}s in the network. Both centralized and distributed relay selection policies were proposed. For evaluating these algorithms we consider $|\mathcal{S}|=16$ locations in the network and we study the complexity as well as the performance of both distributed and centralized relay selection solutions. 

\subsubsection{The Performance of the Distributed Approach}
In figure \ref{fig.CompMultiPlayer}, the complexity of the centralized solution for multi-\gls{MU} scenario (given in \ref{subsec.MP_Cent}) is compared to that of the distributed relay selection solution (given in \ref{subsec.MP_Dist}). The choice of developing a low complexity distributed approach for the multi-\gls{MU} scenario is validated in figure \ref{fig.CompMultiPlayer}. Indeed, the complexity reduction that the distributed approach offers is proportional to both the number of \gls{MU} $N$ and the number of relays $K$. As an example, the distributed approach reduces the complexity up to $150$ times compared to the centralized one in a realistic scenarios of $K=10$ relays and $N=10$ users.

\begin{figure}
\centering
        \begin{minipage}[b]{0.475\textwidth}
            \centering
            \includegraphics[scale=0.3]{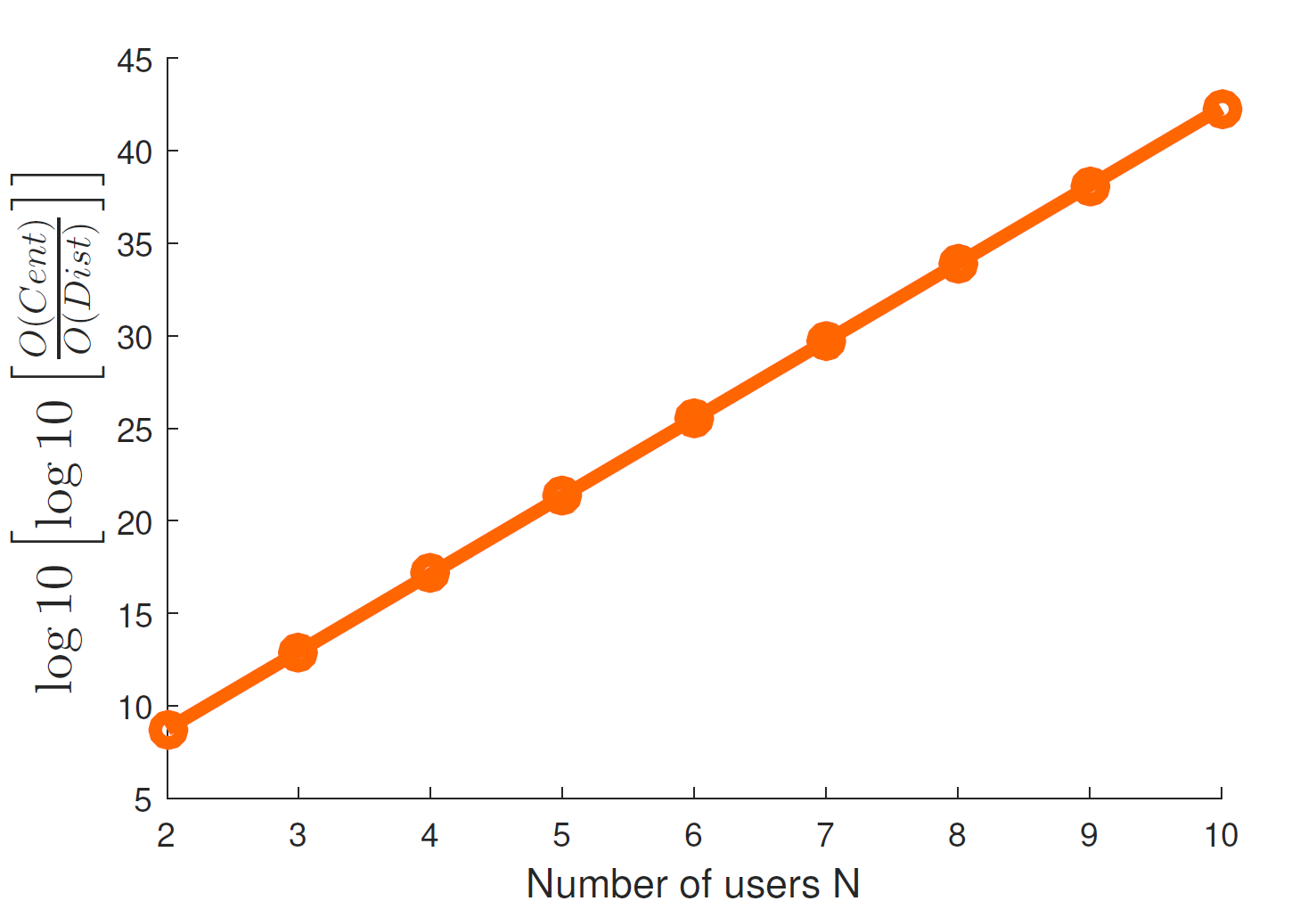}
            \caption[Network2]%
            {{\small Complexity of centralized vs distributed solutions as function of N}}    
            \label{fig.CompCentvsDist}
        \end{minipage}
        \hfill
        \begin{minipage}[b]{0.475\textwidth}  
            \centering 
            \includegraphics[scale=0.3]{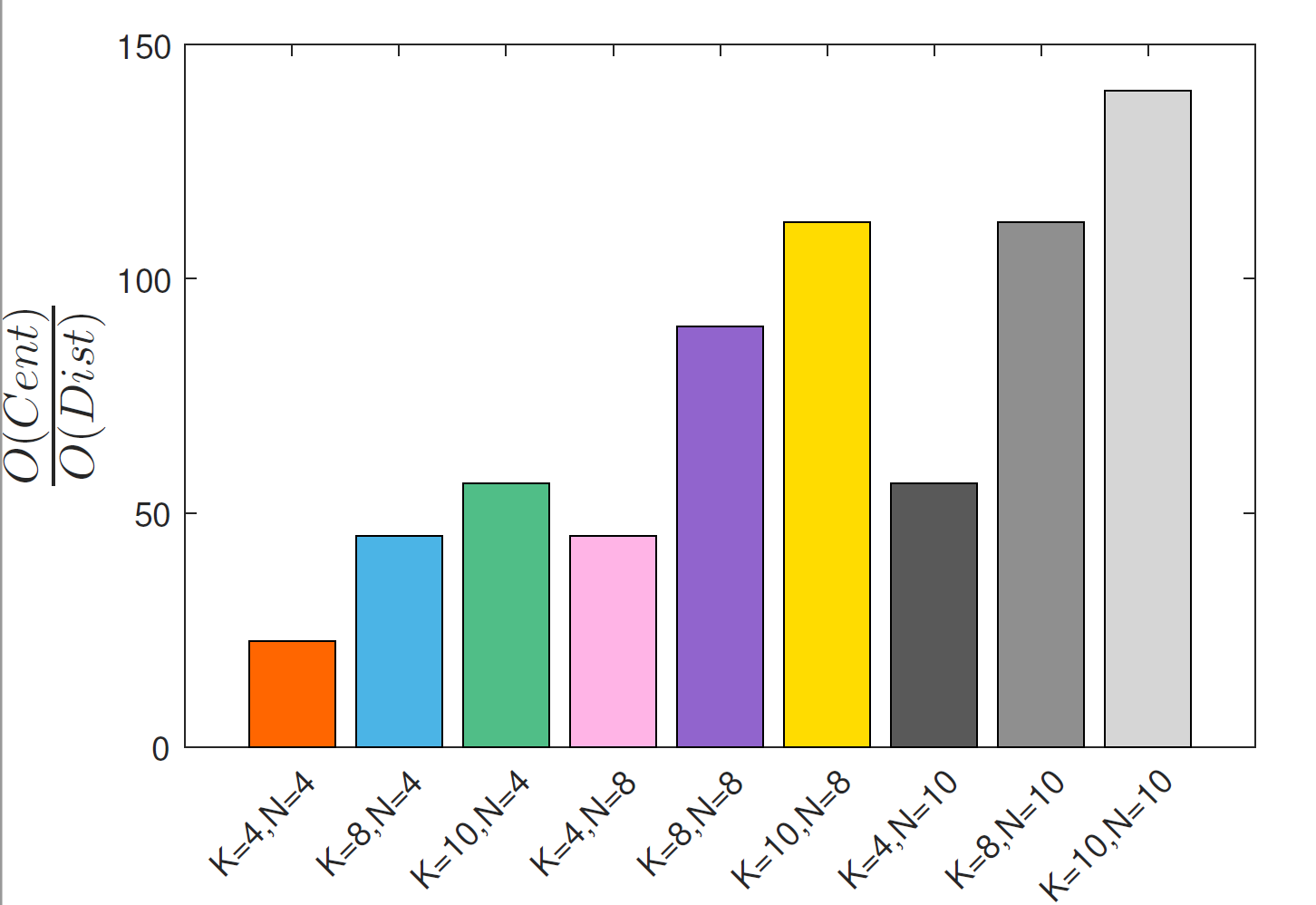}
            \caption[]%
            {{\small Complexity of centralized vs distributed solutions for different configurations}}
            \label{fig.CompCentvsDist_Config}
\end{minipage}
\caption{Complexity study of the proposed relay selection policies}
       \label{fig.CompMultiPlayer} 
    \end{figure}
    
The distributed approach reduces the computational complexity while satisfying a performance close to the centralized one. From figure \ref{fig.MP_CR} to figure \ref{fig.MP_D2DvsCell}, the performance of the distributed approach is compared to the centralized one. Similarly to the single \gls{MU} case, the comparison is done in terms of the average cumulative reward, cost and \gls{EE} of both distributed and centralized solutions. For this comparison, a simple scenario of $N=5$ and $K=4$ is considered. Figure \ref{fig.MP_CR} verifies that the low complexity distributed approach almost return the same average cumulative reward as the centralized one. Moreover, both centralized and distributed solutions verify the required cost constraints for each \gls{MU} (i.e. average cumulative cost lower than the cost threshold $C_{th}$ given in \ref{table.RS_NumSettings}). In result, similar average cumulative energy efficiency is deduced for both algorithms. We deduce, that applying distributed relay selection induces interesting performance enhancement of the network with a low computational complexity compared to the centralized approach.

     \begin{figure}[ptb]
        \centering
        \begin{minipage}[b]{0.475\textwidth}
            \centering
            \includegraphics[height=0.6\textwidth]{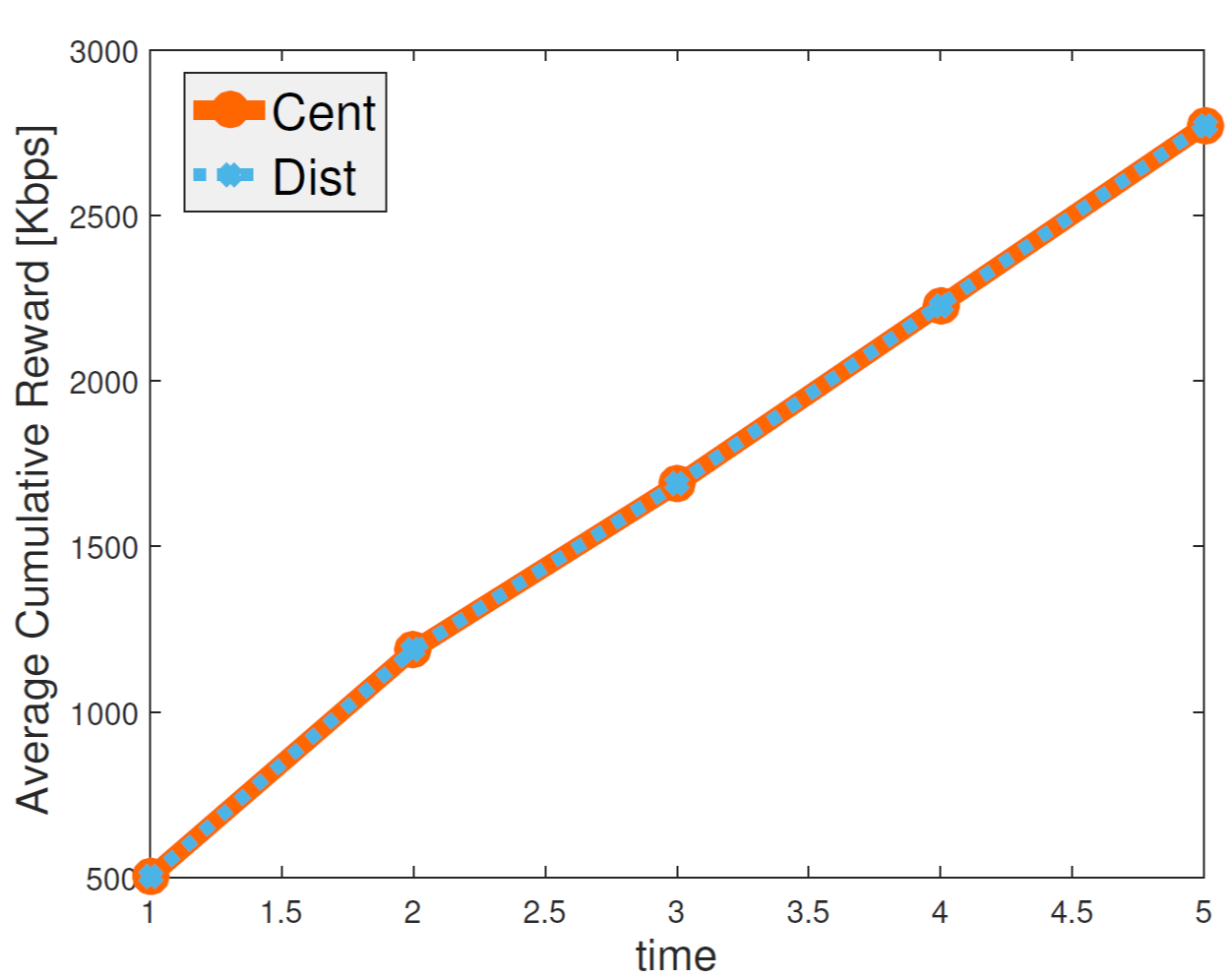}
            \caption[Network2]%
            {{\small Average Cumulative Reward}}    
            \label{fig.MP_CR}
        \end{minipage}
        \hfill
        \begin{minipage}[b]{0.475\textwidth}  
            \centering 
            \includegraphics[height=0.6\textwidth]{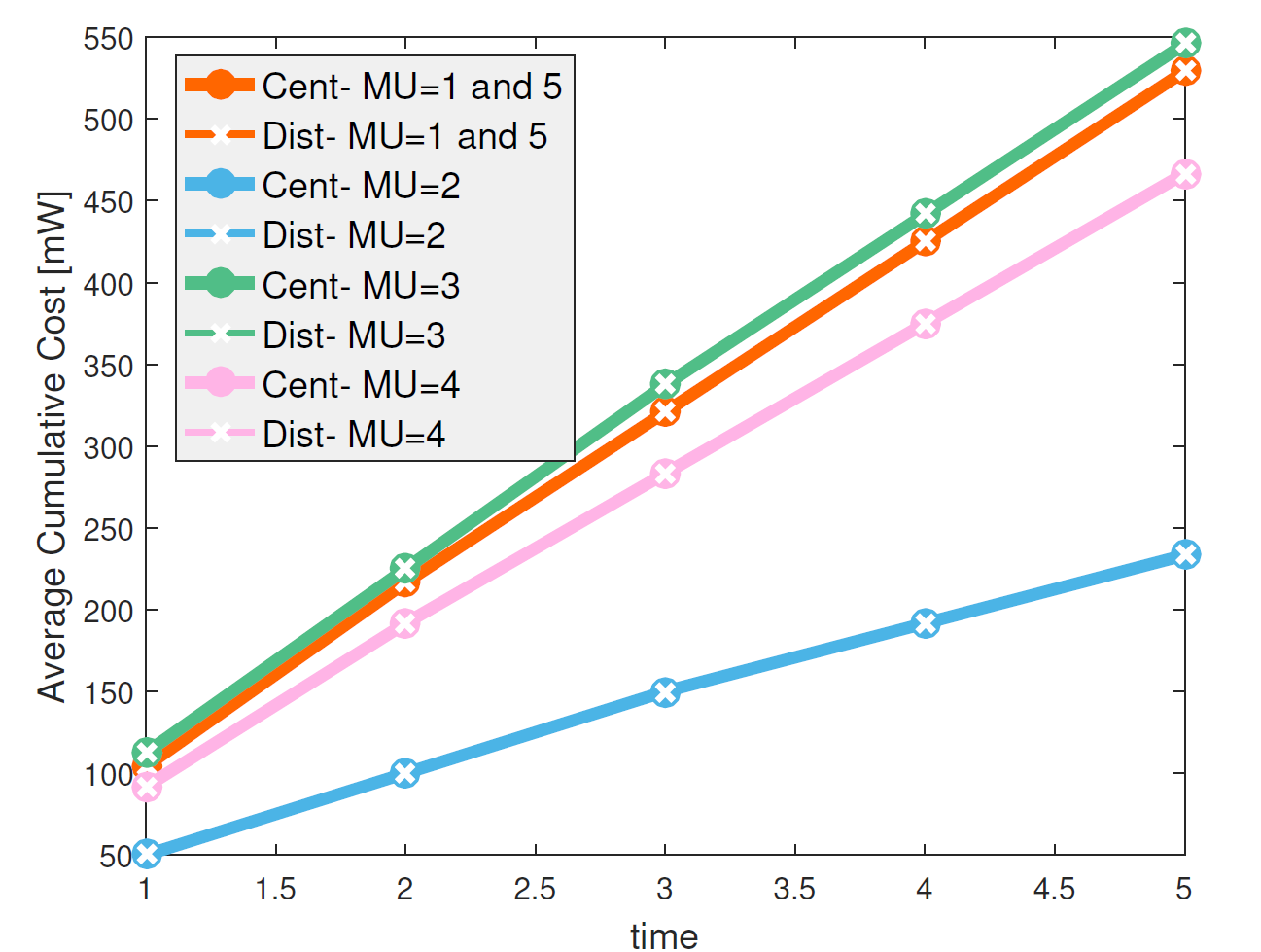}
            \caption[]%
            {{\small Average Cumulative Cost}}    
            \label{fig.MP_CC}
        \end{minipage}
        \vskip\baselineskip
        \begin{minipage}[b]{0.475\textwidth}   
            \centering 
            \includegraphics[height=0.6\textwidth]{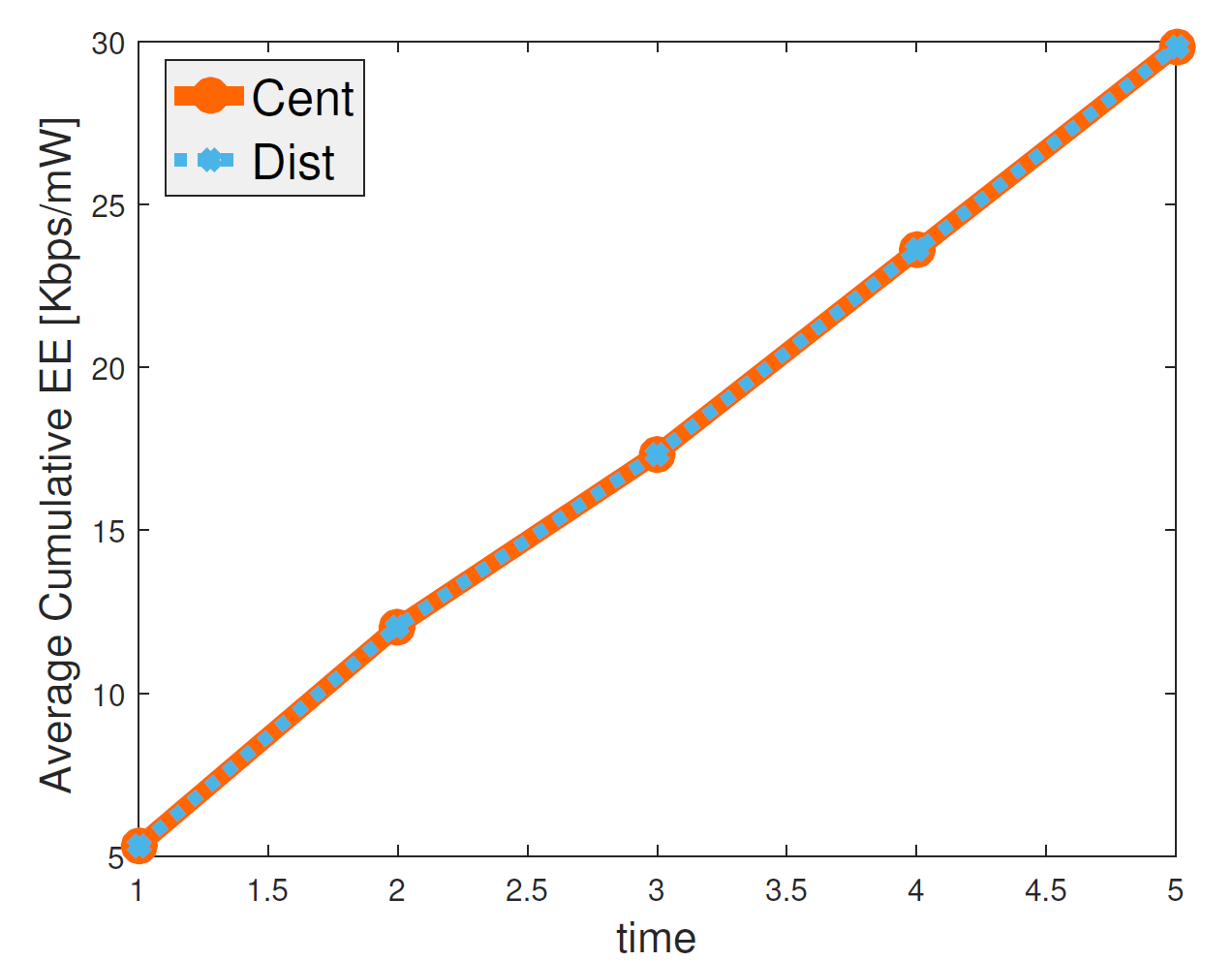}
            \caption[]%
            {{\small Average Cumulative \gls{EE}}}    
            \label{fig.MP_CEE}
        \end{minipage}
        \hfill
        \begin{minipage}[b]{0.475\textwidth}  
            \centering 
            \includegraphics[height=0.6\textwidth]{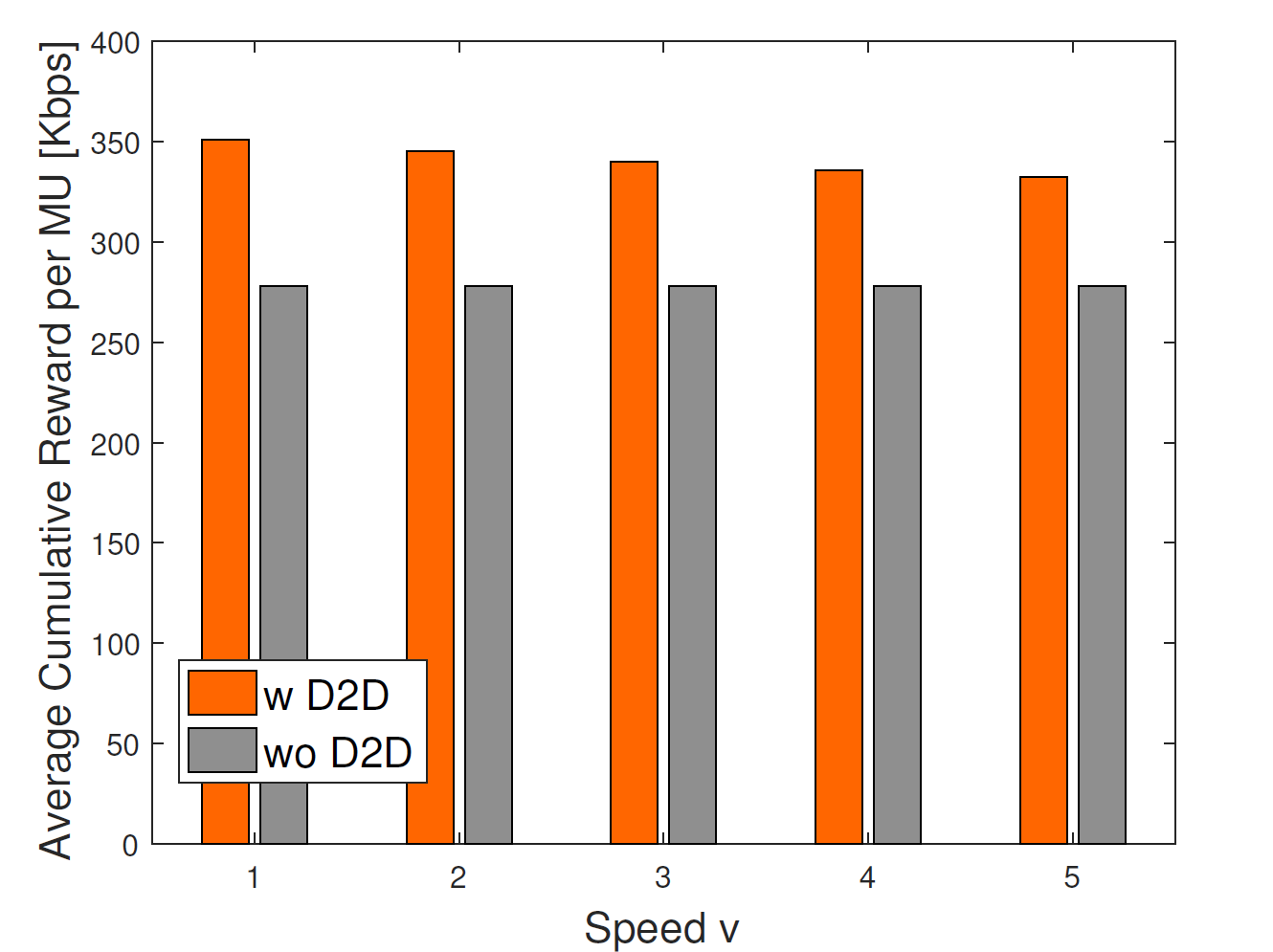}
            \caption[]%
            {{\small For multi \gls{MU} scenario, performance comparison between both scenarios: with and without \gls{D2D} relaying}}    
            \label{fig.MP_D2DvsCell}
        \end{minipage}
       \label{fig.RS_MP_CentDistPerf} 
    \end{figure}
\subsubsection{\gls{D2D} Relaying Performance}
We show how implementing the proposed relay selection policy can improve the throughput of cellular networks. We consider the multiple \gls{MU}s scenario with $N=5$, $K=4$ discovered relays and $|\mathcal{S}|=16$ regions. For different speeds of state changing $v$, we plot in figure \ref{fig.MP_D2DvsCell} the histogram of the average cumulative reward per \gls{MU} for both scenarios with and without \gls{D2D} relaying. In the scenario where \gls{D2D} is enabled we apply the proposed \gls{GCPBVI} based relay selection algorithm. The speed $v$ in this figure illustrates the velocity of the relays in moving between the regions. Hence, a speed $v$ is modeled by considering a transition matrix of $\bm{P}^v$. Each \gls{MU} applies, in a distributed manner, the relay selection policy proposed in \ref{subsec:RS_GCPBVI}. Figure \ref{fig.MP_D2DvsCell} shows that, in average, a \gls{MU} can gain up to $30\%$ percent of throughput by deploying our policy of relay selection in a cellular network. We note a slow decreasing in the throughput when the speed of the relays increases.

\section{Simulation Results \label{sec:RS_SimuResults}}
We have developed a \gls{3GPP} compliant system level simulator that supports \gls{DL}, \gls{UL} and \gls{D2D} communications. We describe the details of the simulations' settings in appendix \ref{sec:SimuSetting}. In this section we describe the \gls{D2D} relaying scenario that has been implemented in the simulator and we evaluate the gain that our proposed relay selection policy achieves in a simulated cellular network. 

\subsection{Scenario}
To complete the results of this paper, we simulate the scenario presented in the flow chart \ref{fig.FlowChart}. The performance of \gls{D2D} relaying is evaluated separately for \gls{UL} and \gls{DL} scenarios. In this report, we will only present the results corresponding to the \gls{DL} case for lack of space; however, we would like to note that similar results were evaluated for the \gls{UL} scenario. The initialization phase of the simulations consists of  filling the simulation's parameters described in \ref{sec:SimuSetting} and then generating the cellular network (\gls{BS},  \gls{UE}s, received power map, traffic generation etc.).  We denote by $N$ the number of users and $K$ the number of relays generated in each macro cell of the network. The mobility model of the relays is the one considered in the numerical section \ref{sec:RS_NumResults}. We use $v$ to denote the speed of the relays. (i.e. transition matrix $\bm{P}^v$).

The three main blocks of the simulated scenario are the following: (i) \textbf{relays discovery} where each user discovers nearby relays, (ii)\textbf{ relays selection} where each user selects the relays that will ensure the transmission of its traffic (none of the relays is selected if the user found it more beneficial to have traditional cellular communications) and (iii) \textbf{transmission} where a user is scheduled in order to transmit its data (i.e. either in a direct way or through relays). With a periodicity of $10$ s, each user runs its discovery process from which it deduces nearby devices. The relays move from a region to another (among the $16$ possible regions of each cell) with a time-scale of $T_r=2$ s. The relay selection decision is taken with the same periodicity $T_r$. Each \gls{TTI}, the fast fading is updated for all cellular and \gls{D2D} communications. The allocation of cellular and \gls{D2D} resources is done at the \gls{TTI} scale. After data transmission, if the whole user's file is transmitted then a new file is generated according to the FTP2 model described in \ref{sec:SimuSetting}.

We denote by $T_{simu}$ the simulation duration of $30 000$ \gls{TTI}s and $W=10$ MHz the bandwidth of the cellular network. The three metrics that are considered for studying the performance of the proposed relay selection policy are the following:
\begin{itemize}
\item \textbf{\gls{CDF} of the \gls{SINR}} which reflexes the coverage area enhancement due to \gls{D2D} relaying in cellular network.
\item If $V_{data}$ [bits] is the amount of data transmitted by all the users then the \textbf{first throughput} criteria  (i.e. from the network point of view) is computed as follows:
\begin{equation}
\label{eq.Throughput1}
Th_1=\frac{V_{data}}{\bar{N}_{waiting}\times T_{simu}\times W}
\end{equation}
where $\bar{N}_{waiting}$ is the average number of users that are waiting to be served at a given \gls{TTI}.
\item If $F$ is the number of files that have been transmitted with $R_i$ the throughput at which the $i^{th}$ was transmitted, then \textbf{the second throughput} criteria  (i.e. from the user's point of view) is computed as follows:
\begin{equation}
\label{eq.Throughput2}
Th_2=\frac{\sum \limits_{i=1}^F R_i}{F}
\end{equation}
\end{itemize}

\begin{centering}
\begin{figure} 
\centering
\captionsetup{justification=centering}
\includegraphics[scale=0.8]{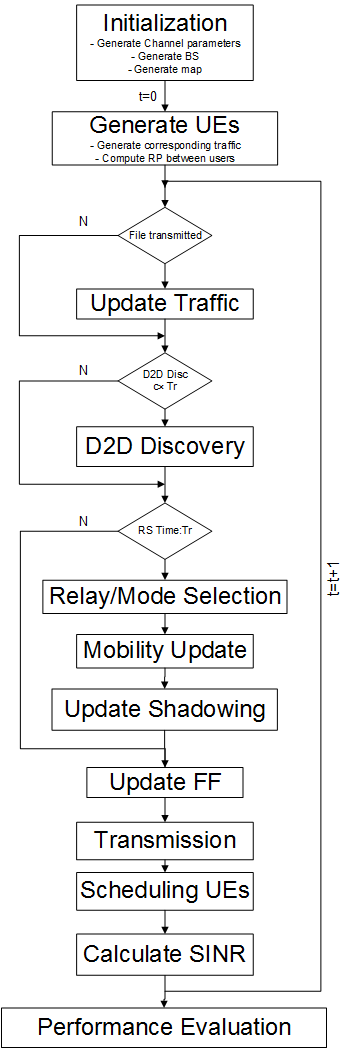}
\caption{Flow Chart of \gls{D2D} relaying procedures in the system level simulator}
\label{fig.FlowChart} 
\end{figure}
\end{centering}

In the following results, the considered reward criteria is the bit-rate which is deduced from the \gls{SINR} based on the \gls{SINR}-throughput mapping described in \ref{sec:SimuSetting}. In addition, the cost criteria is given by the relays' transmission power. Since a power control mechanism is applied to \gls{UL} and \gls{D2D} communications, the relays' transmission power will depend on their channel conditions (i.e. \gls{SINR}). The cost threshold $C_{th}$ that each user should not exceed is equal to $1$ W. 

\vspace{-20pt}
\subsection{Results}
We consider the case of \gls{DL} communications aided by \gls{D2D} relaying. The relay selection algorithm implemented is the one described in \ref{subsec:RS_GCPBVI}. The performance of this scenario is given as function of the number of relays $K$ and the number of users $N$ existing in each macro cell. We show in figure \ref{fig.Simu_SINR_DL} that enabling \gls{D2D} relaying improves the \gls{SINR} of \gls{DL} communications. Moreover, we study the effect of the number of relays $K$ on the performance of \gls{DL} communications. For this aim, we consider $N=5$ users per cell and we plot both throughput metrics (i.e. $Th_1$ and $Th_2$ from equations \ref{eq.Throughput1} and \ref{eq.Throughput2}) as function of the number of relays $K$ (see figure \ref{fig.Simu_Th_f_Relays_DL}). This result illustrates that more the number of relays $K$ increases more the gain achieved by \gls{D2D} relaying is higher. However, one can see that without considering large number of relays, the \gls{D2D} relaying can attain important gain in terms of spectral efficiency of \gls{DL} communications.

\begin{figure}[ptb]
\centering
        \begin{minipage}[b]{0.475\textwidth}
            \centering
            \includegraphics[width=\textwidth]{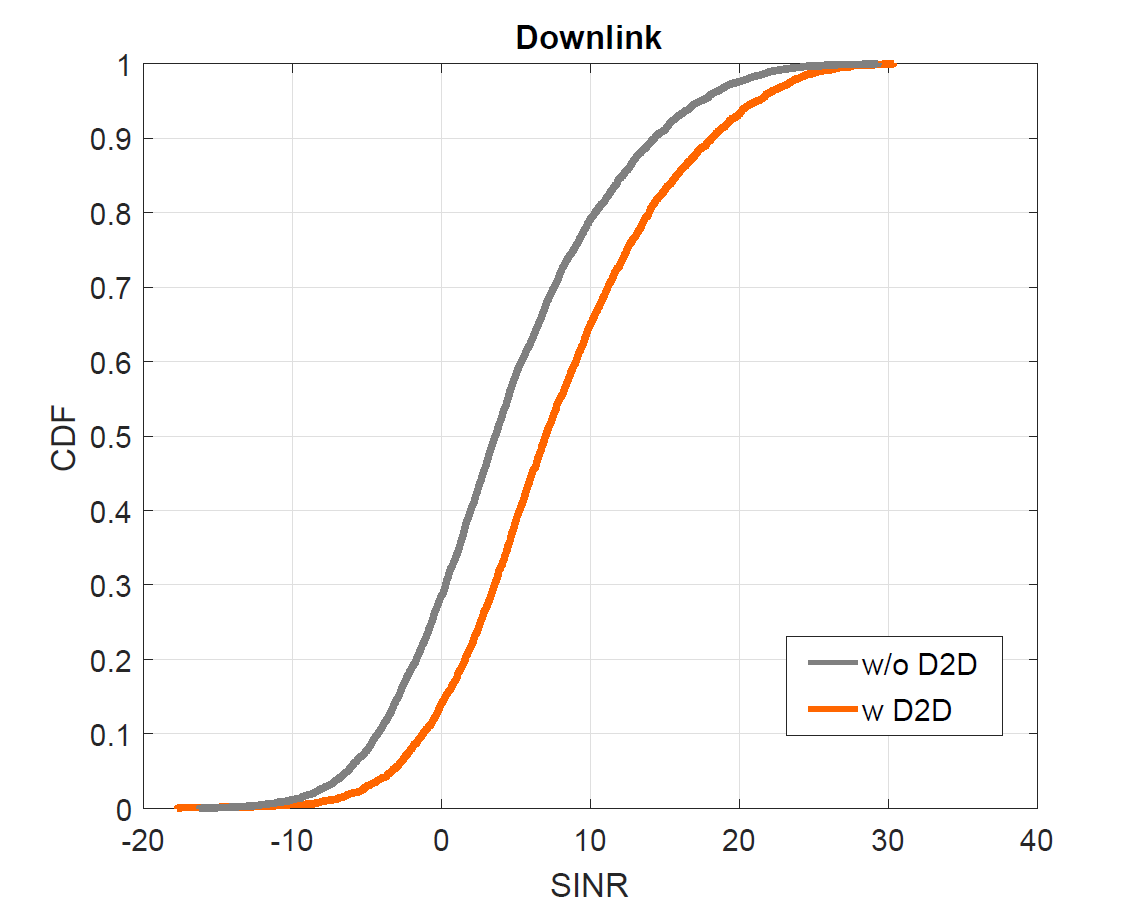}
            \caption[Network2]%
            {{\small \gls{CDF} of the \gls{DL} \gls{SINR} for both scenarios (with and without \gls{D2D} relaying)}}    
            \label{fig.Simu_SINR_DL}
        \end{minipage}
        \hfill
        \begin{minipage}[b]{0.475\textwidth}  
            \centering 
            \includegraphics[width=\textwidth]{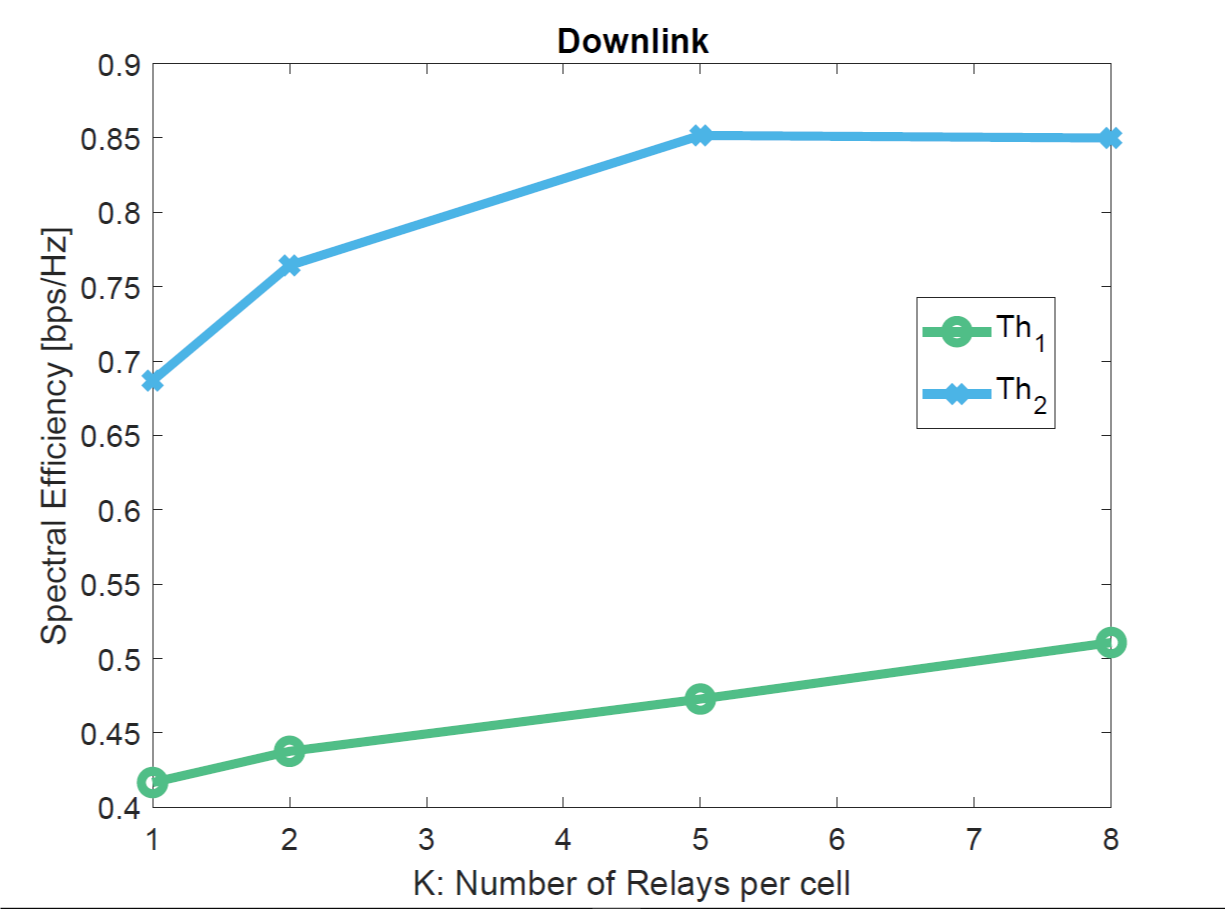}
            \caption[]%
            {{\small The spectral efficiency of \gls{DL} communications as function of $K$}}
            \label{fig.Simu_Th_f_Relays_DL}
\end{minipage}
    \end{figure}    
In addition, we compare between both scenarios with and without \gls{D2D} relaying as function of the number $N$ of users dropped in each cell. We consider a scenario with $K=4$ relays per cell. We show in \ref{fig.Simu_Th_f_UE_DL} that for different values of $N\in\lbrace 2, 4, 8, 10, 12, 14 \rbrace$, the spectral efficiency of \gls{DL} communications increases due to the deployment of our proposed algorithm of relay selection policy. However, in this figure it is not visible for the reader how much the spectral efficiency is increased. Therefore, for different values of $N$, we present in \ref{fig.Simu_relativeTh_f_UE_DL} the relative gain achieved by both $Th_1$ and $Th_2$ (i.e. given in equations \ref{eq.Throughput1} and \ref{eq.Throughput2}) due to the implementation of our proposed algorithm of relay selection. A spectral efficiency enhancement of more than $50\%$ is recorded.

\begin{figure}
\centering
            \begin{centering}
            \includegraphics[width=0.5\textwidth]{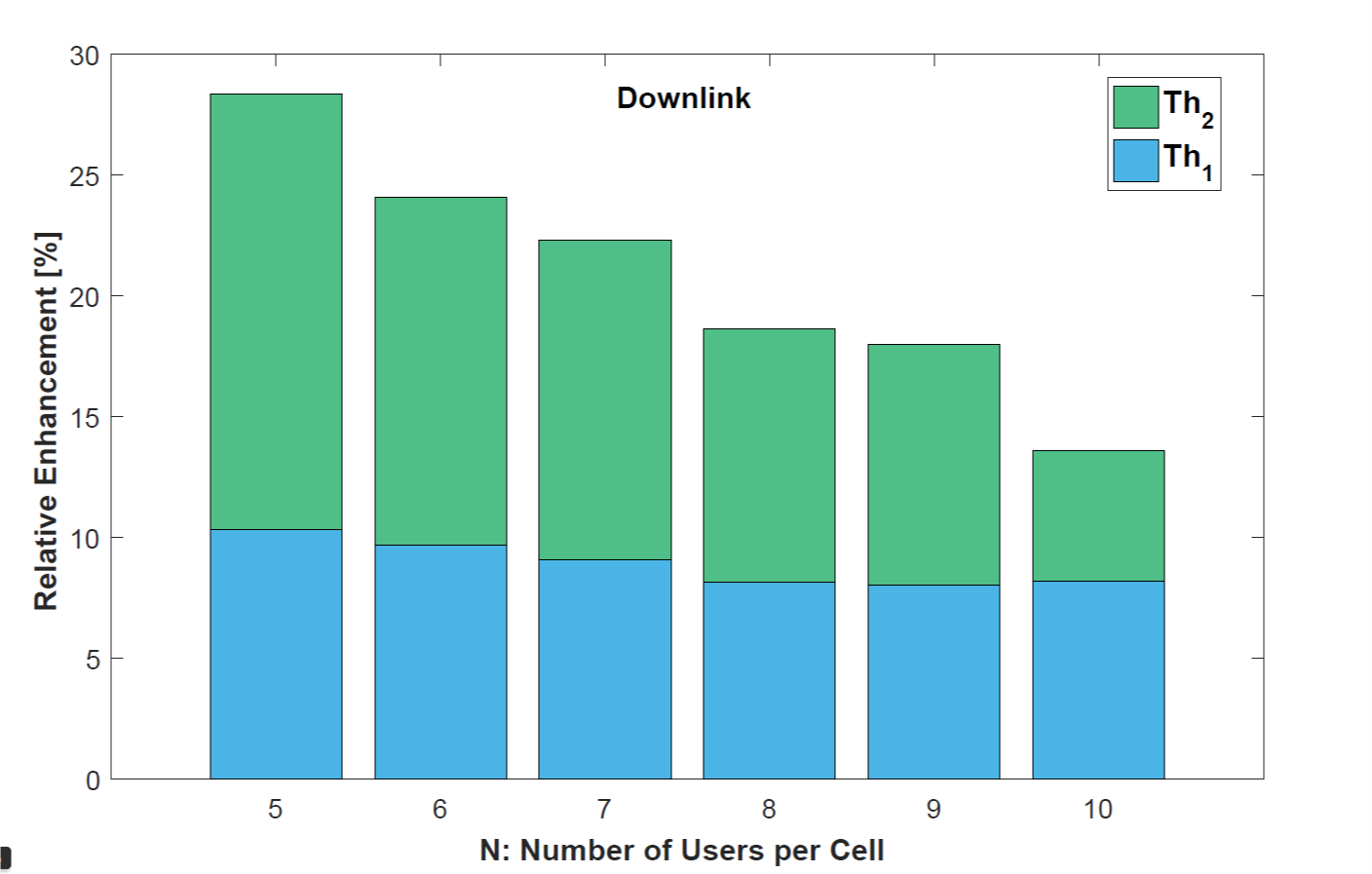}
            \caption{\small The spectral efficiency of \gls{DL} communications as function of the number of users $N$ for both scenarios}    
            \label{fig.Simu_Th_f_UE_DL}
            \end{centering}
        \end{figure}
        
\begin{figure}[H]
\hspace{-20pt}
\centering
            \includegraphics[width=0.5\textwidth]{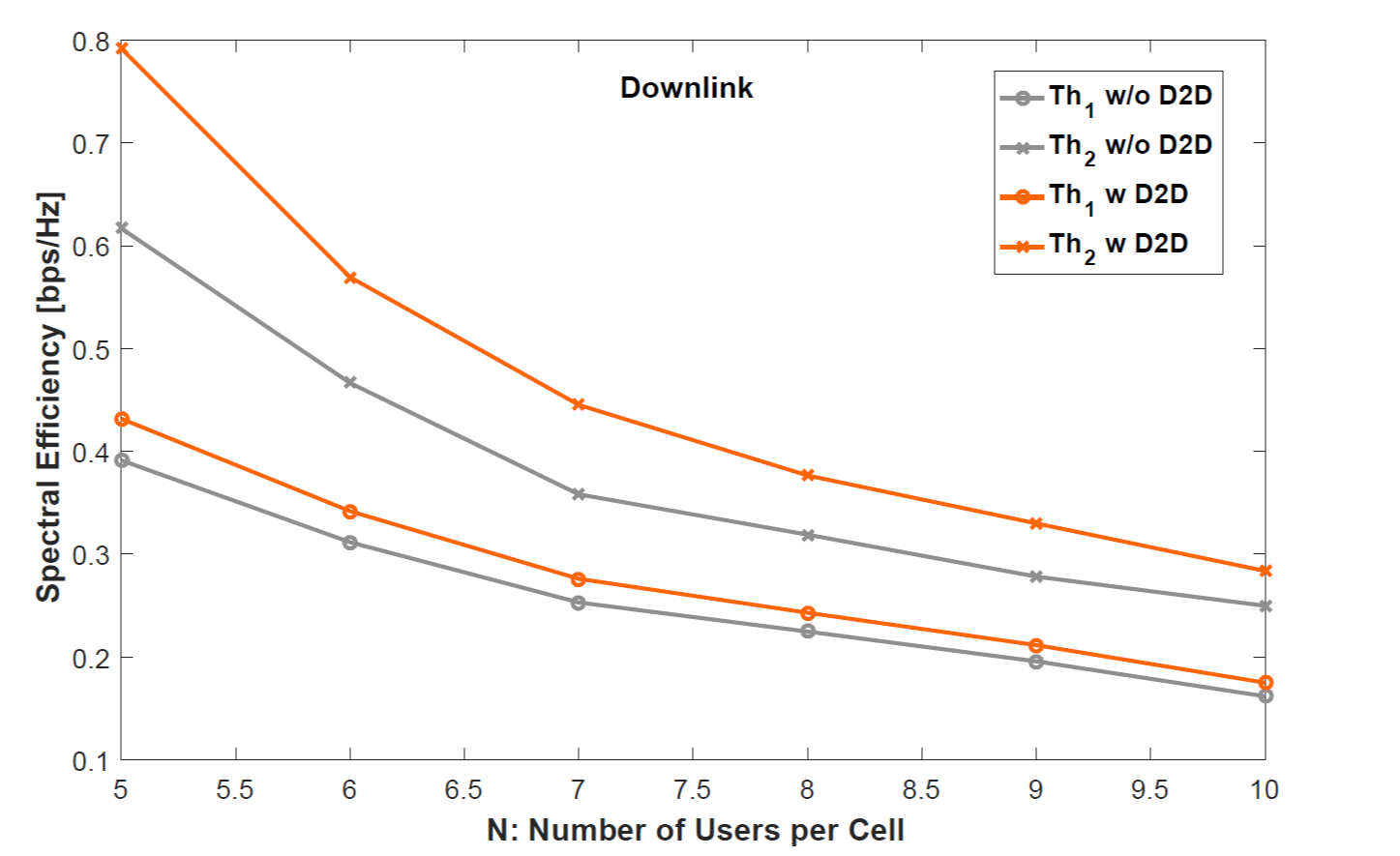}
            \caption{ The relative enhancement of the spectral efficiency of \gls{DL} communications as function of the number of users $N$}
            \label{fig.Simu_relativeTh_f_UE_DL}
        \end{figure}

\section{Conclusion \label{sec:RS_Conclusion}}
Finding the optimal relay selection policy is a challenging problem especially when the mobility of the relays is considered. In this chapter, we have developed a dynamic relays selection strategy that maximizes a certain performance metric of the cellular networks while guaranteeing some cost constraints. A \gls{CPOMDP} problem has been formulated and its complexity is discussed. Thus, a greedy \gls{GCPBVI} algorithm is addressed for achieving a low-complexity and close to optimal solution for the problem. Numerical results show the advantage of such approximation in reducing the problem complexity and prove the performance gain that such solution provides to mobile networks. Our claims are verified by implementing the suggested relay selection scheme in a system-level simulator of cellular networks. 

\appendices

\section{Proof of theorem \ref{th.DensityBound}  \label{proof.DensityBound}}

The following definition and lemma are used for proving theorem \ref{th.DensityBound}.
\begin{definition}
The distance function of relay $i$ with transition probability $\bm{P}_i$ and stationary distribution $\pi_i$ is defined as follows:
\begin{equation}
\label{eq.distFunction}
d_i\left(t\right)= \max \limits_{s\in \mathcal{S}}|| \bm{P}_i^t\left(s,:\right)-\bm{\pi}_i||
\end{equation}
\end{definition}

\begin{lemma}
\label{lem.DistFunc}
From \cite{levin2017markov}, the distance function of relay $i$ verifies the following properties:
\begin{itemize}
\item $d_i\left(t\right)\leq \frac{\lambda_i^{*t}}{\pi_{i,min}}$ with $\lambda_i^*$ the highest eigenvalue of the matrix $\bm{P}_i$ and $\pi_{i,min}$ the lowest component of the stationary distribution corresponding to $\bm{P}_i$.\\
\item $d_i\left(t+a\right)\leq d_i\left(t\right) \forall a \in \mathbb{N}^+$ 
\end{itemize}
\end{lemma}

The density $\epsilon^{\mathcal{B}}$ of a belief set $\mathcal{B}\left(\bm{s},h\right)$ is given by:
\[
 \epsilon^{\mathcal{B}}\left( \mathcal{B}\left(\bm{s},h\right) \right)=\max \limits_{\bm{\tilde{b}}\in \Delta}\min \limits_{\bm{b}\in  \mathcal{B}\left(\bm{s},h\right)} {||\bm{b}-\bm{b}'||}_1
 \]
 \[
 \leq \max \limits_{\bm{\tilde{b}} \in \Delta}\min \limits_{\bm{b}\in  \mathcal{B}\left(\bm{s},h\right)} \sum \limits_{i \in \mathcal{K}} \left[ \sum_{s \in \mathcal{S}} |b^i\left(s \right)- \tilde{b}^{i}\left(s \right)| \right]
 \]
For any reachable belief point $\tilde{\bm{b}}^i$ of relay $i$, there exists an integer $n$ and a state $S_i \in \mathcal{S}$ such that $\tilde{\bm{b}}^i =\bm{P}_i^n\left( S_i,: \right)$. The construction of the belief set $\mathcal{B}\left(\bm{s},h\right)$ induces that ${\bm{b}}^i $ in the equation above can be written as ${\bm{b}}^i =\bm{P}_i^m\left( s_i\left(0\right),: \right)$ with the integer $m \leq h$ and $s_i\left(0\right)$ the initial state of relay $i$. 

We will consider both cases:
\begin{itemize}
\item if $n \leq h$, then the corresponding reachable point $\tilde{\bm{b}}^i$ has been taken into account in the belief set $\mathcal{B}\left(\bm{s},h\right)$, thus:
\begin{equation}
\label{eq.sum_b1}
\sum_{s \in \mathcal{S}} |b^i\left(s \right)- \tilde{b}^{i}\left(s \right) |=0
\end{equation}
\item if $n>h$, the expression $\ \sum\limits_{s \in \mathcal{S}} |b^i\left(s \right)- \tilde{b}^{i}\left(s \right)|$ can be replaced by:
\[
\sum_{s \in \mathcal{S}} |b^i\left(s \right)- \tilde{b}^{i}\left(s \right) |=  \sum_{s \in \mathcal{S}} |P_i^m\left(s_i\left(0\right),s \right)- P_i^n\left(S_i,s \right) |
\]
\[
\resizebox{0.9\hsize}{!}{$\leq  \sum\limits_{s \in \mathcal{S}} |P_i^m\left(s_i\left(0\right),s \right)- \pi_i\left(s \right) |  + \sum\limits_{s \in \mathcal{S}} |P_i^n\left(s_i\left(0\right),s \right)- \pi_i\left(s \right)| $}
\]
\[\leq d_i\left(m\right) + d_i\left(n\right) 
\]
Based on lemma \ref{lem.DistFunc} and since $n>h$, then:
\begin{equation}
\label{eq.sum_b2}
\sum_{s \in \mathcal{S}} |b^i\left(s \right)- \tilde{b}^{i}\left(s \right) | \leq  2d_i\left(h\right) \leq 2\frac{\lambda_i^{*h}}{\pi_{i,min}}
\end{equation}
\end{itemize}
From equations (\ref{eq.sum_b1}) and (\ref{eq.sum_b2}) and summing over all the relays $i \in \mathcal{K}$ we prove the theorem \ref{th.DensityBound}.

\section{Proof of theorem \ref{th.Belief_set_eps} \label{proof.Belief_set_eps}}

The studied \gls{CPOMDP} of horizon $T$ and an initial state $\bm{s}_0$ achieves an error bound of $\epsilon$ when $\eta_T^r$ and $\eta_T^c$ (see equation (\ref{eq.PBVI_Error})) are lower than $\epsilon$. Limiting the belief set to $\mathcal{B}\left(\bm{s}_0,h\right)$ generates the following errors on the value functions:
\begin{equation}
\label{eq.eta_r_bound}
\eta_T^r \leq \frac{\left (R_{max}-R_{min}\right)}{{\left( 1 - \gamma \right)}^2}\sum\limits_{i \in \mathcal{K}}2\frac{\lambda_i^{*h}}{\pi_{i,min}} \leq 2K\frac{\left (R_{max}-R_{min}\right)\lambda^{*h}}{{\pi_{min}\left( 1 - \gamma \right)}^2}
\end{equation}
and
\begin{equation}
\label{eq.eta_c_bound}
\eta_T^c \leq 2K\frac{\left (C_{max}-C_{min}\right)\lambda^{*h}}{{\pi_{min}\left( 1 - \gamma \right)}^2}
\end{equation}
with $\lambda^*=\max\limits_{i \in \mathcal{K}}\lambda_i^*$ ; $\pi_{min}^*=\min\limits_{i \in \mathcal{K}}\pi_{i,min}^*$.\\

Therefore, limiting $\eta_T^r$ and $\eta_T^c$ to $\epsilon$ corresponds to choosing the parameter $h$ of $\mathcal{B}\left(\bm{s}_0,h\right)$ in such a way that the expressions in equations (\ref{eq.eta_r_bound}) and (\ref{eq.eta_c_bound}) are bounded by $\epsilon$. Hence, theorem \ref{th.Belief_set_eps} is deduced.

Note that for a discount factor $\gamma=1$, the sum over the horizon $T$ gives the following upper bounds of the errors $\eta_T^r$ and $\eta_T^c$:
\[
\eta_T^r \leq 2KT\frac{\left( R_{max}-R_{min}\right)\lambda^{*h}}{\pi_{min}}
\]
and
\[
 \eta_T^c \leq 2KT\frac{\left( C_{max}-C_{min}\right)\lambda^{*h}}{\pi_{min}}
\]
Thus, for $\gamma=1$, the expressions of $f_r\left(\epsilon\right)$ and $f_c\left(\epsilon\right)$ in equation (\ref{eq.Belief_set_eps}) of theorem \ref{th.Belief_set_eps} are given by:
\[
f_r\left(\epsilon \right)=\log\left( \frac{\epsilon \pi_{min}}{2KT\left(R_{max}-R_{min}\right)} \right)
\]
and
\[ f_c\left(\epsilon \right)=\log\left( \frac{\epsilon \pi_{min}}{2KT\left(C_{max}-C_{min}\right)} \right)
\]

\section{Proof of theorem \ref{th.SubmodQ} \label{proof.SubmodQ}}

Please note that we consider the reward $Q$-function as an example. By analogy, we can deduce the same property for the cost $Q$-function. In the following, we define the discrete derivativeof a relay $e \in\mathcal{K}$ knowing that an action set $\hat{a}_M$ (corresponding to action $\bm{a}_M$) is taken:
\begin{equation}
\label{eq.DiscDeriv}
\Delta_{Q^{\pi}_t}\left( e| \hat{a}_M\right):=Q_t^{\pi}\left(\bm{b}^t,\hat{a}_M \cup e \right)-Q_t^{\pi}\left(\bm{b}^t,\hat{a}_M\right)
\end{equation}
\[
=\rho\left(\bm{b}^t, \hat{a}_M \cup e \right)-\rho\left(\bm{b}^t, \hat{a}_M \right)
\]
\[+\sum \limits _ {k=t+1}^{T}\left[G_k^{\pi}\left(\bm{b}^t,\hat{a}_M\cup e\right) -G_k^{\pi}\left(\bm{b}^t,\hat{a}_M\right) \right]
\]
$Q_t^{\pi}\left( \bm{b},\bm{a} \right)$ is non-negative, monotone and submodular in $\bm{a}$ if the discrete derivative of the $Q$-function $\Delta_{Q^{\pi}_t}$ verifies the following:
\begin{equation}
\label{eq.properties1}
\Delta_{Q^{\pi}_t}\left( e| \hat{a}_M\right) \geq 0
\end{equation}
and
\begin{equation}
\label{eq.properties2}
\Delta_{Q^{\pi}_t}\left( e| \hat{a}_M\right) \geq \Delta_{Q^{\pi}_t}\left( e| \hat{a}_M\cup \hat{a}_N\right)
\end{equation}
Considering the reward model (\ref{eq.TotalReward}), the discrete derivative $\Delta_{Q^{\pi}_t}\left( e| \hat{a}_M\right)$ is computed as follows:
\[
\Delta_{Q^{\pi}_t}\left( e| \hat{a}_M\right)=\rho\left(\bm{b}^t, \hat{a}_M \cup e \right)-\rho\left(\bm{b}, \hat{a}_M \right)
\]
\[+\sum \limits _ {k=t+1}^{T}\left[G_k^{\pi}\left(\bm{b}^t,\hat{a}_M\cup e\right) -G_k^{\pi}\left(\bm{b}^t,\hat{a}_M\right) \right]
\]
\[
=\sum\limits_{s\in\mathcal{S}}b_t^e\left(s\right)r_e\left(s\right)+\sum \limits _ {k=t+1}^{T}\left[G_k^{\pi}\left(\bm{b}^t,\hat{a}_M\cup e\right) -G_k^{\pi}\left(\bm{b}^t,\hat{a}_M\right) \right]
\]
\[
=\sum\limits_{s\in\mathcal{S}}b_t^e\left(s\right)r_e\left(s\right)
\]
\[
+\sum \limits _ {k=t+1}^{T}\gamma^{k}\sum \limits _{\bm{z}^{t:k}}\left[ P\left(\bm{z}^{t:k}|\bm{b}^t,\hat{a}_M\cup e, \pi\right)\rho\left(\bm{b}_{\bm{z}^{t:k}}^{\bm{a}^{\pi}},\bm{a}^{\pi}\right)\right]
\]
\[-\sum \limits _ {k=t+1}^{T}\gamma^{k}\sum \limits _{\bm{z}^{t:k}}\left[- P\left(\bm{z}^{t:k}|\bm{b}^t,\hat{a}_M, \pi\right)\rho\left(\bm{b}_{\bm{z}^{t:k}}^{\bm{a}^{\pi}},\bm{a}^{\pi} \right)\right]
\]
Since the total reward model, given by equation (\ref{eq.TotalReward}), is equal to the sum of the reward of each selected relay, then the difference $G_k^{\pi}\left(\bm{b},\hat{a}_M\cup e\right) -G_k^{\pi}\left(\bm{b},\hat{a}_M\right)$ will be limited to the reward of relay $e$ as follows:
\[
\Delta_{Q^{\pi}_t}\left( e| \hat{a}_M\right)=\sum\limits_{s\in\mathcal{S}}b_t^e\left(s\right)r_e\left(s\right)
\]
\[
+\sum \limits _ {k=t+1}^{T}\gamma^{k}\sum \limits _{{z}_e^{t:k}}\left[ P\left({z}_e^{t:k}|\bm{b}_t^e,\hat{a}_M\cup e, \pi\right)\rho\left(\bm{b}_{{z}_e^{t:k}}^{e},\bm{a}^{\pi}\right)\right]
\]
\[
-\sum \limits _ {k=t+1}^{T}\gamma^{k}\sum \limits _{{z}_e^{t:k}}\left[ P\left({z}_e^{t:k}|\bm{b}_t^e,\hat{a}_M, \pi\right)\rho\left(\bm{b}_{{z}_e^{t:k}}^{e},\bm{a}^{\pi} \right)\right]
\]
where $z_e^{{t:k}}$ is the observation of the state of relay $e$ between $t$ and $k$ epochs, $\bm{b}_t^e$ the belief vector of relay $e$ at epoch $t$, $\bm{b}_{{z}_e^{t:k}}^{e}$ the belief vector of relay $e$ after an observation ${z}_e^{t:k}$ between $t$ and $k$ epochs.
In order to prove the properties (\ref{eq.properties1}) and  (\ref{eq.properties2}) and deduce by that the submodularity of the $Q$-function, we compute $\Delta_{Q^{\pi}_t}\left( e| \hat{a}_M\right)$ for the following two cases:

\begin{itemize}
\item {{Relay $e$ has not been chosen by the policy $\pi$ between the epochs $t+1$ and $T$:}} in this case, the belief vector $\bm{b}_n^{e}$ of relay $e$ (for $n=t+1,...,T$) is the same in both expressions $G_k^{\pi}\left(\bm{b},\hat{a}_M\right)$ and $G_k^{\pi}\left(\bm{b},\hat{a}_M\cup e\right)$. Hence,
\[
G_k^{\pi}\left(\bm{b},\hat{a}_M\cup e\right) -G_k^{\pi}\left(\bm{b},\hat{a}_M\right) =0 \text{ for all $k =t+1,...,T$}
\]
\[
\Rightarrow \Delta_{Q^{\pi}_t}\left( e| \hat{a}_M\right)=\sum\limits_{s\in\mathcal{S}}b_t^e\left(s\right)r_e\left(s\right)
\]
Thus, in this case the discrete derivative of the $Q$-function verifies both properties (\ref{eq.properties1}) and (\ref{eq.properties2}). 

\item{Relay $e$ has been chosen by the policy $\pi$ at epoch $t^*$:} in this case, the decision maker observes the state of relay $e$: $z_e^{t^*}$. Thus,
\[
\Delta_{Q^{\pi}_t}\left( e| \hat{a}_M\right)= \sum\limits_{s\in\mathcal{S}}b_t^e\left(s\right)r_e\left(s\right) 
\]
\[+ \sum \limits_ {k= t+1}^{t^*}\gamma^{k}\left[G_k^{\pi}\left(\bm{b},\hat{a}_M\cup e\right) -G_k^{\pi}\left(\bm{b},\hat{a}_M\right) \right]
\]
\[ +\gamma^{t^*+1}\left[G_{t^*+1}^{\pi}\left(\bm{b},\hat{a}_M\cup e\right) -G_{t^*+1}k^{\pi}\left(\bm{b},\hat{a}_M\right) \right]
\]
\[+ \sum \limits_ {k=t^*+2}^T\gamma^{k}\left[G_k^{\pi}\left(\bm{b},\hat{a}_M\cup e\right) -G_k^{\pi}\left(\bm{b},\hat{a}_M\right) \right]
\]
The second term $\sum \limits_ {k= t+1}^{t^*}\gamma^{k}\left[G_k^{\pi}\left(\bm{b},\hat{a}_M\cup e\right) -G_k^{\pi}\left(\bm{b},\hat{a}_M\right) \right]=0$ since relay $e$ has not been chosen between $t+1$ and $t^*$ epochs. In addition, the fourth term $\sum \limits_ {k= t^*+2}^T\gamma^{k}\left[G_k^{\pi}\left(\bm{b},\hat{a}_M\cup e\right) -G_k^{\pi}\left(\bm{b},\hat{a}_M\right) \right]=0$ since relay $e$ has been chosen at epoch $t^*$. Therefore, the policy $\pi$ will have the same value for both $G_k^{\pi}\left(\bm{b},\hat{a}_M\cup e\right)$ and $G_k^{\pi}\left(\bm{b},\hat{a}_M\right)$. Thus,
\[
\Delta_{Q^{\pi}_t}\left( e| \hat{a}_M\right)- \sum\limits_{s\in\mathcal{S}}b_t^e\left(s\right)r_e\left(s\right) 
\]
\[= \gamma^{t^*+1}\left[G_{t^*+1}^{\pi}\left(\bm{b},\hat{a}_M\cup e\right) -G_{t^*+1}k^{\pi}\left(\bm{b},\hat{a}_M\right) \right]
\]
\[\resizebox{0.9\hsize}{!}{$= \gamma^{t^*+1}\sum \limits_{{z}_e^{t:k}}\left[ P\left({z}_e^{t:k}|\bm{b}_t^e,\hat{a}_M\cup e, \pi\right)\rho\left(\bm{b}_{{z}_e^{t:k}}^{e},\bm{a}^{\pi}\right)- P\left({z}_e^{t:k}|\bm{b}_t^e,\hat{a}_M, \pi\right)\rho\left(\bm{b}_{{z}_e^{t:k}}^{e},\bm{a}^{\pi} \right)\right]$}
\]
\[\resizebox{0.9\hsize}{!}{$= \gamma^{t^*+1}\sum \limits_{{z}_e^{t^*}}\left[ \sum \limits_{{z}_e^{t}}P\left({z}_e^{t^*}|{z}_e^{t}\right) P\left({z}_e^{t}|\bm{b}_t^e,\hat{a}_M\cup e, \pi\right)- P\left({z}_e^{t^*}|\bm{b}_t^e,\hat{a}_M, \pi\right)\right]\sum\limits_{s'\in\mathcal{S}}\bm{P}_e\left(z_e^{t^*},s'\right)r_e\left(s'\right)$}
\]
\[\resizebox{0.9\hsize}{!}{$= \gamma^{t^*+1}\sum \limits_{{z}_e^{t^*}}\left[\sum \limits_{{z}_e^{t}}P\left({z}_e^{t^*}|{z}_e^{t}\right){\left( \bm{P}_e.\bm{b}_t^e\right)}\left( {{z}_e^{t}}\right) - {\left(\bm{P}_e^{t^*-t+1}.\bm{b}_t^e\right)}\left({z_e^{t^*}}\right)\right]\sum\limits_{s'\in\mathcal{S}}\bm{P}_e\left(z_e^{t^*},s'\right)r_e\left(s'\right)$}
\]
\[\resizebox{0.9\hsize}{!}{$= \gamma^{t^*+1}\sum \limits_{{z}_e^{t^*}}\left[{\left(\bm{P}_e\left(z_e^t,:\right).\bm{P}^{t^*-t}_e\right)}\left({z_e^{t^*}}\right).{\left(\bm{P}_e.\bm{b}_t^e\right)}\left( {z}_e^{t}\right)- {\left(\bm{P}_e^{t^*-t+1}.\bm{b}_t^e\right)}\left({z_e^{t^*}}\right)\right]\sum\limits_{s'\in\mathcal{S}}\bm{P}_e\left(z_e^{t^*},s'\right)r_e\left(s'\right)$}
\]
\[\resizebox{0.9\hsize}{!}{$= \gamma^{t^*+1}\sum \limits_{{z}_e^{t^*}}\left[{\left(\bm{P}^{t^*-t+2}_e.\bm{b}_t^e\right)}\left( z_e^{t^*}\right)- {\left(\bm{P}_e^{t^*-t+1}.\bm{b}_t^e\right)}\left(z_e^{t^*}\right)\right]\sum\limits_{s'\in\mathcal{S}}\bm{P}_e\left(z_e^{t^*},s'\right)r_e\left(s'\right)$}
\]
\[\resizebox{0.9\hsize}{!}{$= \gamma^{t^*+1}\sum\limits_{s'\in\mathcal{S}}\left[{\left(\bm{P}^{t^*-t+3}_e.\bm{b}_t^e\right)}\left( s'\right)- {\left(\bm{P}_e^{t^*-t+2}.\bm{b}_t^e\right)}\left({s'}\right)\right]r_e\left(s'\right)$}
\]
Thus,
\[
\Delta_{Q^{\pi}_t}\left( e| \hat{a}_M\right)=
\]
\[\resizebox{1\hsize}{!}{$=\sum\limits_{s\in\mathcal{S}}r_e\left(s\right)\left[ b_t^e\left(s\right)+ \gamma^{t^*+1}{\left(\bm{P}^{t^*-t+3}_e.\bm{b}_t^e\right)}\left( s\right)-  \gamma^{t^*+1}{\left(\bm{P}_e^{t^*-t+2}.\bm{b}_t^e\right)}\left( s\right)\right]$}
\]
\begin{lemma}
\label{lemma.BeliefInequality}
For a relay $e$, the two following statements are verified by induction:
\begin{itemize}
\item if $\left(\bm{P}_e.\bm{b}_t^e\right)\left(s \right)\geq \bm{b}_t^e \left(s \right) $ then $\left(\bm{P}^n_e.\bm{b}_t^e\right)\left(s \right)\geq \left(\bm{P}^{n-1}_e.\bm{b}_t^e\right)\left(s \right) \forall n \in \mathbb{Z} \geq 2$ 
\item if $\left(\bm{P}_e.\bm{b}_t^e\right)\left(s \right)\leq \bm{b}_t^e \left(s \right) $ then $\left(\bm{P}^n_e.\bm{b}_t^e\right)\left(s \right)\leq \left(\bm{P}^{n-1}_e.\bm{b}_t^e\right)\left(s \right) \forall n \in \mathbb{Z} \geq 2$ 
\end{itemize}
\end{lemma}

Based on the lemma \ref{lemma.BeliefInequality}, we deduce that, in this case, the discrete derivative $\Delta_{Q^{\pi}_t}\left( e| \hat{a}_M\right)$ verifies the properties (\ref{eq.properties1}) and (\ref{eq.properties2}). 
\end{itemize}
Since the discrete derivative $\Delta_{Q^{\pi}_t}\left( e| \hat{a}_M\right)$ verifies the properties  (\ref{eq.properties1}) and (\ref{eq.properties2}) for the two possible cases studied above; theorem \ref{th.SubmodQ} is deduced.

\section{Proof of theorem \ref {th.Greedy_Bound} \label{proof.Greedy_Bound}}

The proof is done by induction. We start by verifying equation (\ref{eq.Greedy_Bound}) for $t=1$. We recall that:
\[
V^{r,G}_1\left(\bm{b}\right)=greedy-\argmax \rho_r\left(\bm{b},\bm{a}\right) \text{ s.t. } \rho_c\left(\bm{b},\bm{a}\right) \leq C_{th}
\]
\begin{center}
and
\end{center}
\[
V^{r,B}_1\left(\bm{b}\right)=\argmax \rho_r\left(\bm{b},\bm{a}\right) \text{ s.t. } \rho_c\left(\bm{b},\bm{a}\right) \leq C_{th}
\] 
Based on the results given in \cite{Sviridenko2004ANO}, the greedy maximization algorithm achieves a $1-\frac{1}{e}$ approximation of the exact solution. Thus, $V^{r,G}_1\geq \left(1-\frac{1}{e} \right)V^{r,B}_1$. Now, we assume that equation (\ref{eq.Greedy_Bound}) is verified for $t-1$ and prove that it remains verified for $t$. 
\[
V^{r,G}_{t-1}\left( \bm{b}\right)\geq \left(1-\frac{1}{e}\right)^{2t-2}V_{t-1}^{r,B}\left( \bm{b}\right)
\]
\[
\rho_r\left( \bm{b},\bm{a}\right)+\gamma \sum_{\bm{z}\in\mathcal{Z}}P\left(\bm{z}|\bm{a},\bm{b}\right)V_{t-1}^{r,G}\left({\bm{b}}_{\bm{z}}^{\bm{a}}\right)
\]
\[\geq \left(1-\frac{1}{e}\right)^{2t-2}\left[ \rho_r\left( \bm{b},\bm{a}\right)+\gamma \sum_{\bm{z}\in\mathcal{Z}}P\left(\bm{z}|\bm{a},\bm{b}\right)V_{t-1}^{r,B}\left({\bm{b}}_{\bm{z}}^{\bm{a}}\right)\right]
\]
Based on the definition of $Q_t^r$ function (i.e. equation (\ref{eq.Qfunction_r})), we deduce:
\begin{equation}
\label{eq.QBefore}
Q_{t}^{r,G}\left( \bm{b},\bm{a}\right) \geq \left(1-\frac{1}{e}\right)^{2t-2} Q_{t}^{r,B}\left( \bm{b},\bm{a}\right) \forall \bm{a} \in \mathcal{A}
\end{equation}
We use the following notations where all the maximization are done under cost constraint:
\begin{itemize}
\item $\bm{a}_{Q^G}^G=greedy-\argmax Q_{t}^{r,G}\left( \bm{b},\bm{a}\right)$
\item $\bm{a}_{Q^G}^*=\argmax Q_{t}^{r,G}\left( \bm{b},\bm{a}\right)$
\item $\bm{a}_{Q^B}^G=greedy-\argmax Q_{t}^{r,B}\left( \bm{b},\bm{a}\right)$ 
\item $\bm{a}_{Q^B}^*=\argmax Q_{t}^{r,B}\left( \bm{b},\bm{a}\right)$
\end{itemize}
Thus,
\[
V_t^{r,G}\left( \bm{b}\right)=greedy-\argmax Q_{t}^{r,G}\left( \bm{b},\bm{a}\right)=Q_{t}^{r,G}\left( \bm{b},\bm{a}_{Q^G}^G\right)
\]
Given that $Q_t^{r,\pi}$-function is submodular and constrained to a modular $Q_t^{c,\pi}$, then results shown in \cite{Sviridenko2004ANO} gives:
\[
V_t^{r,G}\left( \bm{b}\right)=Q_{t}^{r,G}\left( \bm{b},\bm{a}_{Q^G}^G\right) \geq \left(1-\frac{1}{e}\right)Q_{t}^{r,G}\left( \bm{b},\bm{a}_{Q^G}^*\right) 
\]
From the definition of $\bm{a}_{Q^G}^*$ we know that $Q_{t}^{r,G}\left( \bm{b},\bm{a}_{Q^G}^*\right) \geq Q_{t}^{r,G}\left( \bm{b},\bm{a}\right) $ for all $\bm{a} \in \mathcal{A}$ including $\bm{a}_{Q^B}^G$. Therefore,
\[
V_t^{r,G}\left( \bm{b}\right) \geq \left(1-\frac{1}{e}\right)Q_{t}^{r,G}\left( \bm{b},\bm{a}_{Q^B}^G\right) 
\]
For equation (\ref{eq.QBefore}), $Q_{t}^{r,G}\left( \bm{b},\bm{a}_{Q^B}^G\right) \geq \left(1-\frac{1}{e}\right)^{2t-2}Q_{t}^{r,B}\left( \bm{b},\bm{a}_{Q^B}^G\right) $. Thus:
\[
V_t^{r,G}\left( \bm{b}\right) \geq \left(1-\frac{1}{e}\right)^{2t-1}Q_{t}^{r,B}\left( \bm{b},\bm{a}_{Q^B}^G\right) 
\]
Given that $Q_t^{r,\pi}$ function is submodular and constrained to a modular $Q_t^{c,\pi}$, then as shown in \cite{Sviridenko2004ANO}:
\[
Q_{t}^{r,B}\left( \bm{b},\bm{a}_{Q^B}^G\right)  \geq \left(1-\frac{1}{e} \right)Q_{t}^{r,B}\left( \bm{b},\bm{a}_{Q^B}^*\right)\]
We deduce that:
\[
V_t^{r,G}\left( \bm{b}\right) \geq  \left(1-\frac{1}{e}\right)^{2t}Q_{t}^{r,B}\left( \bm{b},\bm{a}_{Q^B}^*\right) =V_t^{r,B}\left( \bm{b} \right)
\]
Therefore, theorem \ref{th.Greedy_Bound} is deduced by induction.

\ifCLASSOPTIONcaptionsoff
  \newpage
\fi

\ifCLASSOPTIONcaptionsoff
  \newpage
\fi


\section{Simulation Settings \label{sec:SimuSetting}}
We have developed a \gls{LTE} system level simulator that considers \gls{DL}, \gls{UL} and \gls{D2D} features. This simulator is mainly used for evaluating the performance of \gls{D2D} communications. In particular, we expose the results of  implementing the proposed relay selection scheme into the developed system level simulator. In this section, we describe the different characteristics of this $C++$ system level simulator.

\subsection*{Layout}
In the \gls{LTE} system level simulator, we assume a \textbf{carrier frequency} of $2$ GHz and a \textbf{bandwidth} of $10$ MHz. We consider an \textbf{hexagonal} grid with $7$ tri-sector macro sites. $21$ hexagonal cells are formed as illustrated in figure \ref{fig.layout}. \textbf{Wrap-around} is deployed to avoid border effects and all the cells are supposed \textbf{synchronized}.

Two \textbf{urban macro} scenarios are considered with an \gls{ISD} of $500$ m: (i) \textbf{ option $1$}: all the users are outdoor and (ii)\textbf{ option $2$}: with one RRH/Indoor Hotspot per cell (see section A.2.1.1.5 in \cite{3GPP_TR36814} for more details about RRH/Indoor Hotzone ). Unless specified otherwise, we implement the layout parameters specified for \gls{3GPP} case $1$ in table A.2.1.1-1 of \cite{3GPP_TR36814}.
\begin{figure}
\centering
\captionsetup{justification=centering}
\includegraphics[scale=0.5]{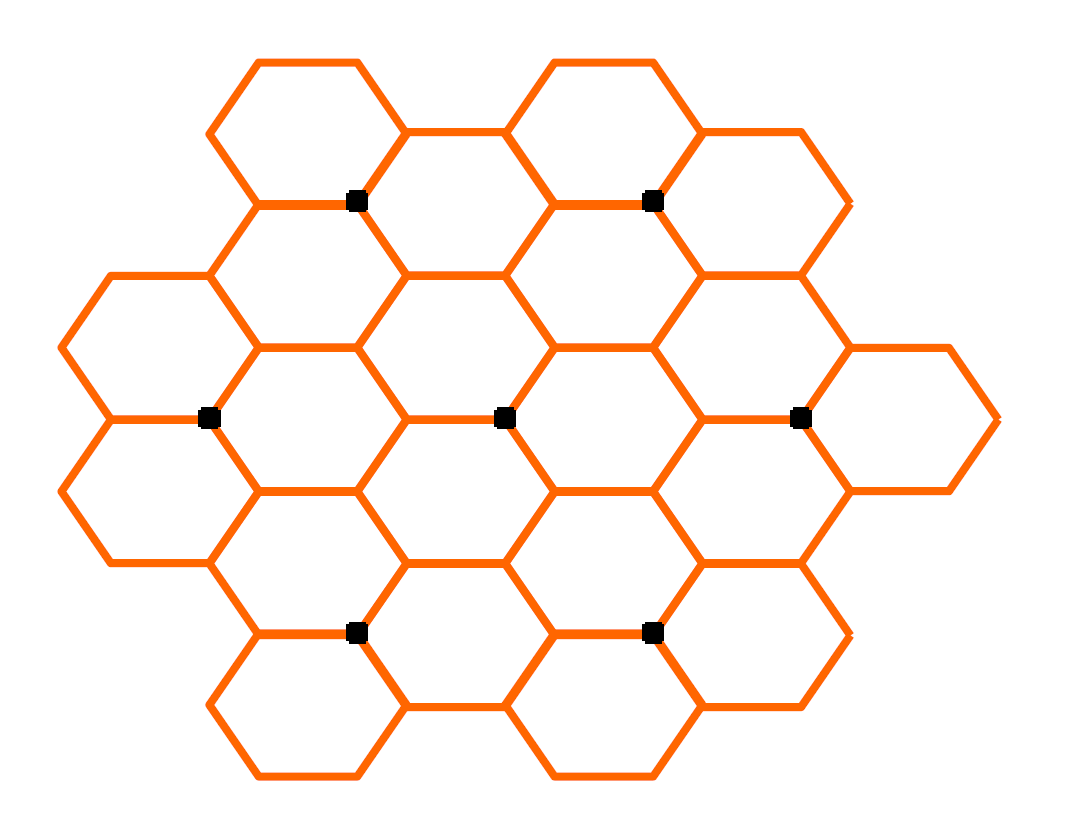}
\caption{The hexagonal grid with 7 tri-sector macro sites}
\label{fig.layout} 
\end{figure}

\subsection*{Users drop \label{subsec:UEDrop}}
The details concerning the users' drop considered in the simulator are given in table A.2.1.1-1 of \cite{3GPP_ProSe}. We give a brief description in the sequel. For layout option $1$, all the users are outdoor and they are randomly and uniformly dropped throughout the macro geographical area. For layout option $2$, $\sfrac{2}{3}$ of the users are randomly and uniformly dropped within the indoor building of each cell and the remaining fraction $\sfrac{1}{3}$ of users are randomly and uniformly dropped throughout the macro geographical area. We consider $20\%$ of the users outdoor and $80\%$ indoor. This can be guaranteed by considering some of the users that are dropped outdoor as virtual indoor users. We assume the following constraints on the building and users drops: $3$ m as the minimum distance between two users, $35$m as the minimum distance between a user and the \gls{BS}, $100$ m as the minimum distance between the building center and the \gls{BS}. In order to illustrate these two layouts, an example for $150$ \gls{UE}s is given in figure \ref{fig.layout_L1_L2}.

\begin{figure}[H]
\hspace{-20pt}
        \begin{minipage}[b]{0.475\textwidth}
            \centering
            \includegraphics[scale=0.4]{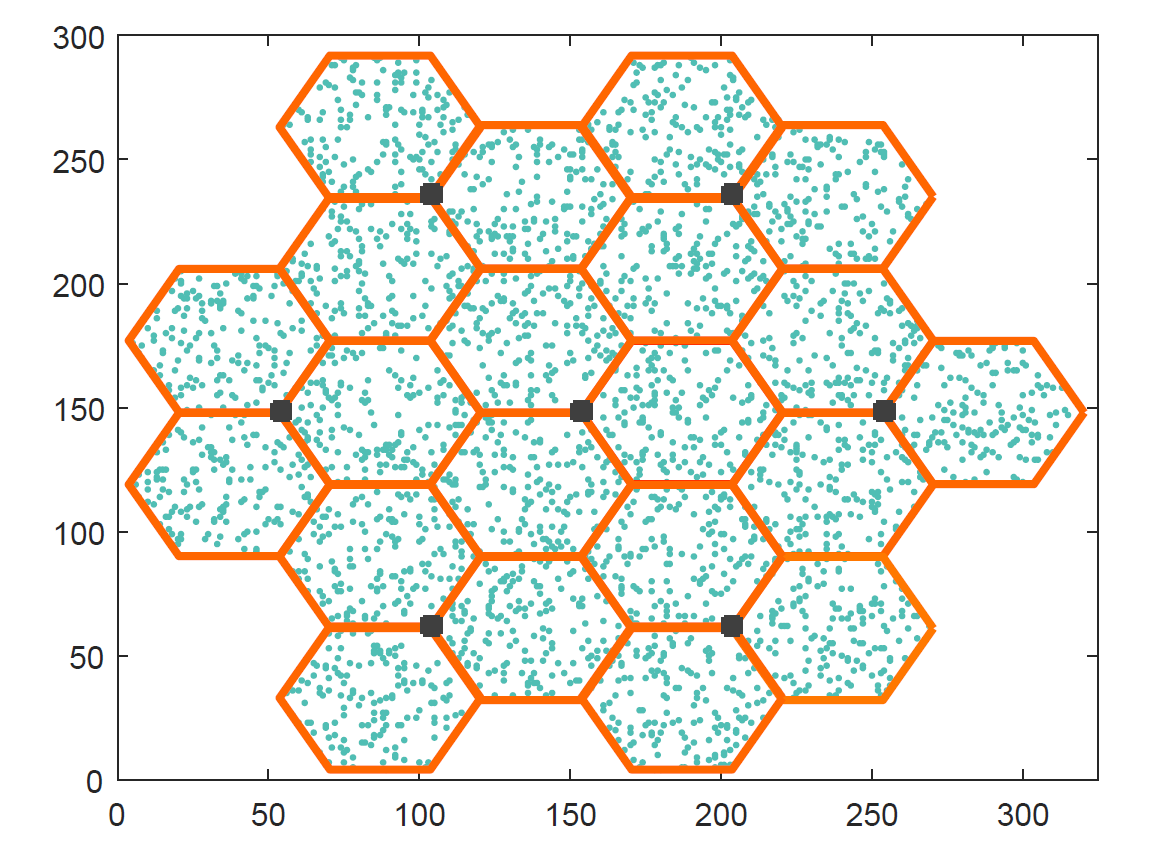}
            \caption[Network2]%
            {{\small Layout option 1: all users outdoor and randomly and uniformly distributed}}    
            \label{fig.layout_L1}
        \end{minipage}
        \hfill
        \begin{minipage}[b]{0.475\textwidth}  
            \centering 
            \includegraphics[scale=0.4]{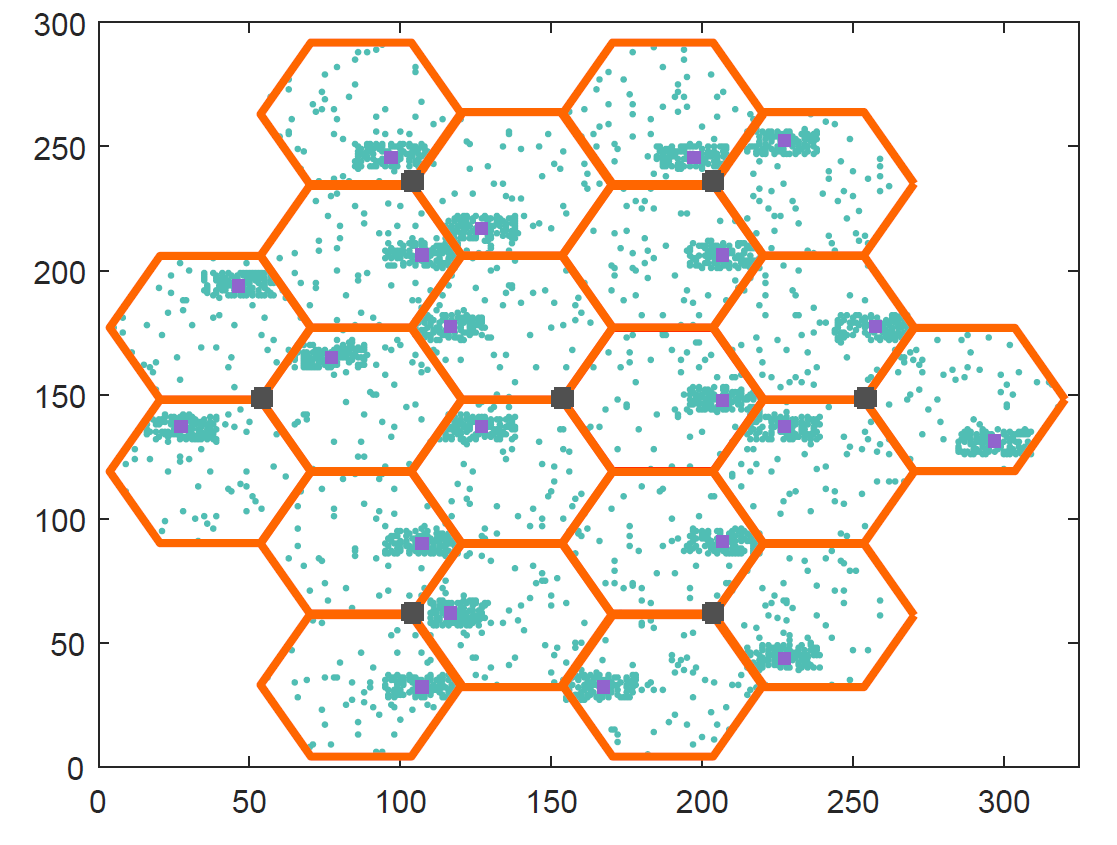}
            \caption[]%
            {{\small Layout option 2: an indoor rectangle building in each cell with $\sfrac{2}{3}$ of the users inside}}
            \label{fig.layout_L2}
\end{minipage}
\caption{The layouts of users drop that are implemented in the system-level simulator}
       \label{fig.layout_L1_L2} 
    \end{figure}
\subsection*{Channel Model}
For cellular communications, we consider the \textbf{antenna pattern} (horizontal and vertical) specified  for \gls{3GPP} case 1 and detailed in table A.2.1.1-2 of  \cite{3GPP_TR36814}.  An \gls{BS} \textbf{antenna gain} of $14$ dBi and an antenna tilt of $15 ^{\circ}$ characterize the \gls{BS}. The user is assumed to be equipped with omni-directional antennas. The \gls{BS} \textbf{transmission power} is $46$ dBm (almost $40$W) and the users' maximum transmission power, i.e. on both \gls{UL} and \gls{D2D}, is $23$ dBm (almost $0.2$W). 

The \textbf{distance dependent pathloss} for cellular communications (\gls{UL} and \gls{DL}) corresponds to the macro-to-\gls{UE} model  for \gls{3GPP} case 1 given in table A.1.1.2-3 of \cite{3GPP_TR36814}. Hence, for a distance $d$ Km between the \gls{BS} and the \gls{UE}, the cellular pathloss is given in dB as follows:
\[
PL_{cell}=128.1+37.6\log10\left (d\right)
\]

In addition, the \textbf{LOS probability} of both \gls{UL} and \gls{DL} communications is the one specified in table  A.1.1.2-3 of \cite{3GPP_TR36814} for the \gls{3GPP} case 1. Thus, for a distance $d$ Km between the \gls{BS} and the \gls{UE}, the LOS probability is given by:
\[
\mathbb{P}_{LOS}=\min\lbrace \frac{0.018}{d},1 \rbrace
\]
\[
 \times \left(1-\exp{\left(-\frac{d}{0.063}\right)} \right)+\exp{\left(-\frac{d}{0.063}\right)}
\]

When the user is indoor, then its \gls{UL} and \gls{DL} communications suffer from a penetration loss of mean $10$dB.

For cellular links, both \gls{UL} and \gls{DL}, we consider large scale \textbf{shadowing} caused by the geographical irregularities of the environment due to the presence of obstacles between the \gls{BS} and the \gls{UE} (e.g. buildings).  In this simulator, shadowing is modeled as a zero-mean Gaussian random variable, with standard deviation $8$ dB.  A shadowing map is constructed for each site in the network with a precision of $5$ Pixels. Both spatial and time correlation exist between the shadowing maps of different sites with a correlation distance of $50$ m and a time correlation of $1$s. The shadowing maps are updated frequently with a periodicity of $1$ s. The details of the shadowing model is described in the appendix A of \cite{abbas01618909}.

For the case of \gls{D2D} links, the distance dependent pathloss, the LOS probability, the penetration loss and the shadowing settings are given in table A.2.1.2 of \cite{3GPP_ProSe} (i.e. for different scenarios outdoor-to-outdoor, outdoor-to-indoor and indoor-to-indoore). The antenna gain for \gls{D2D} communications is considered $0$ dBi (i.e. given in table A.2.1.1-1 of \cite{3GPP_ProSe}).

For cellular communications, the \textbf{fast fading} has a Rayleigh distribution, i.e. square root of the sum of the squares of two normal distribution random variables $F_1\left(t\right)$ and $F_2\left(t\right)$. The fast fading is assumed to vary according to a Gauss-Markov model. If $v$ is the speed traveled by a user then the correlation coefficient $\rho_{f}$ is given by:
\begin{equation}
\label{eq.doppler}
\rho_{f}=\exp{\left( -\frac{2vf_c}{c}\right)}
\end{equation}
where $f_c$ the carrier frequency and $c$ the speed of light. We apply the correlation on the two component of the fast fading as follows:
\[
F_1\left(t\right)=\rho_{f}F_1\left(t-1\right)+\sqrt{\frac{2}{\pi}\left( 1-\rho_{f}^2 \right)}Y_1\left(t\right)
\]
\[
F_2\left(t\right)=\rho_{f}F_2\left(t-1\right)+\sqrt{\frac{2}{\pi}\left( 1-\rho_{f}^2 \right)}Y_2\left(t\right)
\]
where $Y_1\left(t\right)$ and $Y_2\left(t\right)$ are two random variables with normal distribution and the $\sqrt{\frac{2}{\pi}}$ as a correction that is added to guarantee a fast fading of average one.

In addition, $Y_1\left(t\right)$ and $Y_2\left(t\right)$ are generated based on the existing spatial correlation between the fast fading coefficients. We consier the correlation matrix $R_{spat}$ defined in table B.2.2.1 of \cite{3GPP_TR36803} with $\alpha=\beta=0$ (i.e. low spatial correlation). In order to introduce spatial correlation into the channels coefficients according to a specific correlation matrix $R_{spat}$, we consider a vector of \gls{MIMO} uncorrelated channel matrix (i.e.  random variables with normal distribution) and multiply it by the lower triangular matrix corresponding to the Cholesky factorization of the correlation matrix $R_{spat}$.

The same model of fast fading is applied to \gls{D2D} communications. The only difference is that the velocity $v$ in equation \ref{eq.doppler} is modified as specified in section A.2.1.2.1 of \cite{3GPP_ProSe}.

The \textbf{noise} parameters are the following: thermal noise $=-174$ dbm/Hz, the \gls{BS} noise figure $=5$ dB, \gls{UE} noise figure $=9$ dB and \gls{D2D} noise figure $=9$ dB.

On the \gls{UL}, a \textbf{power control} is installed in order to reduce the power emitted by the users and limit by that the interference between neighbor cells. For a pathloss $PL$ between the \gls{UE} and the \gls{BS} (i.e. taking into account the distance dependent pathloss, penetration loss, antenna pattern and shadowing loss), the user transmits at the following power:
\[
P =\min \lbrace P_{max}, P_0+\alpha PL \rbrace \text{ in dBm}
\]
where $P_{max}$ is the maximum user's transmission power, $P_0=-106$ dBm is the target received power at the \gls{BS} and $\alpha=1.0$ is the compensation factor. The same mechanism of Power control is applied for \gls{D2D} communications.

We consider \gls{LTE} codebook-based \textbf{precoding} for enabling transmit diversity mode. We assume that a single layer data stream is transmitted via two antennas. The following transmission schemes are simulated: $2\times 2$ \gls{MIMO} for \gls{DL} and $1\times 2$ \gls{MIMO} for both \gls{UL} and \gls{D2D} communications. The precoding at the transmitter is done by choosing the precoding vector among \ref{table.PrecondingVectors} that maximizes the received power.
\begin{table}[H]
\centering
 \begin{tabular}{|c|c|c|}
\hline 
Codebook index & Rank $1$ & Rank $2$ \tabularnewline
\hline 
0 & $\frac{1}{\sqrt{2}}\begin{bmatrix}
    1       \\
   1
\end{bmatrix}$  & 
$\frac{1}{\sqrt{2}}\begin{bmatrix}
    1  & 0    \\
   0 & 1
\end{bmatrix}$  \tabularnewline
\hline 1 & $\frac{1}{\sqrt{2}}\begin{bmatrix}
    1       \\
   -1
\end{bmatrix}$  & 
$\frac{1}{\sqrt{2}}\begin{bmatrix}
    1  & 1    \\
   1 & -1
\end{bmatrix}$  \tabularnewline
\hline 2 & $\frac{1}{\sqrt{2}}\begin{bmatrix}
    1       \\
   j
\end{bmatrix}$  & 
$\frac{1}{\sqrt{2}}\begin{bmatrix}
    1  & 1    \\
   j & -j
\end{bmatrix}$  \tabularnewline
\hline 3 & $\frac{1}{\sqrt{2}}\begin{bmatrix}
    1       \\
   -j
\end{bmatrix}$  & $-$ \tabularnewline
\hline 
\end{tabular}
\caption{Precoding vectors}
\label{table.PrecondingVectors}
\end{table} 

\subsection*{Link curves}
The \gls{SINR} maps of the \gls{DL} and \gls{UL} channel models considered in the simulator are respectively given in figures \ref{fig.SNR_map_DL} and \ref{fig.SNR_map_UL}. Based on a link-level simulator, we deduce the mapping between the \gls{SINR} and the rate (Kbits/s/Hz) for both \gls{DL} (see figure \ref{fig.SINR_Rate_Mapping_DL}) and \gls{UL} (see figure \ref{fig.SINR_Rate_Mapping_UL}) communications. We deduce the amount of data transmitted at each allocated \gls{RB}. The \gls{UL} link adaptation is applied to \gls{D2D} communications; thus the \gls{UL} \gls{SINR}-rate mapping is considered for \gls{D2D} links.

\begin{figure}[ptb]
\hspace{-20pt}
        \begin{minipage}[b]{0.475\textwidth}
            \centering
            \includegraphics[scale=0.6]{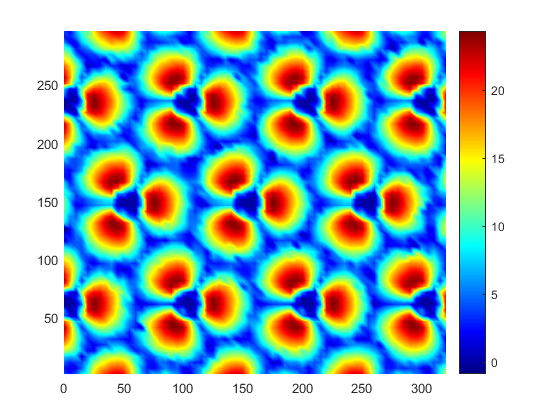}
            \caption[Network2]%
            {{\small \gls{SINR} map for \gls{DL}}}    
            \label{fig.SNR_map_DL}
        \end{minipage}
        \hfill
        \begin{minipage}[b]{0.475\textwidth}  
            \centering 
            \includegraphics[scale=0.45]{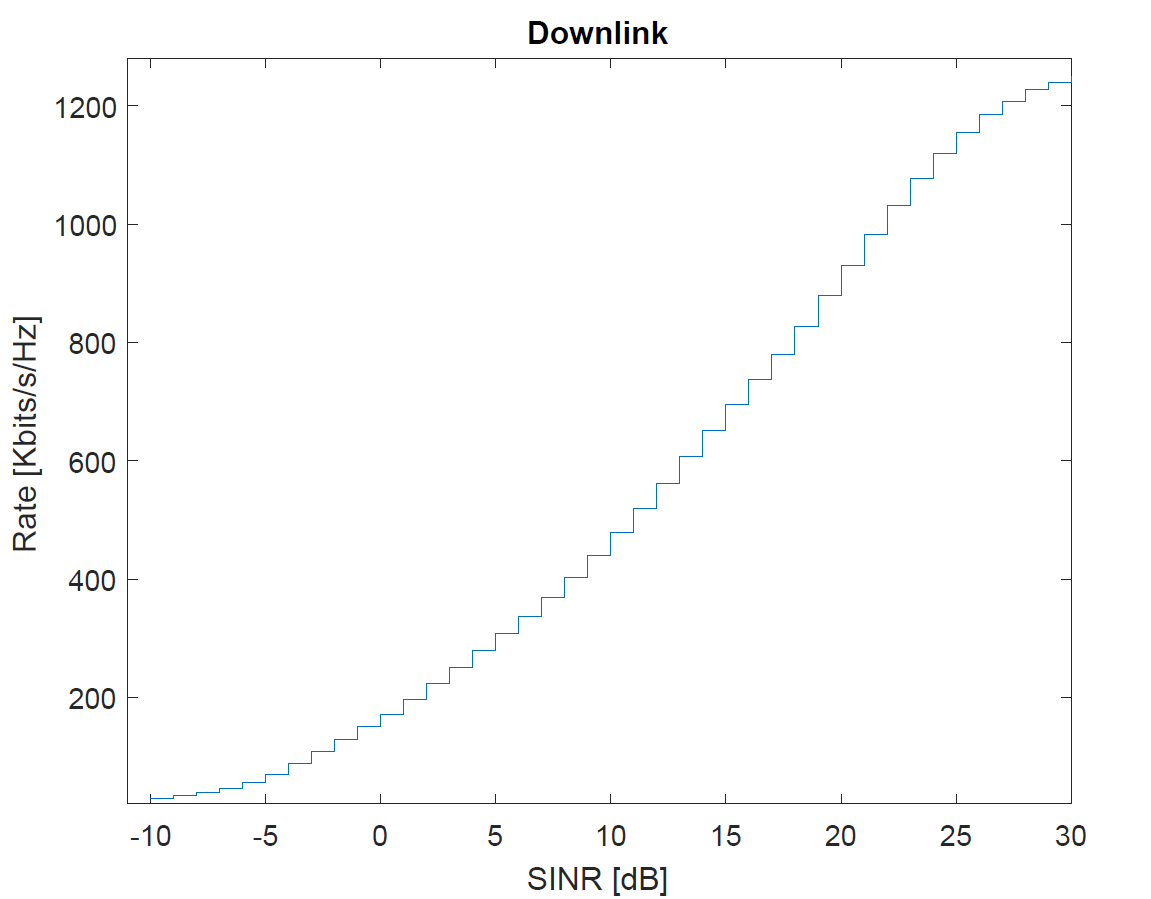}
            \caption[]%
            {{\small \gls{SINR}-rate mapping according to link-level simulations of \gls{DL} communications}}
            \label{fig.SINR_Rate_Mapping_DL}
\end{minipage}
       \label{fig.link_curve_DL} 
    \end{figure}

\begin{figure}
\hspace{-20pt}
        \begin{minipage}[b]{0.475\textwidth}
            \centering
            \includegraphics[scale=0.6]{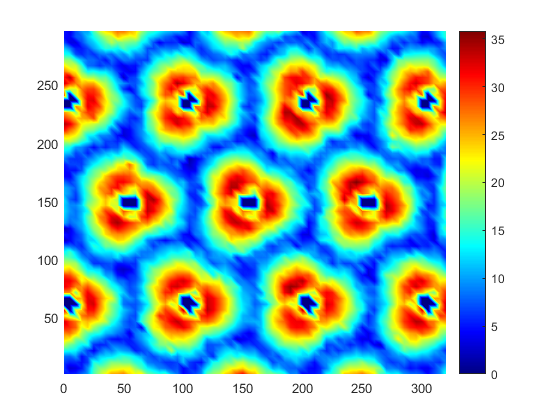}
            \caption[Network2]%
            {{\small \gls{SINR} map for \gls{UL}}}    
            \label{fig.SNR_map_UL}
        \end{minipage}
        \hfill
        \begin{minipage}[b]{0.475\textwidth}  
            \centering 
            \includegraphics[scale=0.45]{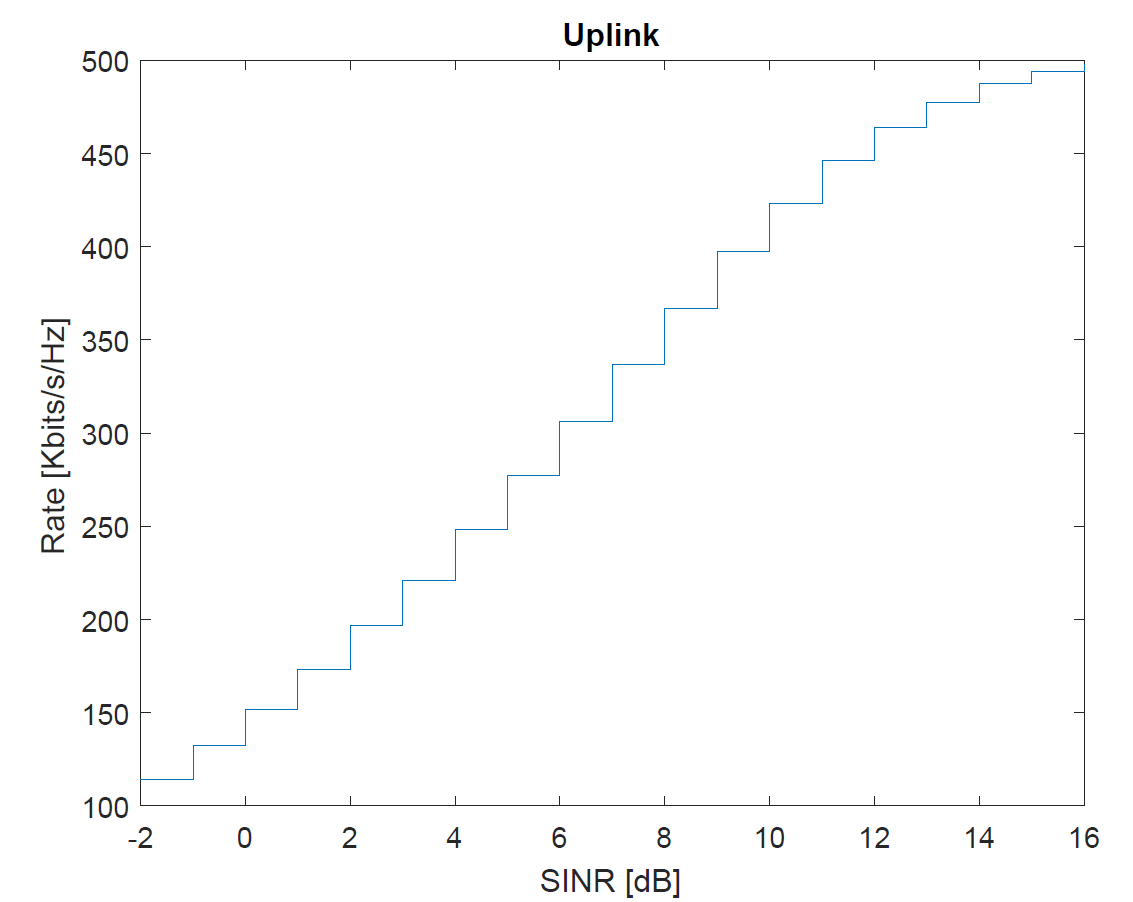}
            \caption[]%
            {{\small \gls{SINR}-rate mapping according to link-level simulations of \gls{UL} communications}}
            \label{fig.SINR_Rate_Mapping_UL}
\end{minipage}
    \end{figure}
    
\subsection*{Traffic Model}
Considering a fix number of users, FTP2 model of traffic given in section A.2.1.3 of \cite{3GPP_TR36814} is assumed (see figure \ref{fig.FTP2}). A \gls{DL} file size of $0.5$ MBytes and an \gls{UL} file size of $0.25$ Mbytes are considered. An exponential reading time $D$ of parameter $0.2$ separates the end of a file download and the arrival of the next file (i.e. the mean of $D$ is $5$ s). 

\begin{figure} [H]
\centering
\captionsetup{justification=centering}
\includegraphics[scale=0.12]{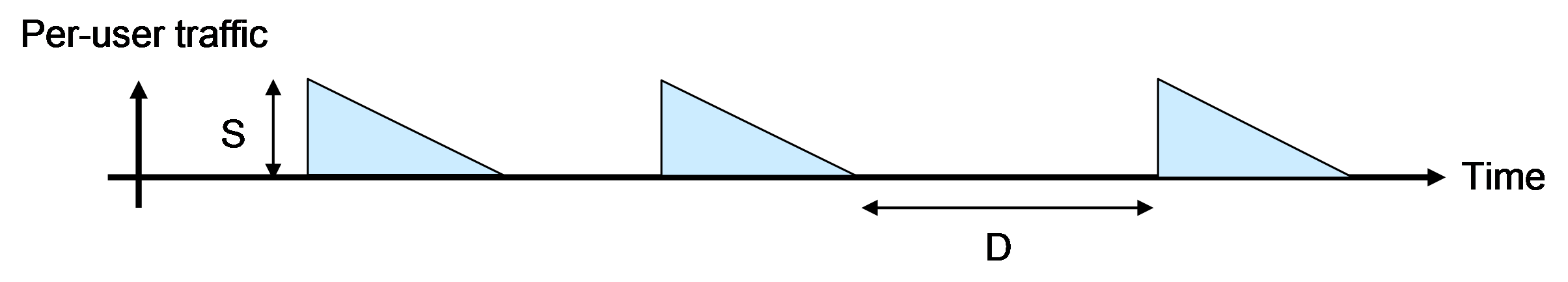}
\caption{Traffic generation of FTP Model 2}
\label{fig.FTP2} 
\end{figure}

\subsection*{Scheduling}
We consider a round-robin scheduler that serves the users in a cycle way and independently of their radio conditions. At a given time-slot, only one \gls{UE} is scheduled (i.e. all the available \gls{RB}s are allocated to the selected \gls{UE}).  Since overlay \gls{D2D} is assumed and the \gls{UL} resources are often less utilized than the \gls{DL} resources, we dedicate $20\%$ of the \gls{UL} resources for \gls{D2D} communications. 

\subsection*{D2D discovery}
The discovery process allow a device to discover nearby users with whom \gls{D2D} communication can be established. We implement the discovery process described in \cite{3GPP_ProSe} and validate it by comparing our results to the one given in \cite{3GPP_ProSe}. The following scenario is considered for evaluating the implementation of the discovery procedure. In each cell, among the $21$ created macro cells, $150$ \gls{UE}s are dropped using one of the two dropping procedures (i.e. options 1 and 2) described in \ref{subsec:UEDrop}.  The \gls{SINR} of \gls{D2D} links is illustrated in figure \ref{fig.SINR_D2D} and is compared to \gls{UL} and \gls{DL} \gls{SINR}.

\begin{figure} 
\centering
\captionsetup{justification=centering}
\includegraphics[scale=0.4]{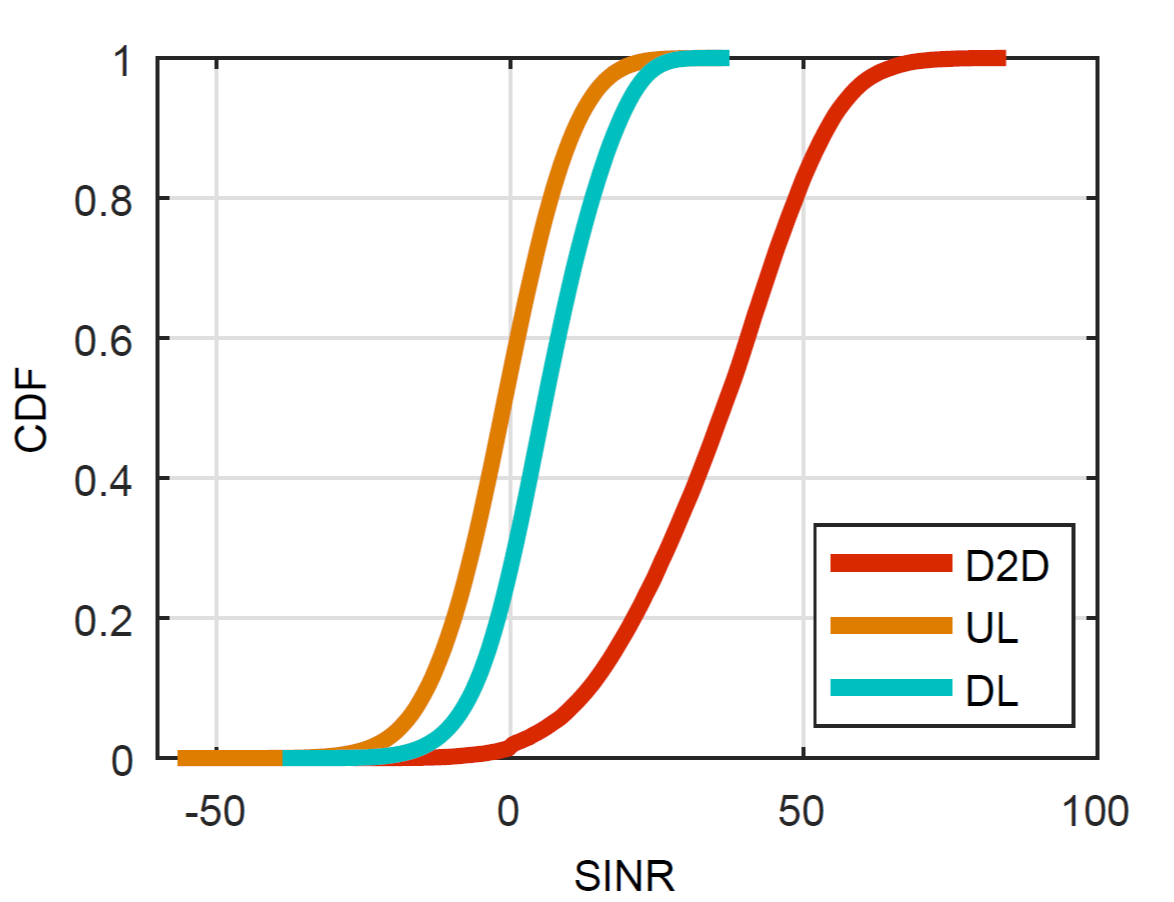}
\caption{\gls{CDF} of the \gls{D2D} \gls{SINR}}
\label{fig.SINR_D2D} 
\end{figure}

For each discovery period of $10$ s, a limited number of subframes (i.e. $30$ \gls{TTI}s) are dedicated for the discovery process. Indeed, one \gls{RB} pair is randomly allocated to each device among all the resources available in the $30$ discovery subframes. Therefore, each user will transmit its discovery signal on its allocated physical \gls{RB}s. When this discovery signal is decoded by a nearby user then the latter considers the device generating the discovery signal as a discovered user (i.e. minimum association received power is $-107$ dBm). Considering half duplex devices, then users transiting at the same \gls{TTI} cannot discover each other. An in-band emissions described in table A.2.1.5-1 of \cite{3GPP_ProSe} is evaluated in each non-allocated \gls{RB} with $ \lbrace 3,6,3,3 \rbrace$ as W,X,Y,Z parameters .

For evaluating the discovery scheme described above, the two following metrics are considered and compared to those given in \cite{3GPP_ProSe}.
\begin{itemize}
\item \gls{CDF} of  the number of discovered \gls{UE}s for different discovery periods (see figure \ref{fig.CDF_Disovery_L1} for layout option 1 and figure \ref{fig.CDF_Disovery_L2} for layout option 2).
\item Number of discovered \gls{UE}s as function of the number of discovery periods (see figure \ref{fig.Mean_Disc_L1}  for layout option 1 and figure \ref{fig.Mean_Disc_L2} for layout option 2).
\end{itemize}

\begin{figure}[ptb]
        \begin{minipage}[b]{0.475\textwidth}
            \centering
            \includegraphics[scale=0.6]{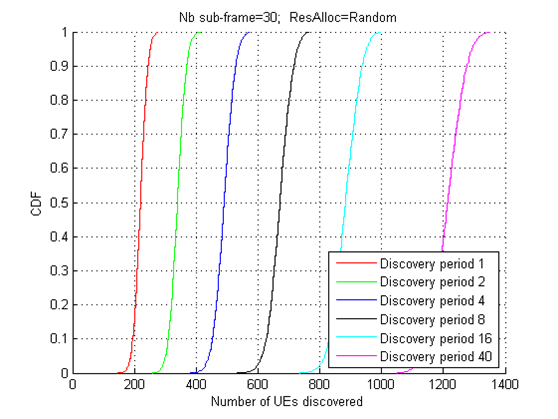}
            \caption[Network2]%
            {{\small \gls{CDF} of  the number of discovered \gls{UE}s for different discovery periods for layout option 1}}    
            \label{fig.CDF_Disovery_L1}
        \end{minipage}
        \hfill
        \begin{minipage}[b]{0.475\textwidth}  
            \centering 
            \includegraphics[scale=0.6]{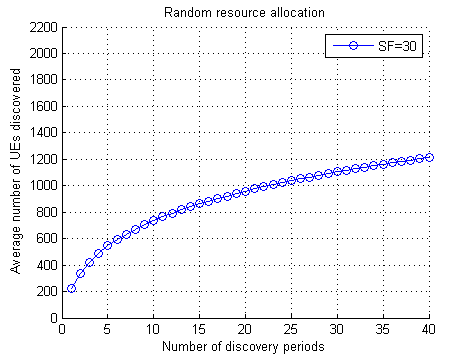}
            \caption[]%
            {{\small Number of discovered \gls{UE}s as function of the number of discovery periods for layout option 1}}
            \label{fig.Mean_Disc_L1}
\end{minipage}
    \end{figure}

\begin{figure}[ptb]
\hspace{-20pt}
        \begin{minipage}[b]{0.475\textwidth}
            \centering
            \includegraphics[scale=0.6]{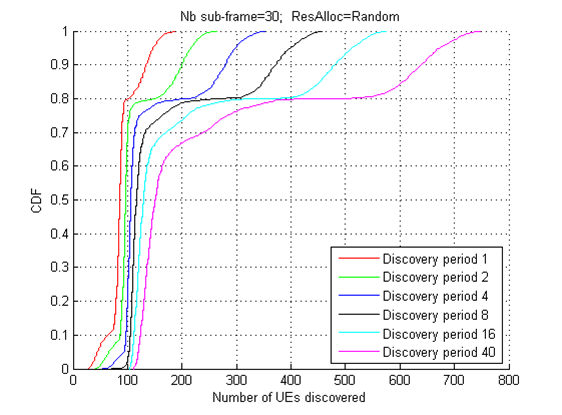}
            \caption[Network2]%
            {{\small \gls{CDF} of  the number of discovered \gls{UE}s for different discovery periods for layout option 2}}    
            \label{fig.CDF_Disovery_L2}
        \end{minipage}
        \hfill
        \begin{minipage}[b]{0.475\textwidth}  
            \centering 
            \includegraphics[scale=0.6]{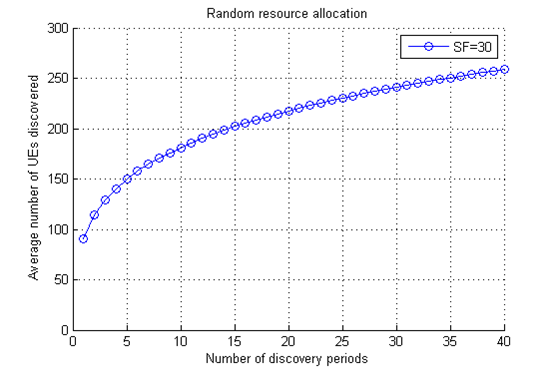}
            \caption[]%
            {{\small Number of discovered \gls{UE}s as function of the number of discovery periods for layout option 2}}
            \label{fig.Mean_Disc_L2}
\end{minipage}
    \end{figure}

\bibliographystyle{IEEEtran}
\bibliography{thesis_bibliography}

The authors would like to thank...

\ifCLASSOPTIONcaptionsoff
  \newpage
\fi



%

%

\begin{IEEEbiographynophoto}{Rita Ibrahim}
received the B.E. degree from Telecom-Paristech, France, and Lebanese University, Lebanon, in 2015, and the M.Sc. degree in advanced communication networks jointly from Ecole Polytechnique and Telecom-Paristech, France, in 2015. She obtained a PhD degree from CentraleSupélec, France, in 2019. She is currently a research engineer in Orange Labs, France. Her current research interests include analytic modeling and performance evaluation of cellular networks, stochastic network optimization, and radio access technologies for 5G networks.
\end{IEEEbiographynophoto}
\begin{IEEEbiographynophoto}{Mohamad Assaad}
received the MSc and PhD degrees (with high honors), both in telecommunications, from Telecom ParisTech, Paris, France, in 2002 and 2006, respectively. Since 2006, he has been with the Telecommunications Department at CentraleSupélec, where he is currently a professor and holds the TCL Chair on 5G. He is also a researcher at the Laboratoire des Signaux et Systèmes (L2S, CNRS). He has co-authored 1 book and more than 100 publications in journals and conference proceedings and serves regularly as TPC member or TPC co-chair for top international conferences. He is an Editor for the IEEE Wireless Communications Letters and the Journal of Communications and Information Networks. He has given in the past successful tutorials on 5G systems at various conferences including IEEE ISWCS'15 and IEEE WCNC'16 conferences. His research interests include 5G networks, MIMO systems, mathematical modelling of communication networks, stochastic network optimization and Machine Learning in wireless networks. 
\end{IEEEbiographynophoto}


\begin{IEEEbiographynophoto}{Berna Sayrac}
 is a senior research expert in Orange Labs. She received the B.S., M.S. and Ph.D. degrees from the Department of Electrical and Electronics Engineering of Middle East Technical University (METU), Turkey, in 1990, 1992 and 1997, respectively. She worked as an Assistant Professor at METU between 2000 and 2001, and as a research scientist at Philips Research France between 2001 and 2002. Since 2002, she is working at Orange Labs. Her current research activities and interests include 5G radio access and spectrum. She has taken responsibilities and been active in several FP7 European projects, such as FARAMIR, UNIVERSELF, and SEMAFOUR. She was the technical coordinator of the Celtic+ European project SHARING on self-organized heterogeneous networks, and also the technical manager of the H2020 FANTASTIC-5G project. Currently, she is the coordinator of the research program on RAN design and critical communications in Orange. She holds several patents and has authored more than fifty peer-reviewed papers in prestigious journals and conferences. She also acts as expert evaluator for research projects, as well as reviewer, TPC member and guest editor for various conferences and journals.

\end{IEEEbiographynophoto}




\end{document}